\documentclass[a4paper, 12pt]{article} 
\usepackage[centertags]{amsmath}
\usepackage{amsmath}
\usepackage[dutch,british]{babel}
\usepackage{amsfonts}
\usepackage{amssymb} 
\usepackage{amsthm}
\usepackage{newlfont}
\usepackage{color}
\usepackage{subfig}
 \usepackage{bbm}
 \usepackage{float}
 \usepackage{import}

\usepackage{stmaryrd}
\usepackage{mathrsfs}
\usepackage{pstricks,pst-node}
\usepackage{palatino}  %palatcm - a missing package
\usepackage{fancyhdr}
\usepackage{graphicx}
\usepackage{fancybox}
\usepackage{multibox}
\usepackage{geometry}
 \usepackage{psfrag}
 
 \usepackage{multirow}

\theoremstyle{plain}
  \newtheorem{theorem}{Theorem}
    \newtheorem{assumption}{Assumption}
  \newtheorem{proposition}{Proposition}[section]
  \newtheorem{lemma}{Lemma}[section]
  
\theoremstyle{definition}
  \newtheorem{definition}{Definition}[section]

\theoremstyle{remark}

  \newtheorem{example}[theorem]{Example}

\numberwithin{equation}{section}

\DeclareMathOperator{\Tr}{Tr}

 \DeclareMathOperator{\supp}{Supp}
\renewcommand{\Re}{\mathrm{Re}\, }
\renewcommand{\Im}{\mathrm{Im}\,}

\newcommand\otimesal{\mathop{\hbox{\raise 1.6 ex
  \hbox{$\scriptscriptstyle\mathrm{al}$}
\kern -0.92 em \hbox{$\otimes$}}}}
\newcommand\oplusal{\mathop{\hbox{\raise 1.6 ex
  \hbox{$\scriptscriptstyle\mathrm{al}$}
\kern -0.92 em \hbox{$\oplus$}}}}
\newcommand\Gammal{\hbox{\raise 1.7 ex
\hbox{$\scriptscriptstyle\mathrm{al}$}\kern -0.50 em $\Gamma$}}
\renewcommand\i{\mathrm{i}}

\let\al=\alpha \let\be=\beta  \let\ep=\epsilon

\let\ve=\varepsilon  \let\ga=\gamma 
\let\ka=\kappa \let\la=\lambda \let\om=\omega 
\let\si=\sigma

 \let\Ga=\Gamma \let\La=\Lambda \let\Om=\Omega
  \let\Si=\Sigma

\newcommand{\caA}{{\mathcal A}}
\newcommand{\caB}{{\mathcal B}}
\newcommand{\caC}{{\mathcal C}}

\newcommand{\caE}{{\mathcal E}}
\newcommand{\caF}{{\mathcal F}}
\newcommand{\caG}{{\mathcal G}}

\newcommand{\caI}{{\mathcal I}}
\newcommand{\caJ}{{\mathcal J}}

\newcommand{\caM}{{\mathcal M}}
\newcommand{\caN}{{\mathcal N}}
\newcommand{\caO}{{\mathcal O}}
\newcommand{\caP}{{\mathcal P}}
\newcommand{\caQ}{{\mathcal Q}}

\newcommand{\caS}{{\mathcal S}}
\newcommand{\caT}{{\mathcal T}}
\newcommand{\caU}{{\mathcal U}}
\newcommand{\caV}{{\mathcal V}}

\newcommand{\scrA}{{\mathscr A}}
\newcommand{\scrB}{{\mathscr B}}

\newcommand{\scrE}{{\mathscr E}}

\newcommand{\scrG}{{\mathscr G}}
\newcommand{\scrH}{{\mathscr H}}

\newcommand{\scrR}{{\mathscr R}}
\newcommand{\scrS}{{\mathscr S}}
\newcommand{\scrT}{{\mathscr T}}

\newcommand{\bbA}{{\mathbb A}}

\newcommand{\bbC}{{\mathbb C}}

\newcommand{\bbE}{{\mathbb E}}

\newcommand{\bbH}{{\mathbb H}}

\newcommand{\bbN}{{\mathbb N}}

\newcommand{\bbP}{{\mathbb P}}

\newcommand{\bbR}{{\mathbb R}}

\newcommand{\bbT}{{\mathbb T}}

\newcommand{\bbX}{{\mathbb X}}

\newcommand{\bbZ}{{\mathbb Z}}

\newcommand{\opunit}{\text{1}\kern-0.22em\text{l}}

\newcommand{\frb}{{\mathfrak b}}

\newcommand{\frh}{{\mathfrak h}}

\newcommand{\frp}{{\mathfrak p}}

\newcommand{\frt}{{\mathfrak t}}

\newcommand{\frA}{{\mathfrak A}}
\newcommand{\frB}{{\mathfrak B}}
\newcommand{\frC}{{\mathfrak C}}

\newcommand{\bsI}{{\boldsymbol I}}

\newcommand{\bsR}{{\boldsymbol R}}
\newcommand{\bsS}{{\boldsymbol S}}

\newcommand{\e}{{\mathrm e}}

\renewcommand{\d}{{\mathrm d}}
\newcommand{\sys}{{\mathrm S}}
\newcommand{\res}{{\mathrm E}}

\renewcommand{\sp}{\sigma}

\newcommand{\Dom}{\mathrm{Dom}}

\newcommand{\beq}{ \begin{equation} }
\newcommand{\eeq}{ \end{equation} }
\newcommand{\bet}{ \begin{theorem} }
\newcommand{\eet}{ \end{theorem} }

\newcommand{\baq}{\begin{eqnarray}}
\newcommand{\eaq}{\end{eqnarray}}
\renewcommand{\supp}{\mathrm{Supp}}

\newcommand{\tor}{ {\bbT^d}  }

\newcommand{\norm}{ \|}
\newcommand{\str}{ |}
\newcommand{\spin}{ \mathrm{spin}}

 \newcounter{smallarabics}
\newenvironment{arabicenumerate}
{\begin{list}{{\normalfont\textrm{\arabic{smallarabics})}}}
  {\usecounter{smallarabics}\setlength{\itemindent}{0cm}
  \setlength{\leftmargin}{5ex}\setlength{\labelwidth}{4ex}
  \setlength{\topsep}{0.75\parsep}\setlength{\partopsep}{0ex}
   \setlength{\itemsep}{0ex}}}
{\end{list}}

\newcounter{smallroman}

\newcommand{\ben}{\begin{arabicenumerate}}
\newcommand{\een}{\end{arabicenumerate}}

\newcommand{\initialresfinite}{\rho_\res^{\mathrm{ref}}}
\newcommand{\initialsysfinite}{\rho_\sys^{\mathrm{ref}}}
\newcommand{\initialfinite}{\rho_{\sys\res}^{\mathrm{ref}}}
\newcommand{\initialgibbsfinite}{\rho_{\sys\res}^{\be}}

\newcommand{\initialnu}{\nu^{\mathrm{ref}} }
\newcommand{\initialnugibbs}{\nu^{\be}}
\newcommand{\initialnun}{\nu^{\mathrm{ref}}_n}
\newcommand{\initialnugibbsn}{\nu^{\be}_n}
\newcommand{\reff}{\mathrm{ref}}

\newcommand{\links}{L}
\newcommand{\rechts}{R}
\newcommand{\adjoint}{\mathrm{ad}}
\newcommand{\ad}{\adjoint}
\newcommand{\Adjoint}{\mathrm{Ad}}

\newcommand{\kin}{\mathrm{kin}}

\newcommand{\lone}{\mathbbm{1}}

\newcommand{\uv}{\underline{v}}
\newcommand{\uu}{\underline{u}}
\newcommand{\uw}{\underline{w}}
\newcommand{\ut}{\underline{t}}
\newcommand{\dist}{\mathrm{dist}}

\newcommand{\weird}{}
\newcommand{\poly}{\frB}

\newcommand{\largegap}{\frb}%{{b}_{\mathrm{L} }}
\newcommand{\bana}{\frp}%{{b}_{\mathrm{a} }} 

\newcommand{\largegapn}{\largegap}

\newcommand{\normw}{\norm^{}_{\ga}}
\newcommand{{\banone}}{\scrG} 
\newcommand{{\gamzero}}{\ga_0} 
\newcommand{{\tengam}}{10\gamzero} 
\newcommand{{\fifteengam}}{15\gamzero} 
\newcommand{{\twentygam}}{20\gamzero} 
\newcommand{{\normba}}{\norm^{}_{\banone}} 
\newcommand{{\normdia}}{\norm^{}_{\diamond}} 
\newcommand{\lin}{\mathrm{lin}} 
\newcommand{\nlin}{\mathrm{nlin}}

\newcommand{\offdx}{v}
\newcommand{\offsetx}{\eta}

\newcommand{\ini}{\mathrm{I}}
\newcommand{\indicator}{1}
\newcommand{\tr}{\mathrm{tr}}
\newcommand{\lamb}{\mathrm{Lamb}}

\newcommand{\hf}{{_1\over^2}}

\newcommand{\smallspace}{{\scrS}}
\begin{document}
\begin{center}
\large{ \bf{Diffusion for a quantum particle coupled to phonons in $d\geq 3$.} } \\
\vspace{15pt} \normalsize

{\bf   W.  De Roeck\footnote{
email: {\tt
 w.deroeck@thphys.uni-heidelberg.de}}  }\\
\vspace{10pt} 
{\it   Institut f\"ur Theoretische Physik  \\ Universit\"at Heidelberg \\
Philosophenweg 16,  \\
D69120 Heidelberg,  Germany 
} \\

\vspace{15pt}

{\bf   A. Kupiainen\footnote{
email: {\tt    antti.kupiainen@helsinki.fi  }}  }\\
\vspace{10pt} 
{\it   Department of Mathematics \\
 University of Helsinki \\ 
P.O. Box 68, FIN-00014,  Finland 
} \\

\end{center}

\vspace{20pt} \footnotesize \noindent {\bf Abstract: }  We prove diffusion for a quantum particle coupled to a field of bosons (phonons or photons).  The importance of this result lies in the fact that  our model is fully Hamiltonian and randomness enters only via the initial (thermal) state of the bosons. This model is closely related to the one considered in \cite{deroeckfrohlichhighdimension} but  various restrictive assumptions of the latter have been eliminated.  In particular, depending on the dispersion relation of the bosons,  the present result holds in dimension $d \geq 3$ and no severe infrared conditions on the coupling are necessary.

\section{Introduction}

The rigorous derivation of long-time diffusion from first principles of mechanics, be it quantum or classical, remains an inspiring challenge in mathematical physics.  To our best knowledge, there are up to this date very few results of this type, see  Section \ref{sec: related work} for a brief review. Recently, in \cite{deroeckfrohlichhighdimension}, a model was introduced  which  is quite tractable and for which diffusion was proven in dimension $d \geq 4$.    It is a quantum system described by a Hamiltonian of the type
\beq
H= H_\sys+ H_{\res} +  \la H_{I}, \qquad  \la \in \bbR   \label{eq: first def ham}
\eeq
where $H_{\sys}$ is the Hamiltonian of a free particle moving on the lattice and it consists of two parts $H_\sys= H_\kin+H_{\spin}$ describing the translational (kinetic) degrees of freedom and a spin degree of freedom, respectively.
The Hamiltonian $H_\res$ describes a free field of bosons (the environment), and  $H_{I}$ effectuates the coupling between both. The system is started with the environment in a thermal state at inverse temperature $\be$ or in a nonequilibrium state (then we have two phonon fields, at different inverse temperatures $\be_1 \neq \be_2$). Such models are a paradigm of open quantum systems.   The form of the Hamiltonian will be given in Section \ref{sec: model}, but let us already list the properties that allow us to handle this model: 
\begin{itemize}
\item The mass of the particle, or, since we are on a lattice, rather the inverse hopping strength, is chosen large.  In \eqref{eq: first def ham} this is accomplished by choosing $H_\kin$ in $H_\sys$
small. 
 This allows a better control of a diagrammatic expansion in real space, since the particle needs a long time to explore a large volume on the lattice.
\item  Even though the mass is large, the 'mixing rate' that the spin and momentum degrees of freedom  of the particle experience due to the interaction with the phonons is not small.   This is possible because of the inclusion of the spin-degree of freedom. 
\item By choosing the interaction Hamiltonian sufficiently smooth (in the momentum of the phonons) and the dimension $d$ sufficiently large, we ensure that the free space-time correlation functions of the boson field decay at an integrable rate in time. We need them to decay at least as $\caO(t^{-(1+ \al)})$ for large times $t$ with $\al>1/4$.  To engineer this in $d \geq 3$, it suffices to choose the dispersion relation of the bosons to be quadratic in the momentum for small momenta, and to cut off the interaction for large momenta.
\item  By choosing the coupling constant $\la$ small, we have a well controlled Markovian approximation (Lindblad equation) that describes the particle for times of $\caO(\la^{-2})$. This Markovian approximation serves as a first approximation to the true behavior and we set up an expansion to control the deviations from it.
\end{itemize} 

Under these assumptions we prove that the reduced dynamics of the particle is diffusive.
Our proof is based on a renormalization group (RG) method that was developed in \cite{bricmontkupiainenrwre, bricmontkupiainenexponentialdecay, ajankideroeckkupiainen} to prove diffusion
for random walk in a random environment (RWRE). In the present context the random environment is 
provided by the phonon field. Unlike in the case of RWRE, in the case at hand the particle influences the environment
and the reduced dynamics is non-Markovian. However, the Markovian approximation mentioned above
provides a starting point for the analysis where a Markovian dynamics is perturbed by a small non-Markovian noise. In units of the weak coupling time scale $\caO(\la^{-2})$ our model
can then be viewed as a (quantum) random walk in a (quantum)  random environment.
The RG method consists of an iterative scheme to show that on successive larger temporal
and spatial scales the random environment becomes smaller and smaller and 
the dynamics tends to a renormalized Markovian "fixed point". We show that the renormalized noise
vanishes in this limit by showing that its (quantum) correlation functions tend to zero. Here
we use a formalism developed earlier by us \cite{deroeckkupiainen} for the confined case, i.e. the proof
that the state of a confined quantum system interacting with a similar field as here tends to the equilibrium state. 

The difference of the model considered in this paper and the one treated in
 \cite{deroeckfrohlichhighdimension} is that in the latter case an additional condition was imposed on the free boson correlation function that restricts the model to dimensions  $d \geq 4$ and to a rather special class of analytic  particle-phonon interaction terms. In the context of these models where the particle 
 mass is chosen to scale as $\caO(\la^{-2})$ it still remains a challenge to treat more generic
 phonon or photon reservoirs where the temporal correlations decay as $\caO(t^{-1})$ (which is the case in $d=3$ if the dispersion relation is linear for small momenta). To deal
 with these cases with our method one needs a more careful RG analysis. A much more
 difficult and interesting problem is to relax the large mass assumption. In this case the
 control of the corrections to the Markovian approximation seems still beyond current
 techniques.

\vskip 3mm

\noindent {\bf Acknowledgements.} We thank for European Research Council and 
Academy of Finland for financial support.

\section{The model}\label{sec: model}
 We consider a finite, discrete hypercube $\La= \La_L  = \bbZ^d / L \bbZ^d$ with $L \in 2\bbN$. Whenever necessary, we identify  $\La_L$ with  the subset of $   \bbZ^d$ given by $\{-L/2+1,\ldots, L/2\}^d$.  We will take the thermodynamic limit at the end by letting $L \to\infty$.  Since most concepts in the present section depend on the volume $\La$, we often do not indicate it explicitly. 

\subsection{Dynamics}

\subsubsection{Particle}

The Hilbert space of the particle is 
\beq
\scrH_\sys = \smallspace  \otimes l^2(\La) 
\eeq
where $\smallspace$ is the finite dimensional space of internal degrees of freedom.
The Hamiltonian is 
\beq
H_{\sys} =  H_{\spin}\otimes  \lone +  \lone \otimes H_\kin, \qquad   H_\kin = -\frac{1}{ \la^{-2}{m_{\mathrm{p}}}} \Delta
\eeq
where $\Delta$  is the discrete Laplacian on $l^2(\La)$ with periodic boundary conditions, defined by the kernel $\Delta(x,y)=\delta_{|x -y|_1,1}-2d\delta_{x,y}$  where
  $\str x-y \str_1 = \sum_{i=1}^d \str x_i -y_i \str$ with $\str \cdot \str$ the distance on the discrete torus $\bbZ/L\bbZ$. 
The parameter  $\la^{-2} {m_{\mathrm{p}}}$ is the mass of the particle, where the factor $\la^{-2} $ is put in to stress that we choose the mass very big by making $\la$ small. 
For simplicity, we assume that the Hamiltonian $H_{\spin}$ is nondegenerate such that we can label its eigenvectors by their eigenvalues, which we denote by $e \in \sp(H_{\spin})$.  

\subsubsection{Phonon field}
A single phonon is described by the one-particle Hilbert space $l^{2}(\La)$. It will however be convenient to consider this space in the Fourier-representation,  $l^2(\La^*)$ where $\La^* =    \tfrac{2\pi}{L} (\bbZ^d/ L \bbZ^d)$  is the dual lattice. We consider $\La^*$ as a subset of $\bbT^d$-the $d$-dimensional torus identified with $[-\pi,\pi]^d$. 
 The Hilbert space of the phonon field is
 \beq
 \scrH_\res =    \Ga(l^2(\La^*))%\sim  \scrH_\res
 \eeq
where $\Ga(\frh)$ is the  symmetric (bosonic) Fock space built on the one-particle Hilbert space $\frh$.

The Hamiltonian of the phonon field is given by
\beq
H_\res =  \sum_{q \in \La^*}   \omega(q) a^*_{ q} a_{ q} 
\eeq
where $a^*_{q}/a_{q} $ are the  creation/annihilation operators satisfying the canonical commutation relations (CCR)
\beq
[a_{q}, a^*_{q'}] =  \delta_{q,q'}.
\eeq
$ \omega(q)\geq 0 $ is the frequency of the phonon with momentum $q$ and
we take $\omega$ a smooth  function defined on $ \bbT^d$.
We impose later further conditions on $\omega$.

We will also consider a non-equilibrium setup, in which case we consider two different phonon fields, distinguished by the label $j=1,2$.  In this case
 the Hilbert space  is 
 $
 \scrH_\res =    \Ga(l^2(\La^*)) \otimes   \Ga(l^2(\La^*)) %\sim  \scrH_\res
 $
and the Hamiltonian 
\beq
H_\res =  \sum_{q \in \La^*}   \left( \omega_1(q) a^*_{q,1} a_{ q,1} +   \omega_2(q) a^*_{q,2} a_{q, 2} \right)
\eeq
where $a^*_{q,1}/ a_{ q,1} $ act on the first tensor in $\scrH_\res$ and $a^*_{q,2}/ a_{ q,2} $ on the second.

We assume the reader is familiar with these basic notions of second quantization and we refer to \cite{derezinski1} for more detailed accounts. 
\subsubsection{Full Hamiltonian} \label{sec: full ham}

The interaction between particle and phonon field is chosen linear in the creation/annihilation operators. It is of the form
\beq
H_{I} = (2\pi/L)^{d/2}   \sum_{q \in \La^*, j=1,2}     \left(   W \otimes \e^{\i q X} \otimes   \phi_j(q)  a_{q,j}  \right)   + \text{h.c.}
\eeq
where `h.c.' stands for `hermitian conjugate',  $W$  is a Hermitian matrix acting on $\smallspace$ (the space of the internal degree of freedom), $X=(X_j), j=1,\ldots, d$ is the vector of multiplication operators with the variable $x \in \bbZ^d$ on $l^2(\bbZ^d)$, $q X$ is shorthand for $\sum_{j=1}^d q_j X_j$,
and $\phi_j (q)$ are "structure factors" that describe the coupling to the  phonons.  We will  take   $\phi_j (q)$ the values of smooth functions $\phi_j : \bbT^d\to \bbC$
at $q\in \La^*$. We will   impose some
further conditions on $\phi_j $ later.

The fact  that the particle couples to both phonon fields via the same $W$-matrix, is by no means important, and we make it for notational simplicity. 

The total Hamiltonian is then 
\beq
H= H_\sys + H_\res +  \la H_{I}
\eeq
acting on
$ \scrH =\scrH_\sys  \otimes \scrH_\res$. 
If $\om_j(q) > 0$ for any $q \in \La^{*}$,  then $H_{I}$ is an infinitesimal perturbation (in the sense of Banach operator theory) of $H_\res$. By a standard application of the Kato-Rellich theorem, $ H$ is self-adjoint on the domain of $H_\res$ (we need not worry about $H_\sys$ since it is  bounded). 
The condition $\om_j(q) > 0$ is also necessary to have a well-defined finite-volume Gibbs state (see e.g.\ \eqref{eq: microscopic expression correlation function}). However, once we take the thermodynamic limit, the only remaining regularity assumption will be Assumption \ref{ass: decay micro alpha}. More concretely, this means that if we wish to consider a smooth function $\om_j$ that vanishes in some points, then we should modify it in these points such that it is discontinuous and strictly positive.  The modification will be invisible in the thermodynamic limit).

\subsection{States}\label{sec: initial state}

States of the combined system and environment are given by density operators $\rho_{\sys\res}\in\scrB_1(\scrH_\sys\otimes\scrH_\res)$ where $\scrB_1$ denotes trace class operators.
The initial states that we will consider are  of the following form
\beq
\rho_{\sys\res}=\rho_\sys \otimes \initialresfinite
\eeq
with
\beq
 \initialresfinite = \frac{1}{\caN}\e^{-\be_1 H_{\res_1}} \otimes \e^{-\be_2  H_{\res_2}}, \qquad  \caN =\Tr \left( \e^{-\be_1 H_{\res_1}} \otimes \e^{-\be_2  H_{\res_2}} \right)
\eeq
and $\rho_\sys \in \scrB_1(\scrH_\sys)$ is chosen to have support concentrated around the origin in the following sense. Elements 
of $\scrH_\sys$ can be represented by functions $\psi(x)$  with $x \in \La \subset \bbZ^d$ and taking values in
the spin space 
$ \smallspace $. 
In this basis $\rho_\sys$ is given by a kernel  (matrix) $\rho_\sys(x',x)$ taking values in  $\scrB(\smallspace )$.
% with $x,x' \in \bbZ^d$ and $e,e' \in {\sp}(H_{\spin})$. 
 Then we require
\beq
\rho_\sys(x',x)= 0, \qquad \text{for} \,\, |x|,|x'|>R 
\label{initialS}
\eeq
for some $R <\infty$. 

Since the environment Hamiltonians are quadratic in the creation/annihilation operators, the density matrix $\initialresfinite$ is a ``Gaussian state'', sometimes also referred to as a `quasifree state''. Moreover, since these density matrices are functions of the environment Hamiltonian, the initial environment  state $\initialresfinite$  is obviously invariant under the free environment evolution:
\beq
  \e^{-\i t H_\res} \initialresfinite  \e^{\i t H_\res} = \initialresfinite 
\eeq
In fact, these two properties, Gaussianity and invariance under the free dynamics, are what we will really use in our analysis, and the specific choice that we made for $\initialresfinite$ is not important, except for one of our results which will additionally require $\be_1=\be_2$ and where we exploit the fact that the system is in (close to) thermodynamic equilibrium.
The dynamics of the density matrix of the entire system is given by 
\beq
\rho_{{\sys\res},t}= \e^{-\i t H} \left(  \rho_{\sys} \otimes \initialresfinite \right) \e^{\i t H}=\e^{-\i t L}( \rho_{\sys} \otimes \initialresfinite )
\label{fulldyn}
\eeq
where we denote $L:={\ad}(H)$.

The reduced dynamics of the particle is defined by taking a partial trace $ \Tr_\res$ over the 
environment Hilbert space $\scrH_\res$:
\beq
\rho_{\sys ,t} = \Tr_\res  \rho_{{\sys\res} ,t}= \Tr_\res( e^{-\i  tL}( \rho_{\sys} \otimes \initialresfinite )
) :=Z_t\rho_{\sys}
\label{Sdyn}
\eeq
All our results will concern the reduced density matrix $\rho_{\sys, t}$, which means that we only consider particle observables, but it is straightforward to extend the formalism so as to handle observables of the boson field, as well as more general initial states.
\subsection{Thermodynamic limit}\label{sec: thermodynamic limit in intro}

Since we are ultimately interested in long-time properties, we have to perform the thermodynamic limit to eliminate Poincar{\'e} recurrences.  We take care of this below, but first we introduce the free boson correlation function which will play an important role in our our main assumptions.  First, we define the so-called Segal field operators
\beq\label{Segal}
\Phi(x,t) :=    (2\pi/L)^{d/2}   \sum_{q \in \La^*,\, j =\{1,2\}}   \left(\e^{\i (q x + \om_j(q)t)} \phi_j(q) a^*_{q,j}  + \e^{-\i (q x + \om_j(q)t)} \overline{\phi_j(q)} a_{q,j} \right)
\eeq
The correlation function is then defined as
\baq
\zeta(x,t)& := &  \Tr \left[ \initialresfinite   \Phi(x,t) \Phi(0,0) \right]   \nonumber \\[1mm]
& = &  (2\pi/L)^{d}
\sum_{q \in \La^*, j =\{1,2\}}   \str \phi_j(q)\str^2  \left( \frac{1}{\e^{\be \om_j(q)}-1} \e^{\i( x q+ t \om_j(q))}    +   \frac{1}{1-\e^{-\be \om_j(q)}} \e^{-\i ( x q+t \om_j(q)) }  \right)  \label{eq: microscopic expression correlation function}
\eaq
where the second equality is again a standard exercise in second quantization, and one recognizes the Bose-Einstein distribution $1/(\e^{\be \om_j}-1)$.
Let us denote by $\zeta_j(x,t)$ the corresponding expressions where $j$ in the last sum is fixed, such that $\sum_{j=1,2}\zeta_j(x,t)=\zeta(x,t)$. It is clear that $\zeta_j$ is the correlation function due to reservoir $j$.
  The correlation function $\zeta(x,t)$  will naturally come up in the evaluation of the perturbation series for the dynamics. \\

To describe the thermodynamic limit, let us indicate the dependence in (\ref{eq: microscopic expression correlation function}) on the volume $\La$ explicitly by a superscript $\zeta^{\La}(x,t), \zeta^{\La}_j(x,t)$.  We write $\lim_{\La \nearrow \bbZ^d}$ as a shorthand for the limit $\lim_{L \to \infty}$, since $\La=\La_L$.
We assume that  the finite-volume free correlation functions converge
\beq
\lim_{\La \nearrow \bbZ^d}\zeta_j^{\La}(x,t)  =  \zeta(x,t) \label{eq: assumption free correlation functions thermo}
\eeq
for any $x$ and uniformly on compacts in $t \in \bbH_{\be_j} = \{ z \in \bbC, 0\leq \Im z \leq \be_j \}$, and that the limiting correlation functions $ \zeta_j(x,t)$ are bounded and continuous.  
This will hold for the basic examples of optical and acoustic phonons in
 Section \ref{sec: results} (in the latter case we need to exclude $q=0$ 
 in the sum in eq. (\ref{eq: microscopic expression correlation function}), cfr.\ the discussion in Section \ref{sec: full ham}).
 The reason that it is natural to consider times in the strip $\bbH_{\be_j}$ has to do with the KMS-relation, which is the mathematical translation of the detailed balance property for reservoirs in thermal equilibrium. We will not need any of this theory here, but we refer the interested reader to \cite{bratellirobinson}.
 
 Note also that the spaces $\scrH_\sys$ for $\La$ finite are naturally embedded into  $\scrH_\sys$ for $\La = \bbZ^d$  ($L=\infty$) by our identification of $\La,\La^*$ with subsets of $\bbZ^d, \tor$. In particular, a density matrix $\rho^{\La}_\sys$ satisfying \eqref{initialS}, is embedded into $\scrB_1(\scrH_\sys)$ for $L > R$.
More generally, we have

\begin{lemma} \label{lem: thermo limit}
If  the initial state is chosen as described above (including in particular  the convergence \eqref{eq: assumption free correlation functions thermo} and the boundedness of $\zeta(x,t)$) then the limit
\beq
\rho_{\sys,t}   := \lim_{\La \nearrow \bbZ^d} \rho^{\La}_{\sys,t}  
\eeq
 exists in $\scrB_1(\scrH_\sys)$ with $\scrH_\sys$ defined with $\La= \bbZ^d$.
\end{lemma}
Note that this implies that $\rho_{\sys,t}$ is a density matrix, i.e.\ $\rho_{\sys,t}\geq 0$ and $\Tr \rho_{\sys,t}=1$.

\subsection{Results} \label{sec: results}

In this section, all quantities refer to infinite volume ($\La=\bbZ^d$), unless we explicitly state the opposite, but this happens only in Section \ref{sec: thermalization}.

We need an assumption that will guarantee sufficiently fast decay of correlations of the freely evolving phonon field.

 \renewcommand{\theassumption}{\Alph{assumption}}

\begin{assumption} [Decay($\alpha$)] \label{ass: decay micro alpha}
There is an $\al >0$
such that, for $j=1,2$
\beq
 \int_{\bbR^+} \d t  \,  (1+\str t \str)^{ \al}   \sup_{x \in \bbZ^d, {0\leq u \leq \be_j}}     \str  \zeta_j(x,t +\i u) \str <\infty  %  \leq   C_{\zeta}
\eeq
With no loss, we suppose $\al < 1$.
\end{assumption}
The supremum  over $u$ on the RHS is irrelevant for most parts of the argument; it only plays a role for our equilibrium results i.e.\ with $\be=1=\be_2$. In fact, for natural form factors $\phi$, the decay property for $0< u \leq \be$ can often be obtained from the decay for $u=0$.   In what follows, we will always refer to the total correlation function $\zeta=\sum_j\zeta_j$. 

The next assumption eliminates a drift by requiring that the model be reflection-symmetric and rotation-symmetric (as far as the lattice allows). 
Let $O_{\bbZ^d}$ be the subgroup of the orthogonal group $O(d)$ that fixes the lattice, i.e.\ $
O \in O_{\bbZ^d}$ iff.\  $u \in \bbZ^d \Rightarrow Ou \in \bbZ^d$ and $\str O u \str = \str  u \str$. 

\begin{assumption} [Symmetries] \label{ass: symmetries}
The correlation function $\zeta$ is $O_{\bbZ^d}$-invariant in the sense that 
\beq
\zeta(Ox,t)=  \zeta(x,t), \qquad \text{for any}\,\, O \in O_{\bbZ^d}.
\eeq
\end{assumption}
\noindent 
Of course, this assumption is satisfied by choosing the parameters of the Hamiltonians $H_{\res}, H_I$ and the initial state $\rho_\res^{\reff}$ symmetric.
The  $O_{\bbZ^d}$-invariance of the particle dynamics was already assured by our explicit choice of $H_{\kin}$. 
%This assumption ensures that the initial environment state is $O_{\bbZ^d}$-invariant. Of course, the dynamics should be invariant as well, but that is already assured by our choice of the Hamiltonian $H$.
%Denote by  $V_O$ be the implementation of $O \in O_{\bbZ^d}$  on $\scrH$. 
%The Hamiltonian is invariant under transformations $O$
%\beq
%V^{}_O  H  V^{-1}_{O}=   H
%\ee

The next assumption assures that the particle is sufficiently well-coupled to the phonon field. 
To state it, we need  some definitions.   

First, we assume that the spectrum of $H_{\spin}$ is non-degenerate. Then we can choose 
 a  basis $\psi_e \in \smallspace $ of normalized eigenvectors of $H_\spin$, labelled by $e\in{\sp}(H_{\spin})$,
 and we define $W_{e,e'} := \langle \psi_{e}, W \psi_{e'} \rangle_{\smallspace }$ (in general, we write $\langle \cdot, \cdot\rangle_{\scrE}$ for the scalar product in a Hilbert space $\scrE$).

Secondly, let us introduce  the set of Bohr frequencies corresponding to the spin Hamiltonian $H_{\spin}$,
\beq\label{bohr frequencies}
\caE :=  {\sp}(  {\adjoint}(H_{\spin})   )=\{ \varepsilon= e-e' ;  e,e'  \in {\sp}(H_{\spin}) \}
\eeq
We assume the spectrum of ${\adjoint}(H_{\spin})$  is also non-degenerate except for the eigenvalue $0$ which is $\dim \smallspace$-fold degenerate. That is,  for any $\ve \in \caE, \ve \neq 0$, there is a unique ordered pair $(e,e'), e,e' \in \sp(H_{\spin})$ such that $e'-e = \ve$. 
Next, let
\beq
\zeta_\omega(x)=  \int_{\bbR} \d t \, e^{-\i\omega t}\zeta(x,t)
\label{ftxi}
\eeq
be the Fourier transform of the correlation function in time, which is well-defined by Assumption \ref{ass: decay micro alpha}. Without further comment, we will always choose the definition of the Fourier transform in such a way as to minimize the number of factors $2\pi$ in our formulas. We require that 
\beq
\lim_{|x|\to \infty}\zeta_\ve(x)=0 
\label{zetadecay}
\eeq
for all $0\neq\ve\in\caE$. 
We consider its  distributional Fourier transform $\hat \zeta_\omega$, defined by  $ \zeta_\om(x)=\int \d q \, \e^{-\i qx } \hat \zeta_\omega(q)$. From (the infinite volume limit of)  (\ref{eq: microscopic expression correlation function}),  it is formally given by
\beq
\hat\zeta_\omega(q)= (-1)^{\mathrm{sgn}(\omega)}  2 \pi \sum_{j=1,2}      \str \phi_j(q) \str^2
  \frac{1}{\e^{\be_j \omega}-1}
 \delta(\omega_j(q)-|\omega|) .
\label{xihat}
\eeq
We assume that for all  $0\neq \ve\in \caE$ 
 the distribution $\hat\zeta_\ve$ defines a finite positive Borel measure
on $\tor$ which we denote by $\hat\zeta_\ve(\d q)$. It is readily seen that
this assumption holds e.g.\ if we suppose  that $ \nabla \om_j$ is bounded away from $0$ in  a neighborhood of   the set  $\caM_{\ve,j}=\{q \in \bbT^d \ |\ \omega_j(q)=|\ve|\}$ for $0\neq \ve\in \caE$. For $\omega=0$, we
assume $\hat\zeta_\omega=0$. These conditions hold 
in the examples below. 

 We will now define a Markov jump process with state space $\caF:=  \sp (H_\spin) \times \tor$.
The rates $j(e', \d k' ; e,k)$ of the jump process (jumps are from unprimed to primed variables) are translation invariant in $k$-space: 
 $j(e', \d k' ; e,k)=j(e', \d (k'-k) ; e,0)$ and
 given by
 \beq
 j(e', \d k' ; e,0)
=     \str W_{e,e'}\str^2   \hat\zeta_{e'-e}(\d k') 
\label{rates}
\eeq
The main assumption then is that this Markov process acts irreducibly on absolutely continuous densities: The transition map $P_t$ acts on $L^1(\caF)=L^1(\caF, \mu)$, where $\mu$ is the product of Lesbegue measure on $\tor$ and the counting measure $\sp (H_\spin)$ and we demand that, for any $t>0$, the only Borel sets $E$ for which the implication
\beq
\supp f \subset E \quad \Rightarrow \quad \supp P_tf \subset E, \qquad  f \in L^1(\caF) 
\eeq 
holds, are, up to null sets, $\emptyset$ or $\caF$. \newline

 We summarize:

\begin{assumption} [Fermi Golden Rule] \label{ass: fermi golden rule} Assume
\begin{itemize}\item[i)] The spectra of $H_\spin$ and  ${\mathrm{ ad}}(H_\spin)$  (apart from the eigenvalue $0$)
are non-degenerate 
\item[ii)]  $\zeta_0\equiv 0$ and for all $0\neq\ve\in \sp(\mathrm {ad}(H_{\spin}))$,  $\zeta_\ve(x)$
decays at infinity and
 $\hat\zeta_\ve$ defines a finite positive measure. 
\item[iii)]  The Markov process on $\caF= \sp(H_{\spin}) \times \tor$ with jump rates (\ref{rates}) is irreducible in the above sense.
\end{itemize}
\end{assumption}

Before proceeding to the results, let us give some examples of models for which all our assumptions are satisfied.   
\begin{example}[optical phonons]
Let 
\beq\label{eq: dispersion phonon}
\omega(q)=(m^2_{\mathrm{ph}}+ \sum_{i=1}^d \sin^2 \frac{q_i}{2} )^{\frac{1}{2}}, \qquad \textrm{with ${m_{\mathrm{ph}}}\neq 0$}. 
\eeq This is a typical  dispersion relation of an optical phonon branch.  Note that  $\str \mathrm{det}(\mathrm{Hessian}(\om)) \str$ is bounded away from $0$ in a neighbourhood of $0$.  Let the form factor is $\phi(\cdot)$  be a smooth function on $\bbR^d$ with compact support in this neighborhood of $0$. Then a standard argument relying on stationary phase estimates yields
\beq
\sup_x \str \zeta(x,t) \str \leq  C (1+ \str t \str)^{-d/2}
\eeq
 Hence Assumption \ref{ass: decay micro alpha} is satisfied with $\al = (d/2)-1-\delta$, for any $\delta>0$. 
To ensure that Assumption \ref{ass: fermi golden rule} is satisfied, we choose $\scrS= \bbC^2$ with  $W= \sigma_x$ and $H_{\spin}=  \ve_{0} \sigma_z$ where we use the traditional notation  $\scriptsize{\si_x= 
\left(\begin{array}{cc} 0 & 1 \\ 1 & 0 \end{array} \right),   \sigma_z = \left(\begin{array}{cc} 1 & 0 \\ 0 & -1 \end{array} \right)\,
}$. Then a sufficient condition for  Assumption \ref{ass: fermi golden rule} is that  $\caM_{\ve_0} \cap  \supp\phi$ has positive measure on $\caM_{\ve_0}$. With these choices, our equilibrium result   ($\be_1=\be_2$) holds for $d\geq 3$,  and the non-equilibrium result ($\be_1 \neq \be_2$) holds for  $d \geq 4$. 
\end{example}
\begin{example}[acoustical phonons]
If one considers acoustical phonons, for example with the above dispersion relation \eqref{eq: dispersion phonon} with ${m_{\mathrm{ph}}}=0$,  then, for smooth $\phi(q)$,  $\sup_x \str \zeta(x,t) \str \leq  C (1+ \str t \str)^{-(d-1)/2}$. Hence, the equilibrium result holds for $d\geq 4$ and the nonequilibrium one for $d \geq 5$.
From the point of view of our techniques, the important difference with the above example lies not therein that $\inf \om =0$, but in the fact that $\om(q)$ is linear in $\str q \str $ for small $q$.  
\end{example}

\subsubsection{Diffusion}\label{Diffusion}

Our most important result concerns diffusion of the particle. To state this somehow concisely, we note that the density matrix $\rho_{\sys, t}$ determines a   probability measure $\bbP_t(x)$ on $ \bbZ^d$:
 $$
 \bbP_t(x)= \Tr_{ \smallspace} [\rho_{\sys, t} (x,x)]=\sum_{e \in {\sp}(H_{\spin}) } \rho_{\sys, t} (x,e;x,e) $$
 where $ \Tr_{ \smallspace}$ is the partial trace on $\scrB(\smallspace)$ (below $\Tr$ is the full trace on $\scrH_\sys$), and we wrote the matrix $\rho_{\sys, t} (x',x)\in\caB({\smallspace})$ in the basis indexed by
 $\sigma(H_\spin)$ as $\rho_{\sys, t} (x',e';x,e)$. Physically speaking,
 $\bbP_t(x)$ is the probability
 to find the particle at time $t$ on the lattice site  $x \in \bbZ^d$. Equivalently, for an observable $F(X)$
\beq
 \sum_{x}\bbP_t(x) F(x)  =  \Tr   [ \rho_{\sys, t}  F(X)],   \qquad    F \in l^{\infty}(\bbZ^d)
\eeq
Hence, the function $\ga \mapsto  \Tr   [ \rho_{\sys, t}  \e^{\i \ga X}]$, figuring in the theorems below, is the characteristic function of the probability measure $\bbP_t$. 

 \setcounter{theorem}{0}
\bet[Equilibrium]   \label{thm: equilibrium}    
Assume that Assumptions \ref{ass: symmetries}, \ref{ass: fermi golden rule} and Assumption \ref{ass: decay micro alpha} with $\al >  1/4$, hold, and moreover, that  $\be_1=\be_2$.    Then, there is a $\la_0>0$ such that, for $0 <\str \la \str<\la_0$,  the characteristic function of  $\frac{X}{\sqrt{t}}$
converges to a Gaussian: for all  $ k \in \bbR^d  $
\beq
 \lim_{t\to\infty}  \Tr [\e^{\i k \frac{X}{\sqrt{t}}  }  \rho_{\sys, t} ]   =  \e^{-\str k\str^2 D^{\star}},  \label{eq: convergence to gaussian}
\eeq
for some strictly positive diffusion constant $D^{\star}$.  Moreover, also moments of  $\frac{X}{\sqrt{t}}$  converge to moments of the Gaussian. 
That is, for any  multi-index $I$
\beq
 \lim_{t\to\infty}\Tr [(\frac{X}{\sqrt{t}})^I \rho_{\sys, t} ] =     
  (-i \partial_k)^I \e^{-\str k\str^2 D^{\star}}\Big\str_{k=0} .
\eeq
\eet

\bet[Non-Equilibrium]    \label{thm: nonequilibrium}
Assume that Assumptions \ref{ass: symmetries}, \ref{ass: fermi golden rule} and Assumption \ref{ass: decay micro alpha} hold with $\al >  1/2$, but possibly  $\be_1 \neq  \be_2$. 
Then, there is a $\la_0>0$ such that, for $0 <\str \la \str<\la_0$, the conclusions of Theorem \ref{thm: equilibrium} hold
\eet

\subsubsection{Thermalization and decoherence} \label{sec: thermalization}

We describe some further results that could be of interest.  The following result states that some observables tend to an asymptotic value as $t \to \infty$.   
Let $\rho^{\be}_{\sys\res}$ be the (finite volume) Gibbs state at inverse temperature $\be$ corresponding to the interacting Hamiltonian $H$, i.e.\
\beq
\rho^\be_{\sys\res} =  \frac{1}{\Tr(\e^{-\be H})}\e^{-\be H}
\eeq
Take $\La=\bbZ^d$ now and and consider a `particle observable'  $A \in \scrB(\scrH_\sys)$ that is translation-invariant, that is, its kernel  satisfies $A(x',e';x,e)= A(x'+y,e';x+y,e) $ for any $y \in \bbZ^d$ (recall the notation introduced in Section \ref{Diffusion}).   Whenever necessary, we define a corresponding  finite-volume ($\La=\bbZ^d/L \bbZ^d$) observable by  restriction, i.e.\ $A^{\La} :=\lone_{\La} A\lone_{\La}$ with $\lone_{\La}= \lone_{\smallspace} \otimes \lone_{l^2(\La)} $, using the identification of $\La$  with a subset of $\bbZ^d$.  To avoid technicalities, we moreover demand
\beq
\str  A(x',e';x,e) \str \leq C \e^{-c \str x -x'\str }
\eeq
The algebra of all  translation-invariant $A \in \scrB(\scrH_\sys)$ satisfying such a condition (constants $0<c, C <\infty$ can depend on $A$) is denoted by $\frA$. 
 
Define the equilibrium expectation value of $A \in \frA$ as 
 \beq \label{def: gibbs expectation of observable}
 \langle A \rangle_{\be}  :=   \lim_{\La \nearrow \bbZ^d} \Tr  [ \rho^\be_{\sys\res}  (A^{\La} \otimes \lone_\res) ]
 \eeq
Note that $ \rho^\be_{\sys\res} $ depends on $\La$, too. The existence of the limit follows easily by the expansions in Section \ref{sec: bounds on boundary correlation functions}.

\bet \label{thm: convergence to gibbs}
Let  $A \in \frA$.
Under the conditions of Theorem \ref{thm: equilibrium} (including the restriction on $\str \la \str$), we have
\beq
\lim_{t \to \infty}   \Tr [\rho_{\sys,t}  A  ]  =     \langle A \rangle_{\be}, \qquad  \textrm{where} \, \be:=\be_1=\be_2  \label{eq: convergence to stat state}
\eeq 
\eet  
If we drop the condition that $\be_1= \be_2$, then the asymptotic state of the particle is given by a NESS (non-equilibrium steady state):
\bet \label{thm: convergence to ness}
Let  again $A \in \frA$.
Under the conditions of Theorem \ref{thm: nonequilibrium}, we have
\beq
\lim_{t \to\infty}   \Tr [\rho_{\sys,t}  A  ]  =     \langle A \rangle_{\mathrm{ness}} , 
\eeq
for some linear functional $A \to \langle A \rangle_{\mathrm{ness}}  $ on $\frA$, satisfying  $\str \langle A \rangle_{\mathrm{ness}} \str \leq \norm A \norm_{\scrB(\scrH_\sys)}$,   $ \langle \lone \rangle_{\mathrm{ness}} =1$ and  $ \langle A \rangle_{\mathrm{ness}}  \geq 0 $ for $A\geq 0$.
\eet  

The easiest way to give more details on the NESS consists in using the fact that it is a small perturbation (in $\la$) of the NESS that one obtains in the Markovian approximation to our model, and this will be described below. Let us however quote one important property of both functionals $\langle \cdot \rangle_{\textrm{ness}} $ and $\langle \cdot \rangle_{\be}$, namely decoherence. 

Choose the observables $A_{y}$ with kernel $A_y(x',e',x,e)= \delta_{x'-x,y}a(e,e')$, then there is a $\ga_0>0$
\beq \label{result: decoherence}
\str \langle  A_{y} \rangle_{\mathrm{ness}} \str  \leq C  \e^{- \hf \ga_0 |y|}
\eeq
for some constant $C$, independent from $y$.  The same statement holds for $ \langle A_y \rangle_{\be}$ as well, but this does not require our analysis since it is a property of the interacting Gibbs state introduced above. 

\subsubsection{The Markov approximation} \label{sec: markov approximation}
As mentioned already above, we can describe our results qualitatively by referring to the Markov approximation to our model. Strictly speaking this Markov approximation retains the quantum nature of the problem, but as far as the the long-time properties of the system are concerned, we can describe what is happening with the help of the 'classical' Markov process that was already introduced before Assumption \ref{ass: fermi golden rule} by specifying the rates $j(\cdot,\cdot)$ on the state space $\caF={\sp}(H_{\spin})\times \tor$. One should think of the elements in $\caF$ as 'good quantum  numbers' for the Hamiltonian $H_\sys$, with $k \in \tor$ momentum and $e \in {\sp}(H_{\spin})$  spin.
Since we assumed it to be irreducible, this Markov process has a unique invariant state given by an absolutely continuous positive measure. We call the associated density $\mu_{Q}(k,e)$. 

If $\be_1=\be_2$, then  we simply have  
\beq
\mu_Q(k,e) = (1/2\pi )^d  \e^{-\be e} (\sum_{e'} \e^{-\be e'}   )^{-1}  
\eeq
The fact that this Gibbs state is uniform in $k$ is due to the fact that the kinetic energy was chosen so small that on the time scale $\la^{-2}$ for which the Markov approximation is valid, the kinetic degrees of freedom do not couple directly to the phonons. \\
For $\be_1 \neq \be_2$, we do not have an explicit expression for the invariant density $\mu_Q(k,e)$, but it is connected to the NESS discussed in the previous sections as follows.   For  $A \in \frA$, we define
\beq
\langle A \rangle_{\textrm{ness}, Q} : =    \frac{1}{(2 \pi)^d}  \mathop{\int}\limits_{\tor} \d k \sum_e    \mu_Q(k,e)    \sum_{x}  \e^{\i k x} A(0,e, x,e)  \label{eq: explicit construction ness}
\eeq
Then we have
\beq
\langle A \rangle_{\textrm{ness}}  - \langle A \rangle_{\textrm{ness}, Q}  =  o (\str\la\str^0), \qquad  \la \to 0
\eeq
In fact, the same statement holds true when $\be_1=\be_2$, but in that case it  follows simply by explicitly comparing the ($\sys$-part of) the full interacting Gibbs state with $\mu_{Q}(k,e)$, i.e.\ no analysis of the dynamics is necessary.  Furthermore, we note that by the invariance property $j(e,'\d k'; e,k)=j(e,'\d (k'-k); e,0) $, we can deduce that  $\mu_Q(k,e)=   (1/2\pi)^{d} \mu_Q(e)$ (independent of $k$).   
For more details on how the NESS might look like in a concrete example, we refer the reader to \cite{deroeckspehner} where we construct a model that describes a ratchet. Starting from this example, it is straightforward to prove that the NESS is not an equilibrium state.

From the Markov approximation, we can also infer an approximation for the diffusion constant. Assume that the Markov process on $\caF$ describes a particle that jumps in its $(k,e)$ coordinates  and, in between jumps,  propagates freely in space with velocity 
\beq\label{group velocity}
v_i = 2m_{\mathrm{p}}^{-1}\sin k_i, \qquad  k \in \tor 
 \eeq 
 Note that, this is the gradient of  $ k\mapsto m_{\mathrm{p}}^{-1} \sum_
{i=1}^d (2- 2\cos k_i)$, which is the dispersion relation of $H_{\kin}$, up the factor $\la^2$. The disappearance of the factor $\la^2$ is due to the fact that the Markov process describes the dynamics of times of order $\la^{-2}$. 
The diffusion constant  $ D_Q$ of such a  particle is given by the velocity-velocity correlation function
 \beq
 D_Q  \delta_{i,j} =   \frac{1}{2} \int_{\bbR} \d t \langle v_i(t) v_j(0) \rangle 
 \eeq
 where the expectation $\langle \cdot \rangle$ is computed with respect to the stationary Markov process.  Then, 
the diffusion constant $D^{\star}$ in Theorems \ref{thm: equilibrium}, \ref{thm: nonequilibrium} has the asymptotics $D^{\star}=\la^2D_Q + o(\la^2)$ as $\la \to 0$.

\subsection{Related work}\label{sec: related work}
\subsubsection{Classical mechanics} \label{sec: results classical}

Diffusion has been established for the two-dimensional finite horizon billiard in \cite{bunimovichsinai}.   In that setup, a point particle travels in a periodic, planar array of fixed hard-core scatterers. The \emph{finite-horizon condition} refers to the fact that the particle cannot move further than a fixed distance without hitting an obstacle.  

In \cite{knaufergodic}, the hard-core scatterers are replaced  by a planar lattice of attractive Coulombic potentials, i.e., the potential is $V(x) =- \sum_{j \in \bbZ^2 }   \frac{1}{\str x-j \str}  $.
In that case, the motion of the particle can be mapped to the free motion on a manifold with strictly negative curvature, and one can again prove  diffusion. 

Recently, a  different approach was taken in \cite{bricmontkupiainendiffusioncoupledmaps}:  Interpreted freely, the model in \cite{bricmontkupiainendiffusioncoupledmaps} consists of a  $d=3$ lattice of confined particles that interact locally with chaotic maps such that the energy of the particles is preserved but their momenta are randomized. Neighboring particles can exchange energy via collisions and one proves diffusive behavior of the energy profile.

\subsubsection{Quantum mechanics for extended systems} \label{sec: qmextended}

 The earliest result for extended quantum systems that we are aware of, \cite{ovchinnikoverikhman},  treats a quantum  particle interacting with  a time-dependent random potential that has no memory (the time-correlation function is $\delta(t)$). Recently, this was generalized in \cite{kangschenker} to the case of time-dependent random potentials where the time-dependence is given by a Markov process with a gap (hence, the free time-correlation function of the environment is exponentially decaying).  In \cite{deroeckfrohlichpizzo},  a quantum particle interacting with independent heat reservoirs at each lattice site was treated. This model also has an exponentially decaying  free reservoir time-correlation function and as such, it  is very similar to   \cite{kangschenker}.  Notice also that, in spirit, the model with independent heat baths is comparable to the model of \cite{bricmontkupiainendiffusioncoupledmaps}, but, in practice, it is easier since quantum mechanics is linear.

The most serious shortcoming of these results (except for \cite{deroeckfrohlichhighdimension}, already discussed in the introduction) is the fact that the assumption of exponential decay of the correlation function in time is unrealistic.  In the model of the present paper, the space-time correlation function,  $\zeta(x,t)$, is the correlation function of freely-evolving excitations in the reservoir, created by interaction with the particle.
 Since momentum is conserved locally, these excitations cannot decay exponentially in time $t$, uniformly in $x$.  
 
 In the Anderson model, the analogue of the correlation function does not decay at all, since the potentials are fixed in time. Indeed, the Anderson model is different from our particle-environment model: diffusion is only expected to occur for small values of the coupling strength, whereas the particle gets trapped (Anderson localization) at large coupling. 

Finally, we mention two recent and exciting developments: 1) in \cite{disertorispencerzirnbauer}, the existence of a delocalized phase in three dimensions is proven for a  supersymmetric model, and 2) In \cite{erdosknowlesband}, delocalization of eigenvectors is proven for a class of random band matrices. Both models can be thought of as toy versions of the Anderson model. 

\subsubsection{Quantum mechanics for confined systems}

The theory of confined quantum systems, i.e., multi-level atoms,  in contact with quasi-free thermal reservoirs has been intensively studied in the last decade, e.g.\ by \cite{bachfrohlichreturn, jaksicpillet2, derezinskijaksicreturn, merklicommutators, deroeckkupiainen}. 
In this setup, one proves approach to equilibrium for the multi-level atom.  Although at first sight, this problem is different from ours (there is no analogue of diffusion), the techniques are quite similar and we were mainly inspired by these results.   However, an important difference is that, due to its confinement,  the multi-level atom experiences a free reservoir correlation function with better decay properties than that of our model (if one assume enough infrared regularity) and the 'Markov approximation' to the model has a gap, whereas in the extended system it is diffusive. 

\subsubsection{Scaling limits}

Up to now, most of the rigorous results on diffusion starting from deterministic dynamics are formulated in a \emph{scaling limit}. This means that one does not fix one dynamical system and study its behavior in the long-time limit, but, rather,  one compares a family of dynamical systems at different times, as a certain parameter goes to $0$. 
 The precise definition of the scaling limit differs from model to model, but, in general,  one scales time, space and the coupling strength (and possibly also the initial state) such that the Markovian approximation to the dynamics becomes exact.   In the model of the present paper, one can do this by considering the dynamics for times $t =\caO(\la^{-2})$ and then taking $\la \to 0$. The result is a Markovian approximation which was already referred to in Section \ref{sec: markov approximation}.  Of course, the point of  the present paper is that we go beyond the Markovian approximation and we describe the dynamics for infinite times at fixed $\la$.
 There are quite some results on scaling limits in the literature, for example \cite{szasztothscalinglimit, erdosyauboltzmann, erdossalmhoferyaunonrecollision, lukkarinenspohnnonlinear, komorowskiryzhikdiffusion, durrgoldsteinlebowitz},  and we do not attempt an overview.

\subsection{Strategy of proof} \label{sec: strategy of proof}
We will now give a road map for the proof of Theorems \ref{thm: equilibrium}  and \ref{thm: nonequilibrium}. 
It is based on a careful analysis of the long time properties of the reduced particle dynamics
given by the operator 
$Z_t$ defined in eq.
(\ref{Sdyn}) and extended to infinite volume in Section \ref{sec: infinite volume setup}. We view $Z_t$ as a bounded map $Z_t: \scrB_1(\scrH_\sys) \to \scrB_1(\scrH_\sys)  $
i.e. 
$$Z_t\in   \scrB(\scrB_1(\scrH_\sys)).
$$
The main ingredient will be a proof that in a suitable norm the rescaled large time limit 
\beq
\lim_{t\to\infty}\bsS_{\sqrt{t}}Z_t  %=T^{\ast}
\label{scalinglimit}
 \eeq
exists\footnote{This statement is of course only true in infinite volume}.  Here $\bsS_{\sqrt{t}}$ is an operator implementing the scaling of the particle position occurring in
  Theorem \ref{thm: equilibrium} in the space $ \scrB(\scrB(\scrH_\sys))$ explained in detail in Section \ref{sec: kernels and rescaling of space}.
  This large time limit is controlled in two steps. 

\subsubsection{Random walk in random environment}
  The first step consists of studying the dynamics on the time scale ${\mathcal O}(1/\la^2)$. This scale   
  is large enough so that the dissipative effects that we want to exhibit are clearly visible, and small enough such that a simple Duhamel expansion can be  controlled.  Thus, let us fix $t_0= \la^{-2}\frt_{0}$ with $\frt_0$ of $\caO(1)$.  
 The  evolution of the system plus environment up to this time is given by
\beq
 {\cal U} :=  \e^{-\i t_0  L} .%, \qquad  L=\ad(H),
 \label{calU}
\eeq
 % acting on $\scrB_1(\scrH_S\otimes\scrH_E)$.
  Denote the reduced time evolution on this scale by
  \beq
T:=Z_{ t_0  }.%=   \Tr_{\res}  \left[ {\cal U}(   \rho_\sys \otimes \initialresfinite  )\right].
\label{Tdef}
 \eeq
  In Section \ref{sec: weak coupling limit} we show that
 for $\la$ small enough, 
\beq  \label{eq: first closeness t and qrw}
T  = \e^{ t_0 (-\i L_\sys+ \la^2  M) } +\caO(\str \la\str^{2 \al})
\eeq
 in an appropriate norm, cfr.\ Proposition \ref{prop: weak coupling}. Here, $L_\sys=\adjoint(H_\sys)$ and $M$ is the generator
 of a ``quantum Markov process" that is very closely related to the Markov process discussed in Section \ref{sec: results}.
 
 We will now compare the full evolution (\ref{calU}) to the one where the particle dynamics
 is given by  the reduced evolution $T$ and the environment evolves freely:
 \beq
 {\cal U} =      T \otimes \e^{-\i t_0 L_\res}+B 
 \label{Udeco}
\eeq
where $ L_\res=\ad(H_\res)$.
This defines the ``excitation operator $B$''  acting on $\scrB_1(\scrH_S\otimes\scrH_E)$.
This operator is analyzed in Section \ref{sec: estimates on the first scale excitations}. There
we show that $B$ is small and "weakly correlated" (Propositions \ref{prop: integrability}
and \ref{prop: integrability with boundary}). To explain what we mean by this, consider
the full evolution  for times $t$ longer than $t_0$. Taking
$t=Nt_0$ we have
\beq
 \e^{-\i  N t_0 L} =     {\cal U}^N=(T \otimes \e^{-\i t_0 L_\res}+B )^N.
\eeq
We rewrite this
 \beq
 {\cal U}^N=e^{-\i  N t_0 L_\res}(T+B({N}))(T+B({N-1}))
\dots (T+B({1}))
\eeq
where  we use the shorthand $T$ for $T\otimes \lone$ and $e^{-\i   t_0 L_\res}$
for $\lone \otimes e^{-\i   t_0 L_\res}$ and define
\beq  \label{eq: b interaction picture}
B(\tau)    :=    \e^{\i \tau   t_0  L_\res}  B    \e^{-\i (\tau-1) t_0 L_\res}.
\eeq
 Using cyclicity of the trace and the fact $ \e^{-\i t_0 L_\res} \initialresfinite=\initialresfinite$, i.e.\ the invariance of the  environment state under the uncoupled dynamics we get
\baq
Z_{ Nt_0  }\rho_S
 = \Tr_\res \left[   (T+B(N))            \ldots (T+B(2)) (T+B(1)) ( \rho_{\sys} \otimes \initialresfinite)  \right]  
\label{basic}
\eaq

If $B(\tau), \tau=1, \ldots, N$ were set to zero in (\ref{basic}), the reduced evolution would be
given by $  Z_{ Nt_0  }= T^N $ i.e.\ by a discrete-time semigroup acting in the system space.  This is our  `quantum random walk'. Using \eqref{eq: first closeness t and qrw}, it is easy to show that it is diffusive; this is done in Section \ref{sec: estimates on the first scale reduced evolution}. More precisely, we show that the process generated by $M$ is diffusive and the diffusivity of $T^N$ follows then by simple perturbation theory. 
The presence of $B(\tau)$ produces time dependence and dependence on the
environment variables that are traced over in the end. We wish to think about
the latter as time dependent {\it noise} and the $\Tr_\res$ as an expectation
over the noise. Thus, given $D\in\scrB(\scrB_1(\scrH_S\otimes\scrH_E))$ we
define $\bbE(D)\in \scrB(\scrB_1(\scrH_S))$ by
 \beq
\bbE(D)\rho_S:= \Tr_\res \left[ D( \rho_{\sys} \otimes \initialresfinite)  \right]  .
\label{Edef}
 \eeq
Using this notation we have
\baq
Z_{ Nt_0  }=\bbE ({\cal U}^N)
\label{basic1}
\eaq
together with
 \beq
T=\bbE ({\cal U})
\label{TandE}
 \eeq
and 
\beq
\bbE ( {B}(\tau))=0
\label{EBzero}
 \eeq
 for all $\tau$ which follows again from the invariance of $\initialresfinite$.
 In Section \ref{sec: correlation functions} we generalize this expectation
 to a larger algebra needed to analyze (\ref{basic}). 
  Given a set  $A=\{\tau_1,\tau_2,\dots, \tau_m\}\subset \{1,2,\dots,N\}$ with the convention that
$\tau_i<\tau_{i+1}$ 
we define the time-ordered {\it correlation function} 
\beq
G_A :=    \bbE \left( B({\tau_m})\odot  B({\tau_{m-1}})\odot\dots\odot B({\tau_1}) \right) 
\in[ \scrB(\scrB(\scrH_\sys))]^{{\otimes}^m}%  \in   (\scrR_S)^{ \otimes^m}
\label{eq: first definition of correlation function}
\eeq
Here $\odot$ denotes a product that is tensor product in the system space and
operator product in the environment
space, see Section \ref{sec: correlation functions}. Like in classical probability we introduce
in Section \ref{sec: Connected correlation functions}
 {\it connected correlation functions} or cumulants $G_A ^c$ in our non-commutative setup.
Propositions \ref{prop: integrability}
and \ref{prop: integrability with boundary} show that these cumulants are small and
inherit decay in time from the environment correlation function.

Thus we want to think about $T+B(\tau)$ as transition probability kernels in the system space
which are random due to the environment dependence
and we want to prove average or ``annealed" diffusion for the associated process, a quantum
version of random walk in random environment.

\subsubsection{Renormalization group} The second step in the control of the large time asymptotics of $Z_t$  is to
control the large $N$ asymptotics of (\ref{basic1}) composed with the scaling.
This is achieved using  Renormalization Group (RG) method which
 consists of studying  (\ref{basic})  in an inductive way. Pick an integer $\ell$ and define the RG map
\beq
\bsR (\caU) :=  \bsS_\ell (\caU^{\ell^2}), 
\label{RGmap} 
  \eeq
  where the scaling $ \bsS_\ell=\bsS_\ell\otimes \lone$ acts on $ \scrB(\scrB_1(\scrH_\sys))$ only (see Section
  \ref{sec: kernels and rescaling of space}).
Iterating it $n$ times, and using $(\bsS_{\ell})^n=\bsS_{\ell^n}$ we get
\beq
\caU_n := \bsR^n(\caU)  =     \bsS_{\ell^n}( \caU^{\ell^{2n}}).
\eeq
Define now
\beq
T_n:=\bbE (\caU_n).
\label{Tndefi}
\eeq
By (\ref{basic1}) $T_n$ gives the (rescaled) reduced dynamics at time $\ell^{2n}t_0$:
\beq\label{tnvsZ}
T_n= \bsS_{\ell^n}Z_{\ell^{2n}t_0}.
\eeq
Thus the existence  of the limit (\ref{scalinglimit}) is closely related to proving (in a sense to be made precise) the existence of
\beq
\lim_{n\to\infty}T_n=T^{\ast}
\label{scalinglimit1}
 \eeq
This is the content of Proposition 
\ref{prop: overview t behavior}. 

$T_n$ is analyzed inductively in $n$. We define the renormalized noise $B_n$ 
in analogy with (\ref {Udeco}) by
\beq
\caU_n=T_n\otimes e^{-\i  \ell^{2n}t_0 L_\res}+B_n.
   \label{Undeco}
\eeq
$B_n$ is studied through its cumulants $G_{n,A}^c$ defined as for $B$. The RG map
(\ref{RGmap}) leads to a recursion
\beq
T_{n+1}={\caF}(T_n, G^c_{n,\bullet})
   \label{Trecursion}
\eeq
where the map ${\caF}$ is given explicitly in eq.\ \eqref{eq: Z from connected correlation functions}  and another
one for the correlation functions
\beq
G^c_{n+1, A}={\caF}_A(T_n, G^c_{n,\bullet})
   \label{Grecursion}
\eeq
with  ${\caF}_A$ given  in eq. (\ref{eq: from n to npluseen}). The recursion (\ref{Trecursion})
is studied in Section \ref{sec: flow of t} and (\ref{Grecursion}) in Sections \ref{sec: linear rg flow}
and \ref{sec: nonlinear rg flow}. For this analysis a choice of suitable norm for $G_{n,A}^c$  is
crucial: Norms which respect the structure of the iteration (\ref{Grecursion}) are constructed
in Section \ref{sec: norms}.

The convergence (\ref{scalinglimit1}) is a consequence  of the fact that
 that the map (\ref{Grecursion})  contracts the norms of the correlation
functions $G^c_{n,\bullet}$ and preserve their temporal decay  i.e. the noise is ``irrelevant". 
For the contraction one needs to study carefully the linearization of the map ${\caF}_A$
(Section  \ref{sec: linear rg flow}). A crucial input to this analysis is the existence of
two symmetries in the problem: unitarity and reversibility (Sections \ref{sec: ward identity from unitarity}
and \ref{sec: ward identity from equilibrium}) that provide the
contractive factors.  However, the symmetry from reversibility is only present in the equilibrium case, i.e.\ $\be_1=\be_2$, and this is the reason why our results are stronger in that case.

\section{Renormalization group formalism}  \label{sec: general}

In this section, we define the correlation functions $G_{A}$ and derive the recursion
relations for their renormalized versions $G_{n,A}$ as well as for
the $T_n$. We also discuss the symmetries preserved by the RG map.
Throughout this section and the next one, all expressions are in finite volume.  In particular, all vector spaces that we introduce are  finite-dimensional (with one exception: $\scrR^{\otimes^\bbN}$) and all sums that appear are finite. The infinite volume limit and the definitions of the tensor product spaces in that limit
will be discussed in Section \ref{sec: infinite volume setup}.

\subsection{Correlation functions of excitations} \label{sec: correlation functions}

We will now define the correlation function (\ref{eq: first definition of correlation function})
and explain how the reduced dynamics (\ref{basic}) can be expressed in terms of it.
We abbreviate
\beq
\scrR_\sys =  \scrB(\scrB_1(\scrH_S)), \qquad    \scrR_\res =  \scrB(\scrB_1(\scrH_\res))
\eeq
such that $U, B(\tau)$ are elements of $\scrR_\sys  \otimes  \scrR_\res$ and $T$ is an element of $\scrR_\sys$. 
Define, for $D,D' \in \scrR_\sys  \otimes  \scrR_\res$ the object
$$D\odot  D'\in \scrR_\sys  \otimes \scrR_\sys  \otimes \scrR_\res $$
 as an operator product in $\res$-part and tensor
 product in $\sys$-part. Concretely, let first $D=D_\sys \otimes D_\res$ and $D'=D'_\sys \otimes D'_\res$.
 Then 
$$D\odot D':= D_\sys  \otimes D'_\sys  \otimes D_\res D'_\res.$$
Extend then by linearity to the whole space $\scrR_\sys  \otimes  \scrR_\res$ .
Iterating this construction we define for $D_i \in \scrR_\sys  \otimes  \scrR_\res  $, $i=1,\dots,m$
$$ D_m\odot \dots\odot D_2\odot D_1 \in (\scrR_\sys)^{ \otimes^m}\otimes \scrR_\res .$$
We define the `expectation' 
$$\bbE:(\scrR_\sys)^{ \otimes^m}\otimes \scrR_\res
\rightarrow(\scrR_\sys)^{ \otimes^m}$$ on simple tensors $D_\sys \otimes D_\res,  D_\sys \in (\scrR_\sys)^{ \otimes^m}, D_\res \in \scrR_\res$ by
$$
\bbE (D_\sys \otimes D_\res):= D_\sys {\Tr}_{\res} ( D_\res \initialresfinite)).
$$
and then we extend by linearity. 
Hence we have explained the definition  (\ref{eq: first definition of correlation function}) of $G_A$.

Note that by (\ref{EBzero}) $G_{\{ \tau \}}  =0$, and 
 $G_A = G_{A+\tau}$ since $ \e^{-\i t_0 L_\res} \initialresfinite=\initialresfinite$. 
It will be convenient to label the $\scrR_\sys$'s and to drop the subscript $\sys$ (since we will rarely need $\scrR_\res$), writing simply $\scrR$ for $\scrR_\sys$.   Let us denote by $\scrR^{\otimes^{\bbN}}$ the linear space spanned by simple tensors $   \ldots \otimes V_2\otimes V_1$ where all but a finite number of $V_j$ are equal to the identity $\lone$.  For finite subsets $A \subset \bbN$, we then define $\scrR_A$ as the finite-dimensional subspace of $\scrR^{\otimes^\bbN}$ spanned by  simple tensors $   \ldots \otimes V_2\otimes V_1$ with $V_j=\lone, j \notin A$ and we write in particular $\scrR_{\tau}=\scrR_{\{\tau\}}$.   Let $A=\{\tau_1, \tau_2, \ldots, \tau_m\}$ with $ \tau_1 < \tau_2 <\ldots < \tau_m$. 
Obviously,  $\scrR_A$ is naturally isomorphic to $\scrR^{\otimes^m}$ by identifying the right-most tensor factor to $\scrR_{\tau_1}$, the next one to $\scrR_{\tau_2}$, etc\ldots We denote this isomorphism from $\scrR^{\otimes^m}$ to $\scrR_A$ by $\bsI_A$ and we will from now on write $G_A$ to denote $\bsI_A[G_A] \in \scrR_A$ since $G_A$ acting on the `unlabeled' space $\scrR^{\otimes^m}$ will not be used.

Consider a collection $\caA$ of finite disjoint subsets of $\bbN$, then each of the spaces $\scrR_{A \in \caA}$ is a subspace of $\scrR_{\supp \caA}$ where $\supp \caA =\cup_{A \in \caA} A$. Given a  collection of operators $K_A \in \scrR_A$, we have $\prod_{A} K_A  \in \scrR_{\supp \caA}$. 
However, we prefer to denote  such products  by
\beq
\mathop{\otimes}\limits_{A \in \caA}   K_A      \in \scrR_{\supp \caA},
\eeq
i.e.\ we keep the tensor product explicit in the notation. 

To express   (\ref{basic}) in terms of the correlation functions
(\ref{eq: first definition of correlation function}) we need one more operation; the contraction $\caT$ of neighboring spaces. Let $I =\{\tau_1,\ldots, \tau_m\}$ be a discrete interval, i.e.\  $\tau_i=k+i$ for some $k$, and consider a family $V_\tau \in \scrR, \tau \in I$, then we define 
\beq
\caT (\bsI_{\tau_{m}}[V_{m}]\ \otimes \bsI_{\tau_{m-1}}[V_{m-1}]\otimes\dots \bsI_{\tau_1}[V_{\tau_1}]):=V_m V_{m-1}\dots V_1  
\label{contraT}
\eeq
and then extend $\caT$ linearly to an operator $\scrR_I \to \scrR$. In formulas like the one above, we will often abbreviate $\bsI_{\tau}[ V_\tau ]$ by $ V_\tau$, which is an abuse of notation that does not cause confusion since we keep the tensor product explicit on the LHS, as indicated above.

With these definitions eq. (\ref{basic}) can be written as
\baq
Z_{N t_0}  =  \caT\left( \bbE \left[   (T+B(N)) \odot           \ldots (T+B(2)) \odot    (T+B(1))\right]   \right) 
\label{basic0}
\eaq
and expanding the product we end up with desired formula
\beq
Z_{N t_0}=\sum_{A\subset \{1,\dots,N\}}\caT [ \mathop{\otimes}\limits_{\tau \in A^c} T({\tau}) \otimes G_A  ] , \qquad  A^c =  \{1,\dots,N\} \setminus A   \label{eq: first expansion density matrix}
\eeq
where the operators $T(\tau)$ are shorthand for copies of $\bsI_{\tau}[T] \in \scrR_\tau$, as announced above, and hence, for each operator appearing in the product, we have specified the space $\scrR_A$ or $\scrR_\tau$ in which it acts (indeed, recall that $G_A$ was defined to act in $\scrR_A$).  In contrast, 
the order in which we write the factors inside the $\caT[\cdot]$  does not have any meaning.  For a completely explicit  expression of $\caT$ (and of $\caT_{A'}$, which will be introduced later), we refer to Section \ref{sec: kernel representation}.

\subsection{Connected correlation functions}\label{sec: Connected correlation functions}

The `connected correlation functions" or 'cumulants', denoted by $G^{c}_{A}$, 
are defined to be operators on $\scrR_{A}$ satisfying
\beq
G_{A'} =   
   \sum_{\scriptsize{\left.\begin{array}{c}   \textrm{partitions  $\caA $  of  $A'$}    \end{array} \right.  }}  \left( \mathop{\otimes}\limits_{A \in \caA}  G^{c}_{A}  \right)   \label{def: cumulants}
\eeq
The tensor product in this formula is well-defined since $ \scrR_{A'} = \mathop{\otimes}\limits_{A \in \caA} \scrR_{A}$ whenever $\caA$ is a partition of $A'$.
Note that this definition of connected correlation functions reduces to the usual probabilistic definition when all operators that appear are numbers and the tensor product can be replaced by multiplication. 
Just as in the probabilistic case, the relations \eqref{def: cumulants} for all sets $A'$ fix the operators $G^{c}_{A}$ uniquely since the formula \eqref{def: cumulants} can be inverted.  This is done inductively in $|A|$, i.e.\
\begin{align}
G^c_{\{\tau\}} = G_{\{\tau\}}, \qquad    G^c_{\{\tau_1,\tau_2\}} =  G_{\{\tau_1,\tau_2\}}- G^c_{\tau_2}\otimes G^c_{\{\tau_1\}}\\[2mm]
G^c_{\{\tau_1,\tau_2,\tau_3\}} =   G_{\{\tau_1,\tau_2,\tau_3\}} -   \sum_{j =1,2,3}  G^c_{\{\tau_j\}} \otimes  G^c_{\{\tau_1,\tau_2,\tau_3\} \setminus \{\tau_j\}} - \mathop{\otimes}\limits_{j=1,2,3}  G^c_{\{\tau_j\}}.
\end{align}
Note however that there are many simplifications because $G_{\{\tau\}}=G^c_{\{\tau\}}=0$ by definition.

Inserting (\ref{def: cumulants}) into (\ref{eq: first expansion density matrix}) we get
\baq
Z_{N t_0}= T^N+
  \sum_{\caA  \in \frB(I)}      \caT   \left[ \left( \mathop{\otimes}\limits_{\tau \in I  \setminus \supp \caA} T(\tau) \right)   \otimes  \left(\mathop{\otimes}\limits_{A \in \caA} G^{c}_{A} \right)    \right]   \label{eq: Z from connected correlation functions}   
\eaq
where $I=\{1,2,\dots, N\}$,  $\frB(I)$ is the set of non-empty collections $\caA$ of disjoint subsets  $A$ of $I$, and $\supp \caA= \cup_{A \in \caA} A $.  
Formula \eqref{eq: Z from connected correlation functions}   follows from \eqref{eq: first expansion density matrix} by substituting \eqref{def: cumulants}
 since, obviously, any disjoint collection $\caA$ is a partition of $\supp \caA$. The term $T^{N}$ in \eqref{eq: Z from connected correlation functions}  originates from $A=\emptyset$ in \eqref{eq: first expansion density matrix}.

\subsection{Kernels and rescaling of space} \label{sec: kernels and rescaling of space}

To define the scaling operator introduced in (\ref{scalinglimit}) we need to write our
operators in a suitable basis. 
First, let us fix a basis  for the space $\scrB_1(\scrH_\sys)$ . Since we assumed $H_{\spin}$ to be non-degenerate, 
we may label a basis of $\scrS$ by eigenvalues $e \in {\sp}(H_{\spin})$. Hence $\psi\in \scrH_\sys
=  \scrS \otimes  l^2(\La_L)$ may be identified with a function $\psi(x,e)$. 
%\beq
%\str e_\links \rangle \langle e_\rechts \str, \qquad       e_\links, e_\rechts \in {\sp}(H_{\spin})
%\eeq 
%as a basis for the finite-dimensional space $\scrB(\smallspace)$. 
Next, in this basis  $\rho\in\scrB_1(\scrH_\sys)$ becomes a kernel (matrix)   
$$\rho=\rho(x_\links, e_\links; x_\rechts, e_\rechts  ),
$$
 where $x_\links, x_\rechts \in \La_L$ and $e_\links, e_\rechts\in{\sp}(H_{\spin})$ ($\rechts/\links$ for ``right/left").
 The scaling operator will act on a particular combination of $x_\links$ and $x_\rechts $. Therefore
 we need to introduce new coordinates. Note that the construction below depends on the finite volume $\La$ (by the parameter $L$) in a trivial way, and we will not indicate this dependence explicitly.% In particular, all definitions remain meaningful for $L= \infty$. 

Let, for a vector $a \in \frac{1}{2}\bbZ^d$,  $\lfloor a \rfloor $ stand for the vector with components $
  \lfloor a \rfloor_i =  \lfloor a_i \rfloor $ (largest integer smaller than $a_i$).   Then we define
\beq
x :=  \lfloor  \frac{x_\links+ x_\rechts}{2} \rfloor, \qquad   \offdx :=   \lfloor  \frac{x_\links- x_\rechts}{2} \rfloor
\eeq
and 
\beq
\offsetx:=  x_\links   -  (x  + \offdx)    \in \{0,1\}^d
\eeq
such that we have 
\beq
x_\links = x +  \offdx +  \offsetx  \qquad   x_\rechts =  x -  \offdx
\eeq
The variables $e_\links, e_\rechts, \offsetx$ will play a minor role in our analysis. Together with $\offdx$, we gather them in one 
 symbol $s$ which hence takes values in the set 
 $$\caS:= {\sp}(H_{\spin}) \times {\sp}(H_{\spin}) \times \{0,1 \}^d \times \bbZ^d/L\bbZ^d.
 $$ 
 The new variables $(x,s)$ thus run through the set
 \beq
\bbA_0 := \bbZ^d/L\bbZ^d \times \caS,
\eeq
and we may identify $\rho \in \scrB(\scrH_\sys)$ (or, since the space is finite-dimensional, any $\scrB_p(\scrH_\sys)$) with a function on $\bbA_0$, i.e.\  $$
 \rho=\rho(x,s)
 .$$  
 Similarly
 $K\in\scrR$ is given by a function on $\bbA_0\times \bbA_0$ i.e. a kernel $K(x',s';x,s)$ which acts
 on $\scrB_1(\scrH_\sys)$  by 
\beq
(K\rho)(x',s')= \sum_{(x,s)\in\bbA_0} K(x',s';x,s) \rho(x,s) 
\label{Krho}
\eeq
 We define a distinguished subset  $\caS_{0}\subset \caS$
  \beq
 \caS_{0} := \{  s = (e_\links, e_\rechts, \offsetx, \offdx)  \in \caS \, \str  e_\links=e_\rechts,  \offdx=0, \offsetx=0  \}
 \eeq
 Note that $\caS_{0}$ is independent of $L$, $\str\caS_{0} \str=\dim \smallspace $, and that 
 \beq\label{trace vs xs sum}
 \Tr \rho =  \sum_{x, s  }  \rho(x, s)   \indicator_{\caS_{0}} (s)
 \eeq

The scaling in the renormalization step affects only the variable $x$. At first, $x$ takes values in $\bbZ^d/L\bbZ^d$, but later on it lives on finer lattices.    
This motivates the definitions
\beq
\bbX_n :=   \ell^{-n} ( \bbZ^d/L \bbZ^d), \qquad  \bbA_n := \bbX_n  \times \caS 
  \eeq
To reconstruct the original coordinates $x_\links, x_\rechts$ from the $(x,  s)$ at scale $n$, one finds
\beq
x_\links =\ell^{n} x +  \offdx +  \offsetx  \qquad   x_\rechts =  \ell^{n}  x -  \offdx.
\eeq
The scaling map transforms functions on $\bbA_n$ into functions on $\bbA_{n+1}$, i.e.\ we define 
\beq
S_{\ell} : l^{\infty}(\bbA_n) \to  l^{\infty}(\bbA_{n+1})  
\eeq
with
$$
S_{\ell}\rho (x, s)  =    \ell^d    \rho ( \ell x, s) .
$$
Since $\scrR$ is the space of kernels on $\bbA_0 \times \bbA_0$, 
let us denote by 
 $\scrR_n$  the space of kernels $K$ on $\bbA_n\times \bbA_n$ and the action
$$
\bsS_{\ell}:\scrR_n\to\scrR_{n+1}
$$
by $\bsS_{\ell} K =   S_{\ell} K S^{-1}_{\ell}$ i.e.
\beq
(\bsS_{\ell} K)(x',s'; x, s) = \ell^d K (\ell x',s'; \ell x, s).% \qquad  x,x' \in \bbX_{n+1}, \quad s,s' \in \caS
\label{sell}
\eeq
When summing over a variable $x \in \bbX_n$, it is natural to take account of the cell volume i.e.\ consider
Riemann sums. We do this by defining
\baq
\int_{ \bbX_n}  \d x f(x)  :=      \sum_{x \in \bbX_n}       \ell^{-nd}  f(x)  
 \label{def: sum over lattice}
\eaq
Also, we use the shorthand $z=(x,s)$ and the notation
\baq
 \int_{\bbA_n} \d z f(z) =   \int_{\bbA_n}  \d x \d s \  f(x,s) := \sum_{(x,  s) \in \bbA_n}        \ell^{-nd}f(x,s)  \label{def: sum over lattice1}.
\eaq
%Additionally, we will sometimes abuse notation in the following way.  Consider a function $ g \in \ell^\infty(\bbA)$. If we fix $x$, then, depending on  the value of the $i$'th component $x_i$, $\offdx_i $ takes values in either $\bbZ$ or $1/2 + \bbZ$.  We prefer to think of the function $g(x, \cdot)$ as taking values in $(1/2\bbZ)^d$, but being zero in the lattice points that are not in its domain of definition. 
These conventions are motivated by the fact that they preserve normalization. Thus we have
\beq
 \int_{ \bbX_{n+1}}   \d x\ ( S_{\ell}f) ( x )     = \int_{ \bbX_{n}}   \d x\  f ( x ).  
\eeq
and it is natural to define a trace on functions on $\bbA_n$ by
\beq
\Tr f :=    \int_{\bbA_n} \d z   \, f(z)  \indicator_{\caS_{0}} (s) = \sum_{s \in \caS_0, x \in \bbX}    \ell^{-nd}  f(x,s)
\eeq
This trace is invariant under scaling;  
\beq \label{eq: trace under scaling}
\Tr\rho=\Tr S_{\ell}\rho.    
\eeq
In particular, if $\rho \in \scrB_1(\scrH_\sys)$, hence $\rho$ a function on $\bbA_0$, then $S_{\ell^n}\rho$ is a function on $\bbA_n$ and we have $\Tr S_{\ell^n}\rho= \Tr \rho $ so that the '$\Tr$' defined here is indeed connected to the original trace of operators.
Similarly the sum in (\ref{Krho}) becomes an integral and     the product of two  kernels 
 $K_1,K_2 \in\scrR_n$ becomes 
\beq
\left(K_2K_1\right)(z',z)= \int_{\bbA_n} \d z'' K_2(z',z'') K_1(z'',z) %=     \int_{\bbA_n} \d z'' \int_{\bbA_n} \d z'''   \delta(z'',z''') K_2(z',z''') K_1(z'',z) 
\eeq
We will also need a discrete delta-function 
\beq \label{def: normalised deltafunction}
 \delta(z',z)  =   \ell^{-nd} \delta_{x',x}  \delta_{s',s}, \qquad z=(x,s), z'=(x',s')
\eeq
which satisfies
\beq \nonumber%\label{def: normalised deltafunction}
 \int_{\bbA_n} \d z'   \delta(z',z)  =1.
\eeq
In the rest of the paper, we will usually not indicate the integration space in expressions like $\int \d z, \int \d x, \int \d x$ since it can be deduced from the name of the integration variable (and the context).

\subsection{RG recursion for $T$} \label{sec: reduced dynamics}

Let us now study the renormalized reduced dynamics $T_n$ defined in
(\ref{Tndefi}). $T_n$ belongs to  the rescaled space $\scrR_n$. We copy the setup of Section \ref{sec: correlation functions}, defining spaces $\scrR_{n,A}$ as products of $\scrR_{n,\tau}$, copies of $\scrR_n$. 
We also define, analogously to \eqref{eq: b interaction picture},    
\beq
B_n(\tau):=\e^{\i \tau\ell^{2n}t_0 L_\res}   B_n \e^{- \i (\tau-1) \ell^{2n}t_0 L_\res}
\label{BNtau}
\eeq
and the 
 time-ordered {correlation function} 
\beq
G_{n,A} :=    \bbE \left( B_n({\tau_m})\odot  B_n({\tau_{m-1}})\odot\dots\odot B_n({\tau_1}) \right)   \in   (\scrR_{n})^{ \otimes^m},
\eeq
for $A =\{\tau_1, \ldots, \tau_m  \}$ with $\tau_1 < \tau_2 < \ldots <\tau_m$.  With these notations, the concepts of Section \ref{sec: correlation functions} correspond to $n=0$, i.e.\ $T$ and $B(\tau)$ introduced there will now be referred to as $T_0, B_0(\tau)$.
As in Section \ref{sec: correlation functions}, we embed $G_{n,A} $ into $\scrR_{n,A}$ and we write $T_{n}(\tau)$ to denote a copy of $T_n$ embedded in $\scrR_{n,\tau}$. 
The recursion relation for $T_n$ is obtained from the analogue of eq. (\ref{basic1}):
\baq
T_{n+1}  = \bsS_\ell  \caT\left[ \bbE \left(   (T_n+B_n(\ell^2)) \odot           \ldots (T_n+B_n(2)) \odot    (T_n+B_n(1))\right)   \right]
\label{basic2}
\eaq
which leads as in  (\ref{eq: first expansion density matrix}) to
\beq   \label{eq: Z from correlation functions}
 T_{n+1}=\sum_{A\subset \{1,\dots,\ell^2\} } \bsS_\ell \caT [ \mathop{\otimes}\limits_{\tau \in A^c} T_{n}(\tau) \otimes  G_{n,A}], \qquad A^c= \{1,\dots,\ell^2\} \setminus A.
\eeq
To get the analogue of (\ref{eq: Z from connected correlation functions}   )
 define first  $I_{\tau'}$, the set of times at scale $n$ associated to a time ${\tau'}$ at scale $n+1$, by
\beq
I_{{\tau'}} := \{ \ell^2 ({\tau'} -1)+1,   \ell^2 ({\tau'}-1)+2, \ldots,  \ell^2 {\tau'}  \}
\eeq
Then
\baq
T_{n+1}
&=&  \bsS_{\ell}[T_n^{\ell^2}] +    \sum_{\caA  \in \frB(I_{\tau'})}       \bsS_{\ell}\caT   \left[ \left( \mathop{\otimes}\limits_{\tau \in I_{\tau'}  \setminus \supp \caA} T_n(\tau) \right)   \otimes  \left(\mathop{\otimes}\limits_{A \in \caA} G^{c}_{n,A} \right)    \right]   \label{eq: T from connected correlation functions}   
\eaq
where $\frB(I_{\tau'}) $ is the set of non-empty collections $\caA$ of disjoint subsets  $A$ of $I_{\tau'}$, and $\supp \caA= \cup_{A \in \caA} A $.  Note that the $\tau'$ on the RHS is arbitrary because of time-translation invariance, cfr.\  the discussion in Section \ref{sec: correlation functions}, and because the contraction $\caT$ has been defined such that it produces operators in $\scrR$. We could have put simply $\tau'=1$ here, but we will need it anyhow in the next section.

\subsection{Recursion relations for correlation functions}  \label{sec: recursion relations for correlation functions}

We derive  the recursion for  $G^{c}_{n+1,A'}$, the cumulants at scale $n+1$, in terms of those at scale $n$. 
The set of  times at scale $n$ that contribute to $G^{c}_{n+1,A',}$ is
\beq 
I_{A'} := \cup_{\tau' \in A'}  I_{\tau'}
\eeq
We will need an extension of the contraction operator $\caT$. Given $\tau'\in\mathbb{N}$, let 
$$
 \caT_{\tau'}:\scrR_{n,I_{\tau'}}
\to\scrR_{n,{\tau'}} $$
 be the contraction as defined in Section \ref{sec: correlation functions} with the interval $I=I_{\tau'}$, but followed by the imbedding of $\scrR_{n} $ into $\scrR_{n,{\tau'}}$. Given a finite set $A'\subset
\mathbb{N}$, we set 
$$
\caT_{A'}:=\mathop{\otimes}\limits_{\tau'\in A'}\caT_{\tau'}: \scrR_{n,I_{A'}}\to\scrR_{n,A'}
$$
In words, we contract within
the time intervals $I_{\tau'}$. 
Below in Section \ref{sec: kernel representation} we illustrate the action of $\caT_{A'}$ by defining it explicitly with kernels.
Again, let $\poly(I_{A'})$ be the set of non-empty collections $\caA$ of disjoint subsets  $A$ of $I_{A'}$. 
Let us first write $G_{n+1,A'}$ in terms of objects at scale $n$:
\baq \label{eq: unconnected corr induction}
G_{n+1,A'}   =   \sum_{ \caA \in \poly(I_{A'}): \forall \tau' \in A': \supp \caA \cap I_{\tau'} \neq \emptyset }  \, \, \,     \bsS_{\ell}\caT_{A'} \left[  \mathop{\otimes}\limits_{A \in \caA}  G^{c}_{n,A}    \mathop{\otimes}\limits_{\tau \notin \supp \caA}  T_{n}(\tau)  \right]   \label{eq: from n to npluseen nonconnected}
\eaq
Here, the rescaling $\bsS_{\ell}$ acts  on all copies of $\scrR_{n}$, that is, we write $\bsS_{\ell}$ instead of $\bsS_{\ell} \otimes \ldots \otimes \bsS_{\ell}$. 

Note that any collection $\caA\in\poly(I_{A'})$ that satisfies $\forall \tau' \in A': \supp \caA   \cap I_{\tau'} \neq \emptyset$   induces a graph $\caG_{A'}(\caA)$ on the vertex set $A'$ by the prescription that  we connect $\tau_1', \tau_2' \in A'$  by an edge whenever there is an $A \in \caA$ such that 
\beq
A \cap I_{\tau'_1}  \neq \emptyset, \quad \textrm{and} \qquad   A \cap I_{\tau'_2}  \neq \emptyset,
\eeq
This leads to the basic relation between scale $n$ and $n+1$
\baq
G^{c}_{n+1,A'}   =   \sum_{ \caA \in \poly(I_{A'}): \caG_{A'}(\caA)\,  \text{connected} }  \, \, \,     \bsS_{\ell}\caT_{A'} \left[  \mathop{\otimes}\limits_{A \in \caA}  G^{c}_{n,A}    \mathop{\otimes}\limits_{\tau \notin \supp \caA}  T_{n}(\tau)  \right]   \label{eq: from n to npluseen}
\eaq
To prove this relation, we merely need to check that the recursion relation \eqref{def: cumulants}  holds.   This is done by writing each term on the RHS of  \eqref{eq: unconnected corr induction} as a tensor product over connected components $A_1',\ldots, A_m'$ of the graph $\caG_{A'}(\caA)$ and realizing that the factors of the tensor product are precisely given by  \eqref{eq: from n to npluseen} with $A'$ replaced  by $A_1',\ldots, A_m'$. 
For a detailed description of this proof in the probabilistic (i.e.\ commutative) case, we refer to \cite{bricmontkupiainenexponentialdecay}. 
 In Figure \ref{fig: cumulants}, we pick a specific $A'$ and we draw three collections that $\caA$ that contribute to the sum in \eqref{eq: from n to npluseen}. 

\begin{figure}[h!] 
\vspace{0.5cm}
\begin{center}
\def\svgwidth{\columnwidth}
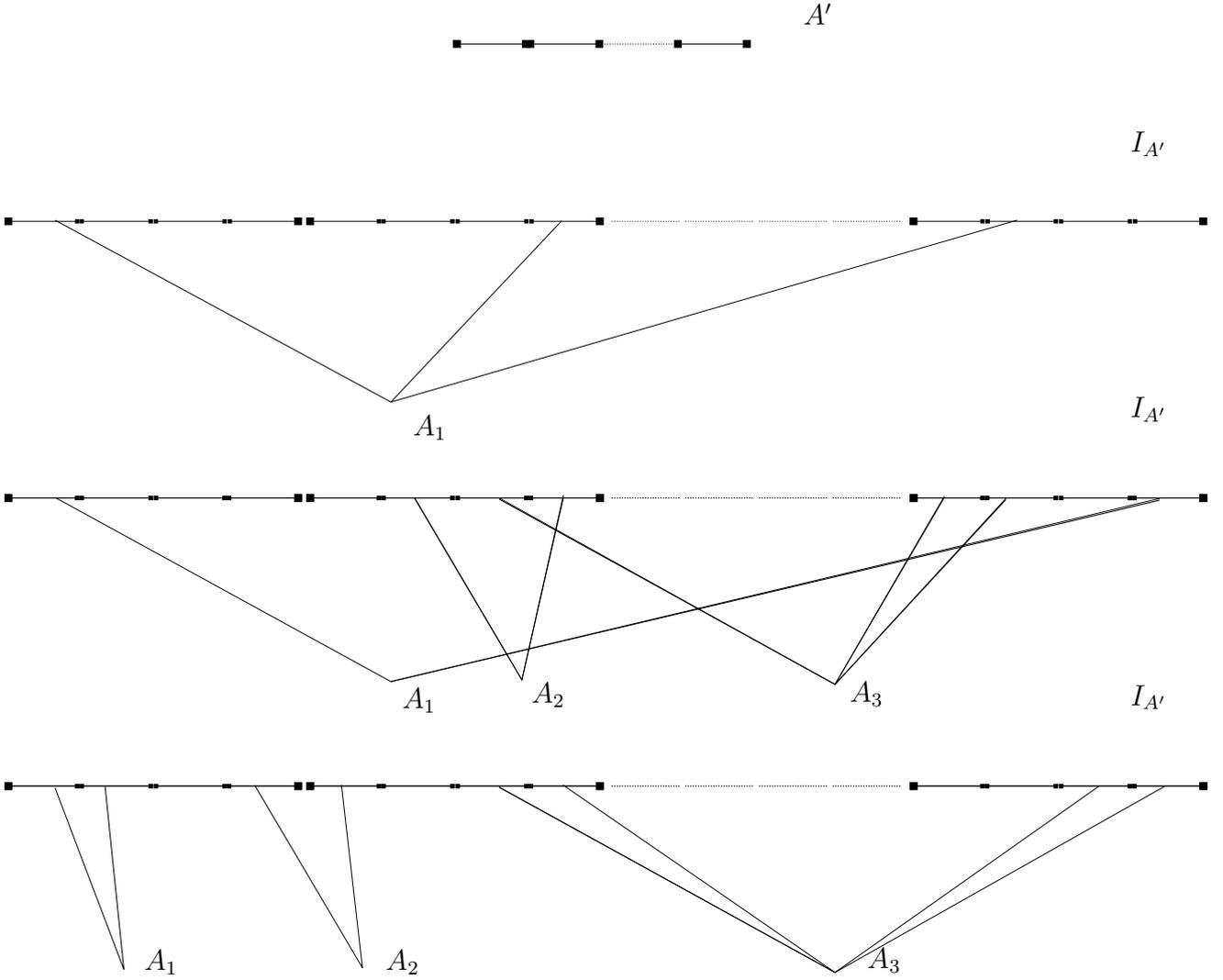
\caption{ \footnotesize{We consider the set $A'=\{1,2,4\}$ (upper part of the figure) at the scale $n$.  At scale $n+1$ this set is represented by the set $I_{A'}$. In the figure, we  took $\ell^2=4$, so that $I_{A'}=[1,12] \cup [13,16]$.  We show three different collections $\caA$ that contribute to $G_{n,A}^c$. From top to bottom, we first have $\caA= \{A_1\}$ with $A_1=\{1,8,14\}$, then  $\caA=\{A_1,A_2,A_3\}$ with  $A_1=\{1,16\}, A_2=\{6,8\}, A_3=\{7,13,14\}$    and    $\caA=\{A_1,A_2,A_3\}$  with    $A_1=\{1,2\}, A_2=\{4,5\}, A_3=\{7,8,15,16\}$.  
}   \label{fig: cumulants} }
\end{center}
\end{figure}

\subsubsection{Kernel representation and the map $\caT_{A'}$} \label{sec: kernel representation}

In Section \ref{sec: kernels and rescaling of space}, we introduced coordinates $x \in \bbX_n ,s \in \caS, z \in \bbA_n$, such that operators in $\scrR_n$ can be identified with kernels on $\bbA_n \times \bbA_n$. In a similar way, we can identify operators in $\scrR_{n,A}$ with kernels on  $\bbA_{n,A} \times \bbA_{n,A}$, where $\bbA_{n,A}$ is the cartesian product $\times_{\tau \in A}\bbA_{n,\tau}$ and  $\bbA_{n,\tau}$ are copies of $\bbA_n$.  Throughout the paper, we denote by  $z_A,z'_A$ elements of $\bbA_{n,A}$, and, hence, for any $K_A \in \scrR_{n,A} $, we have a kernel $K_A (z'_A, z_A)$.   Let us now illustrate the action of the contraction operator $K_{A'}$ on kernels. 

Consider  a partition $\caA $ of $I_{A'}$ (hence in particular $\caA \in \poly(I_{A'})$) and consider a family of kernels $K_A(z'_A, z_A)$  with $A \in \caA$. 
Let us call a set of two consecutive times $\{\tau, \tau+1 \} \in I_{A'} $ a \emph{bulk bond} whenever  $\{\tau, \tau+1 \} \subset I_{\tau'}$ for some $\tau'$. Then we have
\baq 
\left(\caT_{A'} \left[ \mathop{\otimes}\limits_{A \in \caA}  K_A       \right] \right) (\tilde z'_{A'}, \tilde z_{A'})& =& \int \d z'_{ I_{A'}  }\d z^{}_{ I_{A'}  }
 \prod_{\{\tau, \tau+1\}  \, \textrm{bulk bonds}  } \delta (z_{\tau+1},z'_\tau)  \prod_{A \in \caA}  K_{A} (z'_{A}, z_{A}) \nonumber  \\[2mm] 
 &&   \qquad \qquad \prod_{\tau' \in A'}  \delta( \tilde z'_{\tau'} ,z'_{\max I_{\tau'}}  )
 \delta(\tilde z_{\tau'} ,   z_{\min I_{\tau'}}   )   
  \label{eq: from n to npluseen coordinates}
\eaq
where the discrete $\delta$-function was defined in \eqref{def: normalised deltafunction}.  
Note that this formula is slightly different in spirit from the formulas in Section \ref{sec: recursion relations for correlation functions}, because there $\caA$ was a collection of disjoint sets, not necessarily a partition.  In each term of \eqref{eq: from n to npluseen}, we can define a partition $\caA'$  by adding to $\caA$ the sets $\{\tau\}, \tau \in I_{A'} \setminus \supp \caA$ and setting $K_{\{\tau\}}=T_n(\tau)$. Then, each term of \eqref{eq: from n to npluseen} is recast in the form \eqref{eq: from n to npluseen coordinates}.

 To avoid any possible misunderstanding, let us give a concrete example. We take $A'=\{1,2\}$ and $\ell^2=4$, such that $I_{A'}=\{1,2,\ldots,8\}$. For the partition $\caA$, we take $\caA=\{A_1, A_2\}$  with $A_1=\{1,2,5,6\}, A_2=\{3,4,7,8\}$. Then 
\baq 
\left(\caT_{A'} \left[  K_{A_2}\mathop{\otimes}  K_{A_1}       \right] \right) ( \tilde z'_1,\tilde z'_2 ; \tilde z_1,\tilde z_2  )& =& \int \d z_{2,3,4,6,7,8} \,    K_{A_1} ( z_2, z_3, z_6,z_7 ;   \tilde z_1, z_2, \tilde z_2, z_6 )  \nonumber  \\[2mm] 
 &&   \qquad \qquad   K_{A_2}( z_4, \tilde z'_1,  z_8, \tilde z'_2 ; z_3, z_4, z_7, z_8 )      \label{eq: from n to npluseen coordinates example}
\eaq
Note that we write $K_A( z'_1, z'_2 \ldots  ; z_1,z_2, \ldots )$ where $z_1$ are the coordinates of the smallest time, then $z_2$, etc.., and similarly for $z_1', z_2'$ and whenever possible, we have used $z_\tau$ to denote the unprimed coodinate corresponding to time $\tau$ (we could not do it for the primed coordinates, because they are sometimes identified with an unprimed coordinate of the previous time.) 

\section{Unitarity and Reversibility}\label{Unitarity and Reversibility}

The correlation functions $G^c_{n,A}$ satisfy identities that are crucial for the proof
of our results. The first of these  ``Ward identities" (called the 'unitarity' identity below) is  due to invariance of the trace $\Tr_{}(\cdot)$ on $\scrB_1(\scrH)$ under unitary maps
and time-independence on the reference state of the environment w.r.t.\ the uncoupled dynamics. The second one (the 'Gibbsian' identity) is due to the
invariance of the equilibrium Gibbs state under the dynamics. To discuss the latter
we need to elaborate on the relationship between free and interacting Gibbs states, see Section \ref{sec: equilibrium states}.
The Ward identities will produce some additional smallness in our expansion.  The 'Unitarity' identity eliminates certain terms exactly, the 'Gibbsian' identity does this only approximatively, as long as the temperature is not infinite.

\subsection{Ward identity from unitarity} \label{sec: ward identity from unitarity}

We start by  examining the consequences of the conservation of probability, i.e.\ unitarity of  $\e^{-\i t H}$. 
Let  $\rho_{\sys\res}\in\scrB_1(\scrH)$  and abbreviate $\rho=S_{\ell^{n-1}} \rho_{\sys\res}$. Then
\beq
   \Tr B_n (\tau)    \rho =0.% \qquad  \textrm{for any}  \qquad  \rho_{\sys\res} \in \scrB_1(\scrH)
\eeq
where the total trace $\Tr$ could also be written as $\Tr = \Tr \Tr_\res$ where the $\Tr$ on the RHS acts on functions on $\bbA_n$ as defined in \ref{sec: kernels and rescaling of space}.
Indeed, from (\ref{Undeco}), (\ref{RGmap}) and (\ref {BNtau}) we infer 
$$ B_n (\tau)\rho =S_{\ell^{n}}(U_n\rho U_n^{-1})-(T_n\otimes \lone_{\res})\rho$$
where $U_n\in\scrB(\scrH)$ is unitary.
Thus, by using \eqref{eq: trace under scaling},
\baq
   \Tr B_n (\tau) \rho &= &      \Tr  ( U_n  \rho U_n^{-1})  -  \Tr ((T_n\otimes \lone_{\res})   \rho)   \nonumber \\
   &  = &   \Tr \rho -   \Tr  \rho    =0
\eaq
By picking appropriate $\rho_{\sys\res}$ this identity leads to the following
one for  kernels:
\beq  \label{eq: unitarity ward in terms of kernels}
\int_{ \bbA_n} \d z'_\tau      \indicator_{\caS_{0}} (s'_\tau)  G^{c}_{n,A}(z'_A, z_A)       =0, \qquad    \tau=\max A, z'_{\tau}=(x'_{\tau},s'_{\tau}).
\eeq
We rewrite this relation in the form in which it will be used.  Let $P \in \scrB(l^{\infty}(\caS))$ be a one-dimensional projector with kernel 
\beq 
P(s',s)  =        \mu(s')  \indicator_{\caS_{0}} (s),
\eeq
where  $ \mu(\cdot)$ is some function that satisfies (to ensure that $P$ is indeed a projector) the normalization condition: 
\beq
 \sum_{s' \in \caS_{0}}  \mu(s') =1.
\eeq
In the rest of the paper, rank-one operators like $P$ will be written as $P= \str \mu \rangle \langle 1_{\caS_0} \str $.
Then the Ward identity can also be expressed simply as 
\beq
 \mathop{\int}\limits \d x'_\tau   ({P} \otimes 1  \ldots  \otimes 1)   G^{c}_{n,A}(x'_A, x_A )=0, \qquad      \tau = \max A
  \eeq
  where  $P$ acts on (the $s$-coordinates of) $\scrR_{\tau = \max A}$.

\subsection{Equilibrium setup: the case $\be_1=\be_2$}  \label{sec: construction of zero unitary}
In this section, we assume that the reference state of the two environments is in equilibrium at temperature $\be = \be_1= \be_2$. 

\subsubsection{Equilibrium states}  \label{sec: equilibrium states}
We already introduced the state $\initialresfinite$ on $\scrB(\scrH_\res)$. We now introduce a corresponding finite-volume reference  state on $\scrB(\scrH_\sys)$. It is given by 
\beq
\initialsysfinite = \left(\Tr[  \e^{-\be  H^{}_\sys }  ] \right)^{-1}    \e^{-\be  H^{}_\sys }  
\eeq
and we define 
\beq
\initialfinite = \initialsysfinite \otimes \initialresfinite
\eeq
We will also need the Gibbs state of the coupled system. We introduce 
\beq
\initialgibbsfinite = \left( \Tr [\e^{-\be  H } ] \right)^{-1}     \e^{-\be  H  }, \qquad     \rho^{\be}_{\sys} :=   \Tr_\res  \initialgibbsfinite
\eeq
Hence we adopt the convention of using the superscript ``$\mathrm{ref}$'' for states that are Gibbs with respect to the uncoupled Hamiltonian, and the superscript $\be$ for interacting Gibbs states and restrictions of those. In particular, note that $\rho^{\be}_{\sys}$ is not a Gibbs state but  merely a restriction of a Gibbs state. 

 To ease the discussion of the thermodynamic limit, we also introduce 
 \beq
 \initialnu :=     \str \La \str \initialsysfinite, \qquad 
 \initialnugibbs:=    \str \La \str    \rho^{\be}_{\sys}\label{newstates}
 \eeq
As a consequence of translation invariance, the kernels of  $ \initialnugibbs,  \initialnu $are constant in the variable $x \in \bbX_0$, hence they reduce to functions of $s \in  \caS$; 
\beq\label{newstates1}
 \initialnugibbs(x,s)= \mu^\be(s), \qquad    \initialnu (x,s) = \mu^{\reff}(s)
\eeq
Moreover, the chosen normalization ensures that 
\beq
\sum_{s \in \caS_{0}} \mu^\be(s) =  \sum_{s \in \caS_{0}}  \mu^{\reff}(s)=1
\eeq
Next, we introduce the Radon-Nikodym derivative as an unbounded operator on $\scrH$ 
\beq
D_{\mathrm{rd}} :=  \e^{ \frac{1}{2} \Delta F(\beta)}   \e^{-\frac{\be}{2} H}   \e^{ \frac{\be}{2} ( H_\sys + H_\res)} ,    
\eeq
where $\Delta F(\beta) $ is a number defined by
\beq
\e^{\Delta F(\beta)}   =    \frac{ \Tr \left[  \e^{ -\be( H_\sys + H_\res)}  \right]  }{\Tr \left[   \e^{ -\be H}  \right] } 
\eeq
One can check that the 'free energy difference'  $\Delta F(\beta)$ has a limit as $\La \nearrow \bbZ^d$ (it is not proportional to $\str \La\str$, because the number of  particles does not grow with the volume).
 $D_{\mathrm{rd}}$ is an unbounded operator, even in finite volume, but it is obviously well-defined by the functional calculus.
The reason why we call  $D_{\mathrm{rd}}$ a `Radon-Nikodym derivative' is that, formally, 
\beq
 \initialgibbsfinite =    D_{\mathrm{rd}}  \initialfinite D^*_{\mathrm{rd}}
\eeq

\subsubsection{Correlation functions involving  $\tau=0$} \label{sec: correlations involving zero}
Let us define  the operator $\caU_0(0)$  by
\beq
\caU_0(0) \psi  :=   D_{\mathrm{rd}} \psi  D_{\mathrm{rd}}^* 
\eeq
for $\psi$ in an appropriate subset of $\scrB_2(\scrH)$, 
such that formally
\beq
 \initialgibbsfinite =    \caU_0(0) \initialfinite.
\eeq
It is not hard to see that  $\caU_0(0)$ is densely defined on $\scrB_1(\scrH)$. This is however not
our point. Rather we want to treat the operator $\caU_0(0)$ en par with the operator $\caU_0$ introduced previously, hence we write
\beq
\caU_0(0)= T_0(0)\otimes \lone_\res + B_0(0)
\eeq
where
\beq
T_0(0)   :=    \bbE \  \caU_0(0)
\eeq
We can now extend formally the definition \eqref{eq: first definition of correlation function} such that the set $A$ can include $\tau=0$. We define an additional copy, $\scrR_{n,0}$, of $\scrR_n$ such that again $G_{n,A}$ is embedded into $\scrR_{n,A}$.
When defining the renormalization transformation of $\caU_0(0)$, we omit the iteration and keep the rescaling only, such that we have
\beq\label{0operators}
\caU_n(0) :=  \bsS_{\ell^{n}} \caU_0(0), \qquad   B_n(0) :=  \bsS_{\ell^{n}} B_0(0), \qquad
T_{n}(0)  :=  \bsS_{\ell^n} T_0({0}) \\
\eeq
With these definitions, we can also formally define the correlation functions $G_{n,A} \in \scrR_{n,A}$ at higher scales $n\geq1$. 

\begin{lemma}
The operators  $T_{n}(0)$ and $G_{n,A},G^c_{n,A}$ with $A \ni 0$, are bounded.   The relation  \eqref{eq: from n to npluseen}   remains valid for $A \ni 0$ if one defines $I_0 :=\{0\}$.

\end{lemma}
The boundedness will be easily proven with help of expansions in Section \ref{sec: bounds on boundary correlation functions}. The validity of the recursion relation is formally obvious.
For later use, we also define the bounded map $\breve Z_t :  \scrB_1(\scrH_\sys) \to \scrB_1(\scrH_\sys) $
\beq
\breve Z_t \rho_\sys : = \Tr_\res [  \e^{-\i t L}    (D_{\mathrm{rd}} \rho_{\sys} \otimes \initialresfinite D_{\mathrm{rd}}^*)  ] =   \Tr_\res [  \e^{-\i t L}  \caU_0(0) (\rho_{\sys} \otimes \initialresfinite)  ] 
\eeq

\subsection{Ward identity from equilibrium}\label{sec: ward identity from equilibrium}

To state this Ward identity, 
it is convenient to introduce an additional piece of notation. Given $K \in {\scrR}\otimes{\scrR}$
and $\psi\in {\scrB_1(\scrH_\sys)}$, let $K\psi \in \scrR \otimes\scrB_1(\scrH_\sys)$ be given by
\beq   \label{eq: convention acting on only one function}
( V_2 \otimes V_1)\psi:=V_2\otimes( V_1 \psi)\eeq
on elementary tensors  $K = V_2 \otimes V_1$ and then extending by linearity.  This notation is only used in the lemma below and its proof. 

Let us also define the scaled versions of the states (\ref{newstates})
\beq
\initialnun :=  S_{\ell^{n}}  \initialnu, \qquad    \initialnugibbsn :=   S_{\ell^{n}}   \initialnugibbs 
\eeq
Then we have
the Ward identity:
 \begin{lemma} \label{lem: correlation equilibrium} 
\baq
\bbE(  B_n(  \tau)    \odot   B_n (1))  \initialnugibbsn & =& 
-\bbE(    B_n( \tau )  \odot   B_n(1)B_n(0)   )  \initialnun -  \,
\bbE(     B_n( \tau )  \odot  (T_n \otimes \lone_{\res}) B_n(0)   ) \initialnun
  \\[2mm]
  &+ & \bbE( B_n( \tau-1 )  \odot  B_n( 0)  )   \initialnun                   
\eaq

\end{lemma}

\begin{proof}
Up to trivial rescaling, the proof is the same for all $n$, we will therefore for simplicity set $n=0$ and omit it everywhere.
We also abuse notation by writing $T_n, T_n(0)$ for $T_n \otimes \lone_\res, T_n(0) \otimes \lone_\res$.
We introduce the projector $\caP$ that acts on $\scrB_1(\scrH)$ by 
\beq
\caP \psi =     ( \Tr_\res \psi)  \otimes \initialresfinite,  \qquad  \psi  \in  \scrB_1(\scrH)
\eeq
 We start from
 \baq
 B(1 )\caP \,  \initialgibbsfinite  
&=&      -     B(1 ) (1 -\caP )  \initialgibbsfinite -    T  \initialgibbsfinite   
 + \e^{\i t_0 L_\res}   \initialgibbsfinite \eaq
where we used $B(1)=\e^{\i t_0 L_\res} \caU-T$ and $
\initialgibbsfinite = \caU\initialgibbsfinite $.   Next, we substitute
$$
\initialgibbsfinite =   \caU(0)\initialfinite=\caU(0)(\rho^{\reff}_\sys  \otimes \initialresfinite)=(T(0) + B(0))
(\rho^{\reff}_\sys  \otimes \initialresfinite)
$$ 
and use $\caP \initialgibbsfinite =  \rho^{\be}_\sys \otimes \initialresfinite$   to obtain
 \baq
 B( 1  )    (\rho^{\be}_\sys  \otimes  \initialresfinite) 
&=& (  -   B(1)(1 -\caP)B(0)    -        TB(0)+\e^{\i t_0 L_\res}B(0))(
\rho^{\reff}_\sys  \otimes \initialresfinite)\nonumber\\[2mm]
&+&(\e^{\i t_0 L_\res}T(0)
  -    TT(0)  )(
\rho^{\reff}_\sys  \otimes \initialresfinite)
             \label{eq: terms to vanish}  
\eaq
This relation has the form $\sum_j D_j (\psi_{\sys,j} \otimes \initialresfinite)=0$, for some $\psi_{\sys,j} \in \scrB_1(\scrH_\sys)$ and $D_j $  operators on $\scrB_1(\scrH)$, hence it implies that
\beq
\sum_j \bbE(D' \odot D_j)   \psi_{\sys,j}=0, \qquad \textrm{for any} \,\,   D'     \label{eq: abstract ward}
\eeq
where we used the convention \eqref{eq: convention acting on only one function}.  We choose $D' = B(\tau)$ and we spell out \eqref{eq: abstract ward}, obtaining 
\baq
\bbE ( B( \tau) \odot B( 1)   ) \,   \rho^{\be}_\sys 
&=&  -  \bbE (  B( \tau) \odot B(1) B(0) )  \rho^{\reff}_\sys   \nonumber  \\[3mm]
&\,\,  &  - \bbE (  B( \tau) \odot T B(0)  )   \rho^{\reff}_\sys +  \bbE ( B( \tau-1) \odot B(0) )  \rho^{\reff}_\sys   \label{eq: abstract ward concrete}                                
\eaq
where the two first terms in \eqref{eq: terms to vanish} containing no $B(\cdot)$-operator have disappeared because $\bbE(B(\{ \tau\}))=0$, and  we used $\bbE ( B(\tau) \e^{\i t_0 L_\res}\odot  \ldots )=  \bbE (B(\tau-1)\odot  \ldots)$.  
The lemma follows upon multiplying with $\str \La \str$ and reinstating the subscript $n$.

\end{proof}

\subsubsection{Ward identity in terms of correlation functions $G^{c}_{A,n}$}

As it stands, Lemma \ref{lem: correlation equilibrium}  is not written in terms of the correlation functions $G^c_{A}$. We can however easily rewrite it in that way. Since the Ward identity compares operators in different copies of $\scrR_n$, it is unnatural to index the $z,z'$-coordinates by $\tau \in A$ and hence we use arbitrary indices, with the convention that the $z,z'$-coordinates corresponding to earlier times stand to the left of the later ones. We again drop the index $n$.

The Ward identity reads
\beq  \label{eq: ward identity with L}
\int \d z_a   \, G^{c}_{\{1, \tau\}}  ( z_a',{ z_b}' ;     z_a,  {z}_b) \nu^\be(z_a)  =      \int \d z_a L_{ \tau}  (  z_a', { z_b}';     z_a, {z_b}) \nu^{\reff}(z_a)   
\eeq
with
\baq
L_{ \tau}  (z_a', z_b' ; z_a,  z_b)   & := &       - \int \d z \,   G^{c}_{\{0,1, \tau\}}( z, z_a', z_b';  z_a, z, z_b  ) \nonumber    \\[2mm]
& -&      \int \d z  \,  T(z_a' ,z  ) G^{c}_{\{0, \tau\}}  ( z, z_b' ;  z_a,  z_b)  +          G^{c}_{\{0, \tau-1 \}}  ( z_a',  z_b';  z_a,  z_b) \label{Ldefi}
\eaq
The verification of this relation from Lemma \ref{lem: correlation equilibrium} is by inspection, using the fact that $G_{A}=G^{c}_{A}$ whenever $\str A \str \leq 3$, which in turn follows from the vanishing of  $G_{A}$ whenever $\str A\str=1 $.

\vskip 2mm

\noindent {\bf Remark.} The reader should think of eq.\ (\ref{eq: ward identity with L}) as analogous
to the unitarity Ward identity, eq.\ (\ref{eq: unitarity ward in terms of kernels}). Namely, the
RHS turns out to run down to zero fast under the RG (i.e.\ as $n \to \infty$) and this lets us get extra contraction
for the RG flow of the two-point function.

 \section{The infinite volume setup}\label{sec: infinite volume setup}
 Up to this point, our whole treatment was restricted to finite volume $\bbZ^d/L\bbZ^d$. This allowed us to neglect all sorts of analytical questions since all the operators pertaining to the $\sys$-part  were actually finite matrices (this is not true for operators on $\scrH_\res$, though).  In this section, we indicate which quantities remain meaningful in the thermodynamic limit and we give their precise definition. 
 Note that the lattices $\bbX_n, \bbA_n$ remain meaningful with $L = \infty$, as was already indicated in Section  \ref{sec: kernels and rescaling of space}.
 However, the spaces $\scrR_{n,A}$, defined as tensor products in finite volume, are no longer uniquely defined and we will choose a norm to specify them in  Section \ref{sec: norms}.   
  \subsection{Thermodynamic limits}

 Let us now indicate explicitly the volume dependence in the evolution operators $ T^{\La}_{n},  T^{\La}_{n}  (0)$, the correlation functions $  G^{\La}_{n,A},  G^{c,\La}_{n,A}$ and the 'unnormalized states' $\nu^{\be, \La}_n, \nu^{\mathrm{ref}, \La}_n $. Those objects were previously simply denoted by the same symbols without $\La$ but now we reserve this notation for infinite-volume quantities. 
 The thermodynamic limit is then defined by point-wise convergence of kernels (well-defined since  $\bbA^{\La}_n$ is naturally a subset of $\bbA_n$) as $\La \nearrow \bbZ^d$ (recall that this is shorthand for $L \to \infty$ since $\La=\La_L$).
  Recall the assumption \eqref{eq: assumption free correlation functions thermo}, it leads to
  \begin{lemma} 
The kernels of the operators
 $ T^{\La}_{n},  T^{\La}_{n}  (0)$, the correlation functions $  G^{\La}_{n,A},  G^{c,\La}_{n,A}$ and the 'unnormalized states' $\nu^{\be, \La}_n, \nu^{\mathrm{ref}, \La}_n $ converge pointwise as $\La \nearrow \bbZ^d$. 
 The limiting kernels satisfy the recursion relations   \eqref{eq: T from connected correlation functions},    \eqref{eq: from n to npluseen}  and the Ward identities \eqref{eq: unitarity ward in terms of kernels} and  \eqref{eq: ward identity with L}.  In particular, the sums on $\bbA_n$ implicit in the RHS of these recursion relations are absolutely convergent. 
 \end{lemma}  
This lemma follows directly from the bounds of Lemma \ref{lem: bound on h} and (for $\tau=0$)  those of Section \ref{sec: bounds on boundary correlation functions}.

\subsection{Norm on kernels} \label{sec: norms}
 \subsubsection{Tensor product completions} \label{sec: tensor completion}
We have introduced the space $\scrR_{n}$ (on scale $n$) for operators acting on (rescaled) particle density matrices. As long as the volume $\La$ is finite, this space is finite-dimensional and hence isomorphic to $ \scrB(\scrB_1(\scrH_\sys))$.  Similarly, the tensor product spaces 
$\scrR_{A,n}$ are finite dimensional and thus all norms are equivalent. 
However, as $\La \nearrow \bbZ^d$, care is needed. Even if we choose a norm for $\scrR_{n}$ 
there are several ways to complete the algebraic tensor products. Our choice is dictated by
the recursion relation (\ref{eq: from n to npluseen}) which we want to be continuous in the norm.
The norms are of course meaningful for $\La$ finite as well, and in the definition below we will not distinguish different values of $\La$.  

Recall that an element of $\scrR_{n,A}$ is represented by the kernel $K_A= K_A (z'_A, z_A)$  with $A=\{ \tau_1, \ldots, \tau_m \},\tau_i < \tau_{i+1}$ and $z_A=(z_{\tau_1},\dots , z_{\tau_m}) \in \bbA_{n,A}$. 
Recall also that $z_\tau=(x_\tau,s_\tau)$. Since the scale subscript $n$ does not play any role as far as the definition of the norm is concerned, we suppress it in what follows and we will simply write $\scrR, \scrR_A, etc..$. 

 We need to distinguish in the norm the $x$ and
the $s$ variables and so will view $\scrR$  as the tensor product
$\scrR^x\otimes \scrR^s$ where $K\in\scrR^x$ is represented by the
kernel $K(x',x)$ and $K\in\scrR^s$  by $K(s',s)$. 
 Introduce the
following $L^1-L^\infty$ norm in the spaces $\scrR^x$ and $\scrR^{s}$
\beq
\|K\|:= \max \left(\sup_{y'}\int \d y(|K(y',y)|,  \,  \sup_{y}\int \d y' |K(y',y)|)  \right)
\label{eq: llnorm}
\eeq
where $y$ stands either for $x$  if $K \in \scrR^x$ and  for $s$ if $K \in \scrR^s$. We use and recall the conventions introduced in Section \ref{sec: kernels and rescaling of space}: $\int \d x$ denotes the sum over $\bbX_n$ together with a scaling factor and $\int \d s$ is the sum over $\caS$ (without any scaling factor), hence $\int \d y$ can be either of these two 'integrals'. This distinguished use of the symbol $y$ is only in the present section, i.e.\  Section \ref{sec: tensor completion}.
The space $\scrR_{A}$ can be viewed as the tensor product
$   \scrR^x_{ \tau_m} \otimes  \scrR^s_{ \tau_m} \otimes    \ldots \otimes  \scrR^x_{ \tau_1} \otimes  \scrR^s_{ \tau_1}$. It is however simpler to consider more generally kernels  $K\in\otimes_{i=1}^m\scrR^{\sigma_i}$ where ${\sigma_i}\in\{x,s\}$, hence not necessarily with the same number of $x$- and $s$-coordinates (Note that we drop the $\tau$-labels since what follows is  independent of the labelling of tensors).    We now define the norm  $\|K\|$ on such kernels,  inductively in the degree $m$. 

Write $K=K(\underline{y}', \underline{ y})$ with $\underline{ y}= (y_1,\dots,y_m)$ and $y_i$ equals $x_i$ or $s_i$ depending on $\sigma_i$, and analogously for $\underline{y}'$.
For $1\leq i\leq m$, fix ${y_<}=(y_1,\dots,y_{i-1})$ and ${y_>}=(y_{i+1},\dots,y_{m})$
and 
let $K^{(i)}  \in\scrR^{\sigma_i}$ be the restriction of $K$ to $\scrR^{\sigma_i}$ with ${y_<}$
and ${y_>}$ kept fixed. Then $K_{(i)}:=\|K^{(i)} \|$ is a kernel of degree $m-1$.  We define the norm inductively in the degree of kernels by setting
$$
\|K\|:=\max_i\|K_{(i)} \|.
$$
The virtue of this norm is that it behaves well under tensor products
and contractions (the operator $\caT$ introduced in Section \ref{sec: correlation functions}), the basic building blocks of the the recursion relation (\ref{eq: from n to npluseen}).
We define the contraction $\iota:\scrR^{\sigma}\otimes \scrR^{\sigma}\to \scrR^{\sigma}$ by 
\beq
(\iota K)(y',y)=\int \d\tilde y \, K(  \tilde y, y';  y, \tilde y)
\label{eq: iotadef}
\eeq
We extend this map to $\otimes_{i=1}^m\scrR^{\sigma_i}$  in the natural way:
for $i < j$ s.t. $\sigma_i=\sigma_j$ let $\iota_{ij}$ act in the $i$:th and the $j$:th
factors  $\scrR^{\sigma_j}\otimes \scrR^{\sigma_i}$.    Finally, we write $\str K \str$ for the operator with kernel $\str K \str (y',y) :=\str K (y',y) \str$.  

Points 1) and 2) of the following lemma  are used throughout, point 3) is necessary only  in the proof of Lemma \ref{lem: bound on e}.

\begin{lemma} \label{lem: property of norm}
Let $K$ and $L$ be kernels as above. Then
\ben
\item 
\beq
\|K\otimes L\|\leq \|K\|\|L\|
\label{eq: norm of product}
\eeq
\item Let  $\si_i=\si_j$ such that $\iota_{ij}$ is defined, then 
\beq
\|\iota_{ij}K\|\leq \|K\|
\label{eq: norm of contraction}
\eeq
\item Let $K$ be of degree $m$ with $\si_1=\si_2=\ldots=\si_m$, then, for any  $i \in \{1, \ldots,m \}$,
\beq
\sup_{y,y'} \str \iota_{12}\iota_{23} \ldots \iota_{m-1, m} K \str   \leq  \norm  \sup_{y_i, y'_i} \str K \str  \norm 
\eeq
where  on the RHS the norm is applied to a kernel of degree $m-1$. 
\een
\end{lemma}
\begin{proof} For the first claim, we proceed by induction in the degree of $K\otimes L$.
Obviously $(K\otimes L)_{(i)}=K_{(j)} \otimes \str L\str$ if $i=\textrm{degree}(L)+j$, and $(K\otimes L)_{(i)} = \str K\str \otimes  L_{(i)}$ if $i \leq \textrm{degree}(L) $.  Thus
$$
\|K\otimes L\|\leq \max_{ij} (\|K_{(i)} \otimes  L  \|, \|  K  \otimes  L_{(j) } \|)\leq 
 \max_{ij} (\|K_{(i)} \|\| L\|, \|K\|\|  L_{(j)} \|)=\|K\|\|L\|
$$
where we used induction in the second step. 

For the second claim, the inductive definition of our norm means it suffices to
check the claim for $K$ of degree two. Thus consider the first term in
(\ref{eq: llnorm}) for $\iota K$:
$$
\sup_{y'}\int \d y \left\str\int  \d\tilde y \, K(\tilde y, y'; y, \tilde y )) \right\str\leq
\sup_{y'}\int \d\tilde y\int \d y \, |  K(\tilde y, y' ;y, \tilde y) |\leq 
\sup_{y'}\int \d\tilde y\sup_{y''}\int \d y \,  |K(y'', y' ;y, \tilde y)|
$$
and the expression on the right is bounded by  $\norm K_{(1)} \norm $.  By analogous reasoning, the second term in (\ref{eq: llnorm}) is bounded by $\norm K_{(2)} \norm $ and hence we have indeed $
\|\iota K\|\leq \norm K \norm $.

To get the third claim, it is sufficient to check it for $m=2$ and for $m=3$ with $i=2$, since all other cases can be reduced to one of these.  
We check explicitly the case $m=3$ with $i=2$.  The LHS is estimated as
\baq
\sup_{y,y'} \int \d \tilde y \int  \d \breve y  \str K( \tilde y,\breve y,y'; y,\tilde y,\breve y )) \str & \leq &   \sup_{y,y'}   \int \d y''' \sup_{\tilde y} \int  \d \breve y \sup_{y''}  \str K(y''', y'',y'; y,\tilde y,\breve y )) \str  \nonumber \\[1mm]
& \leq &  \sup_{y}   \int \d y'''  \sup_{y'}\int  \d \breve y \sup_{\tilde y, y''}  \str K(y''', y'',y'; y,\tilde y,\breve y )) \str    \nonumber 
\eaq
and one now verifies easily that the result is indeed smaller than the RHS.   The $m=2$ case follows by similar, though simpler reasoning. 
\end{proof}

  \subsubsection{Definition of the norm  $\norm \cdot \norm_\ga$} \label{sec: definition of norm}
  
We will modify the norms introduced above to encode information about
the decay properties of the kernels. Thus we will define a norm $
 \norm K_A \norm_{ \ga} $ that depends on two parameters: $\ga_0$ (giving the decay in the $\offdx-\offdx'$ coordinate, and $\ga$ given the decay in the $x -x'$  coordinate.  
Note that the parameter $\ga_0$ is not included in the notation for the norm, since we will never need to consider a different one.  The parameter $\ga$ will equal a multiple (of order $1$) of $\ga_0$. 
Given a kernel $K_A \in \scrR_A$, let
    \beq \label{def: Kgamma}
 K_A^\gamma (z_A', z_A) :=        K_A(z_A', z_A)   \e^{\ga \sum_{\tau \in A }\str x'_\tau- x_\tau \str} \e^{\ga_0 \sum_{\tau \in A }\str \offdx'_\tau- \offdx_\tau \str}
\eeq
where we  wrote $z_{\tau}=(x_\tau,s_\tau)$ and $s_\tau=(v_\tau,\eta_\tau, (e_\links)_\tau, (e_\rechts)_\tau )$. Then we set
   \beq \label{def: gammanorm}
 \norm K_A \norm_{ \ga} :=\| K_A^\gamma\| .
\eeq
Note that
  \beq \label{gammatensor}
( K_A\otimes K_B)^\gamma= K_A^\gamma\otimes K_B^\gamma.
\eeq
Recall the contraction $\iota_{i,j}$ introduced above, with each of the indices $i,j$ corresponding to a pair $(\tau,\si)$ with $\tau \in A$ and $\si=x,s$. We define now (contractions on the RHS are as in the previous section)
\beq
\iota_{\tau,\tau'} :=  \iota_{(\tau,x),(\tau',x)} \iota_{(\tau,s),(\tau',s)} = \iota_{(\tau,s),(\tau',s)}  \iota_{(\tau,x),(\tau',x)}
\eeq

Starting from (\ref{eq: iotadef}) and using the triangle inequality $|x'-x|\leq |x'-\tilde x|+|\tilde x-x|$ and similarly for $v,v',\tilde v$, 
we get for point-wise {\it nonnegative} kernels $K_A$
  \beq \label{gammaiota}
(\iota_{\tau \tau'} K_A)^\gamma\leq \iota_{\tau \tau'} K_A^\gamma.
\eeq
These observations lead to
\begin{lemma} \label{lem: main property of norms}
\ben
\item Let $\caA$ be a collection of disjoint sets $A \subset \bbN_0$ such that $I_{A'} = \supp \caA$ (see Section \ref{sec: recursion relations for correlation functions}) and let for all $A \in \caA$ an operator $K_A \in \scrR_A$ be given, then
\beq
\left\norm \caT_{A'} \left[ \mathop{\otimes}\limits_{A \in \caA}  K_A       \right]     \right\norm_\ga   \leq   \prod_{A \in \caA}  \norm K_A \norm_\ga
 \label{Tbound}
\eeq
\item Let $K_A \in \scrR_A  $. Then
 \beq
\norm \bsS_\ell [K_A]  \norm_\ga  =     \norm K_A  \norm_{\ga/\ell},    \label{eq: scaling property of norms}
\eeq
and 
$$\norm K_A \norm_\ga \leq \norm K_A \norm_{\ga'}$$
 for $\ga < \ga'$.
 \een
\end{lemma}
\begin{proof}  For 1) we bound  (\ref{Tbound}) from above by
replacing $K_A$ by $|K_A|$ and then use  (\ref{gammatensor}) and  (\ref {gammaiota}). Then, since  $\mathcal T$ is a product of contractions $\iota_{i,j}$, Lemma \ref{lem: property of norm} 2) applies and we get the claim.
For the first claim in 2) observe that the norm (\ref{eq: llnorm}) is invariant under
$ \bsS_\ell $ and the effect in $\gamma$ is the stated one. The second claim is obvious.
\end{proof}
 \vskip 2mm
\noindent{\bf Remark}  From now on, the space $\scrR$ is meant to be equipped with the norm $\|\cdot\|_\gamma$.
The exponential factors in our norm (provided $\ga >0$) imply in
particular that kernels $K\in \scrR$ map the space
 $$ \scrB_p(\scrH_\sys) = \left\{ V \in \scrB(\scrH_\sys) \big\str \, \Tr (V V^*)^{p/2}  < \infty  \right\} $$
  into itself, for any $1 \leq p < \infty$.

\subsubsection{Definition of the norm  $\norm \cdot \norm_{\banone}$} \label{sec: reduced kernels}

When studying phenomena like diffusion that are intimately connected with the translation invariance of the system, we will mainly need to consider the coordinates $x, x'$ (see Section \ref{sec: translation invariance} on translation invariance). 
Therefore, we often want to consider $K \in \scrR= \scrR^x \otimes\scrR^s$ as a kernel in $x, x'$ but taking values in the 'internal space' $\scrR^s$ which for short
will be called  $\banone$.  $\banone$ is completed with the norm 
\beq
\norm F \norm_{\banone}  := \max{\left(
\sup_{s'}\sum_s |F(s',s)| \e^{\gamma_0|v-v'|},  \sup_{s}\sum_{s'}|F(s',s)|)\e^{\gamma_0|v-v'|} \right)}
% \max \left\{\begin{array}{l}     \sup_{s} \sum_{ s'}  F(s, s')  \e^{\ga_0 \str \offdx- \offdx' \str }   \\[2mm]
    %  \sup_{s'} \sum_{ s} F(s, s')  \e^{\ga_0 \str \offdx- \offdx' \str }    \end{array} \right. 
\eeq
 Obviously, the norm $\norm \cdot \norm_{\banone}$ relates well to $\norm \cdot \norm_\ga$ on $\scrR$. Given  $K\in\scrR$,
\beq
 \max{\left( \sup_{x'}\int dx  \norm  K(x',x) \norm_\banone  \e^{\ga \str  x-  x' \str },    \sup_{x}\int dx' \norm  K(x',x) \norm_\banone
 \e^{\ga \str  x-  x' \str }  \right)}
  \leq   \norm K \norm_{\ga}  .
 %  \max \left\{\begin{array}{l}     \sup_{x} \sum_{ x'}  \norm  K(x,x') \norm_\banone  \e^{\ga \str  x-  x' \str }   \\[2mm]   \sup_{x'} \sum_{ x} \norm K(x,x') \norm_\banone  \e^{\ga \str  x-  x' \str }      \end{array} \right \}  \leq  2 \norm K \norm_{\ga}  
  \label{eq: banone bounded by weird}
\eeq
Indeed, looking back at Section \ref{sec: tensor completion} and taking $\sigma_1=x$ and $\sigma_2=s$, the LHS
equals $\|K^\gamma_{(2)} \|$ whereas the RHS equals $\max_{i=1,2}\|K^\gamma_{(i)} \|$.

\subsection{Symmetries} \label{sec: symmetries}  
\subsubsection{Translation invariance} \label{sec: translation invariance}

We implement spatial translations on $\scrH_\sys= \smallspace \otimes l^2(\bbZ^d)$ by the operators $V_u \in \scrB(\scrH_\sys), u \in \bbZ^d$;
\beq
(V_u \psi)(y) : =\psi(y+u) \in \smallspace, \qquad  \psi \in \scrH_\sys
\eeq
Write  in general 
$\Adjoint(S)O=  S O S^{-1}$ for $O,S \in \scrB(\scrH_\sys)$, $S$ invertible. 
Then $\Adjoint (V_u)$ implements translations on density matrices. In $\bbA_0$-coordinates, it acts as  
\beq \Adjoint (V_u)\rho(x,s)= \rho(x+u,s)   \label{eq: translations on densities}  \eeq
We will also denote by $\Adjoint (V_u)$ the operator on $\bbA_n$ defined by \eqref{eq: translations on densities}, now with $x,u \in \bbX_n$. 
It is then natural to call  $K \in \scrR= \scrR_{n}$ translation invariant if it commutes with translations, i.e. 
\beq
 V_{u} K  V_{-u} = K, \qquad   \textrm{for}\,  \, u \in \bbX_n
\eeq
 For such a translation invariant operator $K$, we will often abbreviate the reduced kernel 
 \beq
 K(x'-x) =   K(x',x) \qquad  (K(x',x) \in \scrG)
  \eeq
  and we have 
\begin{lemma} \label{lem: banone equal to weird}
 Let $K \in \scrR_n$ be translation-invariant, then
\beq    \label{eq: banone equal to weird} 
  \norm K \norm_{\ga}  =   \int_{  \bbX_n} \d x   \norm K (x)  \norm_{\banone}   \e^{\ga \str  x \str } 
\eeq
\end{lemma}
\begin{proof} The inequality $ \geq $  follows from \eqref{eq: banone bounded by weird}. 
 Following the discussion after \eqref{eq: banone bounded by weird}, in order to get $\leq$, we need to
bound $\|K^\gamma_1\|$ by $\|K^\gamma_2\|$. This follows since the supremum over the $x$ or $x'$ coordinate can be dropped because of translation invariance and hence we get an upper bound by moving the supremum over the $s$ or $s'$ coordinate to the right.   \end{proof}

\subsubsection{Fourier Transform} \label{sec: fourier transform}

The translation invariance suggests to Fourier transform the kernel in the variable $x \in \bbX_n$.
The dual space to the lattice $\bbX_n$ is the torus 
\beq
\bbT_n :=    (\ell^{n} \bbT)^d
\eeq
We define for $p \in \bbT_n$, 
\beq
\hat\rho(p,s) :=  \int_{\bbX_n} \d x \,   \e^{\i p x  }  \rho(x,s), \qquad  \text{for}\, \,\rho \in \ell^1(\bbA_n),
\eeq
and, for a translation-invariant $K \in \scrR_n$, 
\beq
  \hat{  K} (p)  :=     \int_{\bbX_n} \d x \,   \e^{\i p x  }   K(x), \qquad  (  \hat{  K} (p)  \in\banone).
\eeq
where the sum on the RHS is absolutely convergent if $\norm K \norm_\ga $ is finite for some $\ga\geq 0$.
It follows that, if we have two translation invariant kernels $K_1, K_2$ taking values in $\banone$, then 
\beq  
 \widehat{  K_1  K_2}(p) =    \hat{   K} _1 (p)  \hat{  K} _2(p) 
\eeq
where the product on the LHS is in $\scrR_n$ and on the RHS in $\banone$. The following
standard consequence of exponential decay will be used throughout. If  $K$ is a translation-invariant kernel then
\beq
  \sup_{\Im p \leq \ga}  \norm \hat K (p) \norm_{\banone}      \leq  \norm K \norm_{\ga}  
\eeq
%and conversely
%\beq
% \norm K \norm_{\ga}     \leq   \sup_{e_p \in \bbT_n, \str e_p \str=1}  \int_{ \bbT_n +\i \ga e_p}  \d p   \norm \hat K (p) \norm_{\banone}    
%\eeq
% 
%\end{lemma}

\subsubsection{Symmetries of the lattice}  \label{sec: symmetries of the lattice}

Let us investigate the action of a lattice symmetry $O$ on our operators.  We consider a density matrix $\rho$, and we abbreviate the transformed density matrix as   $\rho_O \equiv \Adjoint(V_O) \rho $,  such that  
\beq
\rho_O (x_\links, x_\rechts)  = \rho( Ox_\links,  Ox_\rechts)
\eeq
To write this transformation in our new coordinates $x, s$,
we let $x_{O}, \offdx_{O}, \offsetx_{ O}$ be the coordinates corresponding to $Ox_\links,  Ox_\rechts$ (we suppress the coordinates $e_\links, e_\rechts$ since they are untouched by $O$). 
Define the vector $\hat e$ by
\beq
\hat e_j = \left\{\begin{array}{lll}
 0 &  \textrm{if} &(O \offsetx)_j \in \{0,1\} \\[2mm]
1  & \textrm{if} & (O \offsetx)_j =-1
\end{array}\right.
\eeq
Then, one checks that
\begin{align}
x_{O} = O x -\hat e, \qquad   \offdx_{O} = O \offdx - \hat e, \qquad  \eta_{ O} = O \offsetx + 2 \hat e
\end{align}
Hence  $\rho_O (x, \offdx, \offsetx)= \rho (Ox -\hat e , O\offdx -\hat e , O\offsetx +2\hat e )$. 
We compute the  fourier transform of $\rho$ in the variable $x$.
\beq
\hat \rho_O(p, \offdx, \offsetx) =     \e^{\i p \hat e}\hat{\rho}(O^{-1} p,  O\offdx  -\hat e , O\offsetx +2\hat e )  =  \left(I_{p,O} \hat{\rho}(O^{-1} p, \cdot, \cdot) \right)(\offdx, \offsetx)
\eeq
where $I_{p,O}$ is an invertible transformation acting on the degrees of freedom $\offdx, \offsetx$.

\section{Convergence to a fixed point} \label{sec: convergence to fixed point}

We now state the induction hypotheses for the RG flow of the reduced
dynamics $T_n$ (Proposition \ref{ass:  properties of T}), and  the correlation functions $G^{c}_{n,A}$ (Proposition \ref{prop: overview b behavior}).  Their validity for all $n$ implies our main results, as we show below in Section \ref{sec: proof of main theorems}.  The inductive proof of these induction hypotheses is postponed to Sections  \ref{sec: flow of t}-\ref{sec: estimates on the first scale reduced evolution}  and we provide a brief guide to the proof in Section \ref{sec: plan of proof}.

\subsection{Induction hypotheses}\label{sec: induction hypotheses}

Let us first discuss constants and small parameters of our theory.   The spin Hamiltonians  $H_\spin$, spin-phonon interaction $W$,  particle mass $m_{p}$, as well as
the environment correlation function $\zeta$ and temperatures $\beta,\beta_i$ are considered fixed and
constants depending only on them are denoted by $C$ (large constants) and $c$ (small constants).  The adjustable parameter in the
problem is the coupling $\la$ and in the proof also the RG scaling factor $\ell$ and $\frt_0$ entering the initial time scale $t_0= \la^{-2}\frt_0$.

We define the exponents
\beq
\tilde \al :=\left\{ \begin{array}{cc} (2 \al-1)/4  &  1/2 <  \al < 1  \\[1mm]    (4 \al-1)/8  &    1/4 <  \al \leq 1/2   \end{array}\right., \qquad  \tilde \alpha_{\ini} :=\left\{ \begin{array}{cc} \text{not defined}  &  1/2 <  \al < 1  \\[1mm]    \al/4  &    1/4 <  \al \leq 1/2   \end{array}\right.   \label{def: tilde alphas}
\eeq
with
$\al$ the correlation decay exponent in
Assumption \ref{ass: decay micro alpha}.  The exponent $\alpha_{\ini} $ ($\ini$ for 'initial') appears in the treatment of correlation functions $G^c_A$ with $A \ni 0$ (`boundary' correlation functions; those were defined only in the case $\be_1=\be_2$) and we define this exponent only for $\al \leq 1/2$. 
Then, we introduce the 'running coupling constants'
\beq  \label{eq: values running epsilons}
\ep_n   :=    \ell^{-n \tilde \al} \ep_0, \qquad  \ep_0 :=   C\str\la\str^{\al} \eeq
\beq  \label{eq: values running epsilons initial}
    \ep_{\ini,n}   =    \ell^{-n  \tilde \alpha_{\ini}} \ep_{\ini,0}, \qquad  \ep_{\ini,0} =   C\str\la\str^{2-2\al}
\eeq

In the Sections \ref{sec: flow of t},  \ref{sec: linear rg flow} and \ref{sec: nonlinear rg flow} the $\la$ appears only through  $\ep_0, \ep_{\ini,0}$, hence the latter can be considered  the fundamental small parameters. 

Our basic convention is that 
 $\ell$ is chosen large enough compared to 
the fixed parameters and we  will for example freely assume that $\ell^{-1} C<1$.  The coupling constant $|\la|$ (or  $\ep_0, \ep_{\ini,0}$) is then chosen small compared to $\ell$, such that we can freely assume that $\str\la\str C(\ell) <1$ for quantities $C(\ell)$ depending on $\ell$ but not on $\la$.
 Finally, unless otherwise stated all constants are uniform in $n$, the RG
iteration.

%The following Proposition proves convergence of the rescaled reduced dynamics $T_n$ 
%to a "fixed point". 
%Recall that the RG iteration involves two parameters, the
%kinetic time scale $ \frt_0^{-1} \la^2$ of the "0:th step" and the rescaling
%parameters $\ell$. Both will be taken large enough to kill constants and this
%is possible by taking the coupling $\la$ small enough. Moreover for later purposes it
%is useful to prove the result for a range of $ \frt_0$. 

\subsubsection{Induction hypothesis for $T_n$}

The first induction hypothesis concerns the reduced evolution $T_n$ and $T_n(\tau=0)$ of \eqref{0operators}.
The hypothesis involves 
parameters $\bana_0$,  %$\smallgapn_0$,   $\largegapn_0$,  
 $\gamma_0$ and $D_0$ that will be fixed in Section   \ref{sec: estimates on the first scale reduced evolution}; they result from  the weak coupling analysis.

\begin{proposition}[Induction hypothesis for $T_n$] \label{ass:  properties of T}
 There exist $\tau_0$, $\ell_0<\infty$ and $\la(\ell)>0$ such that for $\ell>\ell_0$ and $0 < |\la|\leq\la(\ell)$ the following holds,
uniformly in $\frt_0\in[\tau_0,2\tau_0]$ and in $n$.

\ben
\item \textbf{Analyticity}. The operators $\hat T_n(p)$ and $\hat T_n(0,p)$ ('$0$' refers to $\tau=0$) are analytic in a strip of width $\gamzero$
with
\beq  \label{eq: analytic bounds t in prop}
\sup_{\str\Im p \str \leq \gamzero }  \norm \hat T_n(p) \norm^{}_{\banone}  \leq C, \qquad  \sup_{\str\Im p \str \leq \gamzero } \norm \hat T_n(0,p) \norm^{}_{\banone}      \leq C  
\eeq
Moreover,   $\hat T_n(0,p)= \hat T_{n-1}(0,p/\ell)$ for $n \geq 1$.
\item  \textbf{Diffusion}.  Let  $\str \Re p \str < \bana_n$ given below and $\str\Im p \str \leq \gamzero$. The operators $\hat T_n(p)$ have a simple eigenvalue\footnote{For example, consider $\hat T_n(p)$  as an operator on $l^{\infty}(\caS)$, see the discussion in Appendix \ref{app: spectral perturbation theory}} $\e^{f_n(p)}$,i.e.\ 
\beq
  \hat T_n(p)  R_n(p)  =  R_n(p)  \hat T_n(p)  =    \e^{f_n(p)}  R_n(p)%, \qquad   \str \Re p \str < \bana,
\eeq
where  the one-dimensional spectral projector $R_n(p)$ is bounded as
\beq
\norm R_n(p)  \norm_{\scrG} \leq C  \label{eq: hyp bound on r}.
\eeq
and the complementary part as
\beq
\norm (1-R_n(p))  \hat T_n(p)  \norm^{}_{\banone}
\leq\hf\ell^{-\frac{\tilde \al}{8} n}=:\largegapn_n ,
 \label{eq: smallgap}  
\eeq  
There is a diffusion constant $D_n>0$ such that the analytic function
$f_n(p)$ satisfies
\baq
&&   \str   f_n(p) + D_n p^2  \str <   \ell^{-\frac{\tilde \al}{4} n}  \str p\str^3 
 \label{eq: fclosetoparabola}
\eaq
Moreover,  for $n\geq 1$, 
\beq
\norm  R_{n}(p)-  R_{n-1}(p/\ell)  \normba  \leq   \sqrt{\ep_0}  \ell^{-\frac{\tilde \al}{4}(n-1)} , \qquad    \left\str D_{n}  - D_{n-1} \right\str  \leq   \sqrt{\ep_0}  \ell^{- \frac{\tilde \al}{4}(n-1) }
\label{eq: resulting steptostep}
\eeq 
The constant $\bana_n$ is given by
 \beq\label{banandef}
e^{- \frac{1}{2}  D_n \bana_n^2}=\ell^{-\frac{\tilde \al}{4}n}e^{- \frac{1}{2}  D_0 \bana_0^2}.
\eeq
and
 $\bana_0\leq D_0/2$. The spectral projection $R_n(0)$ is given by
 \beq  \label{eq: projector r explicit}
R_n(0) =  \str \mu_{T_n} \rangle \langle 1_{\caS_{0}} |
\eeq
(notation as in Section \ref{sec: ward identity from unitarity}) where $\mu_{T_n}$ is a function on $\caS$ whose restriction to $\caS_0$ defines a probability measure: $ \sum_{s \in \caS_{0}}  \mu_{T_n} (s)=1$ and
$  \mu_{T_n} (s) \geq 0, s \in \caS_0$. Moreover in the equilibrium case $\beta_1=\beta_2=\beta$
 \beq  \label{eq: projector r explicit bound}
\norm R_n(0)- |\mu^\beta\rangle\langle 1_{\caS_0}|\normba\leq C\sqrt{\ep_0}
\eeq
with $\mu^\be$ as in Section \ref{sec: equilibrium states}.
\item \textbf{Gap}. Let $\str \Re p \str \geq \bana_n$ and $\str\Im p \str \leq \gamzero$.  Then% for some $\largegapn_0<1$
\beq
 %\sup_{\str \Re  p \str > \bana, \str\Im p \str \leq \gamzero } 
 \norm \hat T_n(p)  \norm^{}_{\banone}    \leq
%\hf \ell^{-\frac{\tilde \al}{4} n}
  \largegapn_n %:=\ell^{-\frac{\tilde \al}{4} n}\largegapn_0
  \label{eq: largegap} 
\eeq

\item  In position space we have 
\beq
\norm T_n (x)  \norm_{\banone}    \leq C \e^{-\tengam \str x \str }
\label{tpos}
\eeq
\een
\end{proposition} 

\noindent For later use we note that (\ref{eq: fclosetoparabola}) together with $ \bana_0\leq D_0/2$ implies
\beq\label{upper and lower bound for ef}
e^{- \frac{3}{2}  D_n (\Re p)^2-C(\Im p)^2}\leq |e^{f_n(p)}|\leq
e^{- \frac{1}{2}  D_n (\Re p)^2+C(\Im p)^2}
 .
\eeq

%The proof of Proposition \ref{prop: overview t behavior} is inductive: we prove in tandem
%a result (Proposition \ref{prop: overview b behavior}) showing  that the renormalized noise correlation functions $G^{c}_{n,A}$ 
% vanish as $n\to\infty$ and a more detailed representation of $T_n$  (Proposition \ref{ass:  properties of T}. 

\subsubsection{Induction hypothesis for $G^c_{n,A}$}
We move to the correlation functions  $G^{c}_{n,A}$. First, we define a distance-like function on sets $A$ of times  as follows: assume that $A=\{\tau_1, \ldots, \tau_m \}$ with $\tau_1 < \ldots < \tau_m$, then 
\beq
\dist(A)= \dist(\tau_1, \ldots, \tau_m) :=  \prod_{j=1}^{m-1}   (1+ \str \tau_{j+1}- \tau_{j} \str)
\eeq
and we will always use  $\dist(A)^{\al} = (\dist(A))^{\al} $ to quantify the decay in time of the operators $G^{c}_{n,A}$, with $\al$  as in Assumption \ref{ass: decay micro alpha}. 
Since  $G^{c}_{n,A}$ are translation invariant in time if $0\notin A$ it suffices to consider two cases: $\min A=0$ and $\min A=1$.

\begin{proposition}[Induction hypothesis for $G^c_A$]\label{prop: overview b behavior} Let $\frt_0,\ell,\la$ be as in Proposition \ref{ass:  properties of T} and recall $\ep_n$ defined in  eq.\ \eqref{def: tilde alphas}. 
\ben
\item Let $1> \al > 1/2$. Then
\beq  \label{eq: induction assumption bulk}
   \sum_{A \subset \bbN: \str A \str=k , \min A% \mathrm{fixed} 
   =1}  \dist(A)^{\al}    \norm  G^{c}_{n,A} \norm_{\tengam}   \leq    \ep_n^{k},   \qquad k\geq 2
\eeq
\item If  the environment is in equilibrium, i.e., if $\be_1=\be_2$, then \eqref{eq: induction assumption bulk} holds  for $1/4< \al \leq 1/2$. In that case the correlation functions with $A \ni 0$  satisfy
\beq  \label{eq: induction assumption initial}
   \sum_{A \subset \bbN: \str A \str=k , \min A% \mathrm{fixed} 
   =0} 
     \dist(A)^{\al}    \norm  G^{c}_{n,A} \norm_{\tengam}   \leq    \ep^{}_{\ini,n} \ep_n^{k},      \qquad k\geq 2
\eeq
     \een
\end{proposition}

\subsubsection{Plan of proof of induction hypotheses} \label{sec: plan of proof}

Given the induction hypotheses Proposition \ref{ass:  properties of T} and Proposition \ref{prop: overview b behavior} up to scale $n$, we prove  Proposition \ref{ass:  properties of T} on scale $n+1$ in Section \ref{sec: flow of t}. This result is stated explicitly in Section \ref{sec: analysis of tprime}. 

The proof of Proposition \ref{prop: overview b behavior} on scale $n+1$ is spread over Sections \ref{sec: linear rg flow} and \ref{sec: nonlinear rg flow}. In the former section, we treat the linear part of the recursion relation and in the latter the nonlinear part (notions not yet defined).  The nonlinear part is straightforward and one does not need to distinguish between the  cases $\al >1/2$ and $\al\leq 1/2$. In particular, the equilibrium condition $\be_1=\be_2$ plays no role here. 
The linear part is slightly tricky; to treat the case  $\al\leq 1/2$, we have to exploit the equilibrium condition and consider correlation functions with $A \ni 0$. 
The linear part also dictates the choice of the exponents $\tilde \al, \tilde \al_{\ini}$ in eqs.\ (\ref{eq: values running epsilons},\ref{eq: values running epsilons initial}), respectively.

Then, we need to establish the induction hypotheses on scale $n=0$. In Section \ref{sec: estimates on the first scale excitations}, we establish Proposition \ref{prop: overview b behavior} and in Section \ref{sec: estimates on the first scale reduced evolution}, we establish Proposition \ref{ass:  properties of T}. In those sections, we also outline how to choose the constants $\gamma_0, \bana_0$ and $D_0$. 

\subsection{Proof of the main theorems} \label{sec: proof of main theorems}

We assume the induction hypotheses Proposition \ref{ass:  properties of T} and Proposition \ref{prop: overview b behavior} for all $n$. Then we have

\begin{proposition}\label{prop: overview t behavior}
 Let $\frt_0,\ell,\la$ be as in Proposition \ref{ass:  properties of T}.  There is a projector $P^{\star} \in \scrG$ of the form  $\str \mu^{\star}  \rangle  \langle 1_{\caS_{0}}  \str$, and a diffusion constant $\tilde D^{\star} >0$ such that for
$p$ in the strip $\str\Im p\str < \ga_0$
\beq\label{main bound for T}
 \left\norm \hat T_n(p)   -   \e^{-\tilde D^{\star}p^2 }   P^{\star}\right\normba   \leq   C  \ell^{-cn} 
  %(1+ \str p \str^4) \ell^{-(n/3)\min{(\tilde \al,1)}}, \qquad\textrm{for} \,\,    n > C \str p\str^2 (\log \ell)^{-1}  \label{eq: ultimate bound}
\eeq 
for some exponent $c>0$.
   In the case $\be_1 = \be_2$, we have   $\mu^{\star}= \mu^{\be}$, the projected Gibbs state   (\ref{newstates1}).
\end{proposition}

%
%\subsubsection{Proof of Proposition \ref{prop: overview t behavior}} \label{sec: proof of main prop}

\begin{proof}
 Let first
$ \str\Re p\str  \geq \bana_{n}$. Then the claim follows from the  definition \eqref{banandef} and the bounds (\ref{eq: largegap}) and \eqref{upper and lower bound for ef}.

\noindent For $ \str\Re p\str  \leq \bana_{n}$ %we  recall from  (\ref{eq: smallgap}  ) that $\norm \hat T_n(p)- R_n(p) \e^{f_n(p)} \normba \leq (\largegap_n$.
 note first that by  \eqref{eq: resulting steptostep} both $R_n(0)$ and $D_n$ are convergent sequences. Calling their limits $P^{\star}, \tilde D^{\star}$ respectively we get 
\beq\label{RminusP}
\norm R_n(0) - P^{\star} \normba, \str D_n -\tilde D^{\star} \str  \leq  \ep_n^{1/4}.
\eeq  %because of 
%\beq
%\str D_n -\tilde D^{\star} \str  \leq \sum_{j =n}^{\infty}   \str D_j -D_{j+1} \str  \leq  \sum_{j =n}^{\infty}   \ep_j =  \ep_n (1+ \ell^{-\tilde \al/2}+\ell^{-2\tilde \al/2}+ \ldots )
%\eeq
%and analogously for $P^{\star}$.
Let $m=cn$ where $0<c<1$ is such that $\sup_{\str\Re p\str  \leq \bana_{n}  }|\ell^{-m/2}p|\leq \gamma_0$. Such  an $n$-indepdent $c$ exists since  $ \bana_{n}$ grows not faster than $ (C n\log \ell)^{1/2}$ with $n$.
Write
\baq
\norm R_n(0) - R_n(p) \normba  &\leq &    \sum_{j =0}^{m-1}\norm R_{n-(j+1)}(p \ell^{-(j+1)}) - R_{n-j}(p \ell^{-j}) \normba\nonumber \\
& + &  \norm R_n(0) - R_{n-m}(0) \normba  +    \norm R_{n-m}(p \ell^{-m})- R_{n-m}(0) \normba
\eaq 
The first  and second term on the RHS is bounded by $ C\ep_{n-m}^{1/4}= C \ep_{(1-c)n}^{1/4}$ using  \eqref{eq: resulting steptostep}. The third term is bounded by $C(\ga_0)^{-1} \ell^{-m/2} $ by  analyticity of $p \to R_{n-m}(p)$ in the ball
of radius $\ga_0$ at origin, i.e.\ by \eqref{eq: hyp bound on r}. Hence 
\beq\label{RminusR}
\norm R_n(0) - R_n(p) \normba\leq C\ell^{-cn}.
\eeq
Next,  bound 
\baq
\str e^{f_n(p)}- \e^{-\tilde D^{\star} p^2} \str & \leq &  \str e^{f_n(p)}- \e^{- D_n p^2} \str  +  \str \e^{-D_n p^2}- \e^{-\tilde D^{\star} p^2} \str   \nonumber \\[2mm]
 & \leq &            C( \str f_n(p)+ D_n p^2 \str + \str D_n - \tilde D^{\star} \str) \leq C\ell^{-cn}
\label{fminusD}
\eaq
by  \eqref{eq: fclosetoparabola} and \eqref{RminusP}, and the slow growth of $\bana_n$.

Combining (\ref{RminusP}), (\ref{RminusR}) and (\ref{fminusD}) with  (\ref{eq: smallgap}) yields the bound (\ref{main bound for T}).
It remains to argue that $P^{\star}=P^{\be}$ in the case $\be_1=\be_2=\be$.  Take $\La$ finite and recall that   $ \initialgibbsfinite=  \caU_{n=0}(\tau=0)    \initialfinite   $ and, by the invariance of the Gibbs state
\beq  \label{eq: invariance of gibbs under breve}
    \rho_{\sys}^\be=  \Tr_\res [\e^{-\i t L}\initialgibbsfinite ] = \breve Z_t    \initialsysfinite, \qquad t  \geq 0   
\eeq
with  $\breve Z_t$ defined in Section  \ref{sec: correlations involving zero}.
On the other hand, 
\beq  \label{eq: asymptotic splitting breve z}
 \bsS_{\ell^n}[\breve Z_{ \ell^{2n} t_0}]   = \bbE (\caU_{n}  \caU_n(0)) =   T_n T_n(0) +  \caT[G^c_{\{0,1 \},n}  ]
\eeq
where $(0)$ refers to $\tau=0$. On both sides, we can 
take the thermodynamic limit. Moreover, we know that  $\norm G^c_{\{0,1 \},n} \norm_{\tengam}  \leq \e^{-cn} $ by Proposition \ref{prop: overview b behavior}.    Let us then Fourier transform \eqref{eq: asymptotic splitting breve z} and apply the $p=0$-component to $\mu^{\reff}$, this yields
\beq
\mu^{\be} =  \hat T_n(0) \hat T_n(0,0) \mu^{\reff}  +\caO(\e^{-cn})
\eeq
where the LHS follows from \eqref{eq: invariance of gibbs under breve} and we wrote $\hat T_n(0,0)=\hat T_n(p=0,\tau=0), \hat T_n(0)=\hat  T_n(p=0)$. 
In fact, by \eqref{eq: invariance of gibbs under breve} for $t=0$, we have also $\hat T_n(0,0) \mu^{\reff}=\mu^{\be}$ and by Proposition \ref{prop: overview t behavior}, we have $\hat T_n(0)=P^{\star}+\caO(\e^{-cn})$.  Hence
\beq
\mu^{\be} = P^{\star} \mu^{\be}
\eeq
and this of course implies $P^{\star}=P^{\be}$. 

\end{proof}

\subsubsection{Proof of results in Section \ref{sec: results} along a subsequence of times}\label{sec: proof of diffusion}

We argue that 
Proposition  \ref{prop: overview t behavior} implies the results of Section \ref{sec: results} along
the sequence of times  $t_n :=  \ell^{2n}t_0$.  
%  Theorems \ref{thm: nonequilibrium}, \ref{thm: equilibrium}, \ref{thm: convergence to gibbs} and \ref{thm: convergence to ness} 
The resulting diffusion
constant is given by
\beq
D^{\star} =   t_0^{-1} \tilde D^{\star} =  \frt_0^{-1} \la^2 \tilde D^{\star}
\eeq
Let us first write the claims of 
Theorems \ref{thm: equilibrium} and \ref{thm: nonequilibrium}  in terms of the RG. Express the time $t$ in units of the kinetic time scale,  $t= s t_0$,
$t_0= \la^{-2}\frt_0$.  Then we have
\beq\label{diffalarg0}
  \Tr [\e^{\i p\frac{X}{\sqrt{t}}} \rho_{\sys, t} ]  =    \Tr[ \e^{\i p\frac{X}{\sqrt{t_0}}}{\bsS }_{\sqrt{s}}[Z_{s t_0}] S_{\sqrt{s}}\rho_{\sys,0} ]\eeq
For $t=t_n =  \ell^{2n}t_0$,   ${\bsS }_{\sqrt{s}}[Z_{s t_0}]= T_n$ and by the Fourier transform, we get
\beq\label{diffalarg}
  \Tr [ \e^{\i p\frac{X}{\sqrt{t_{n}}} } \rho_{\sys, t_n}  ] =      
 \tr [ \hat T_n(\frac{p}{\sqrt{t_0}})\hat\rho_{\sys,0} (\frac{p}{ \ell^{n}\sqrt{t_0}}) ]
\eeq
where  the 'trace' $\tr$ is defined by  $\tr[ \psi] = \sum_{s \in \caS_0} \psi(s)$.  Both $ \hat T_n$ (by Proposition \ref{ass:  properties of T})
and $\hat\rho_{\sys,0}$ (by the finite-range condition \eqref{initialS})  are analytic in $p$
and uniformly bounded in the strip $|\Im p|<\ga_0$. Hence, by Proposition \ref{prop: overview t behavior}, 
\eqref{diffalarg} converges 
 as $n\to\infty$ to 
$$ e^{-D^\star p^2}\tr[ P^\star\hat\rho_{\sys,0} (0)] = e^{-D^\star p^2}.
$$
By the Vitali theorem, 
the derivatives converge as well and  Theorem \ref{thm: equilibrium} and \ref{thm: nonequilibrium} follow (along a subsequence).

For Theorems \ref{thm: convergence to gibbs} and \ref{thm: convergence to ness} 
we need to consider a translation invariant observable $A$ instead of $\e^{\i p\frac{X}{\sqrt{t}} }$ in 
\eqref{diffalarg}. Its kernel $A(x,s)=A(s)$ is constant in $x$, hence we have $ \Tr [A \rho_{\sys, \ell^{2n}}] =  \tr [A \hat\rho_{\sys, \ell^{2n}}(0)]$, and, by
Proposition \ref{prop: overview t behavior}, 
\beq \label{eq: conv ness}  
\lim_{n\to\infty}  \Tr [A \rho_{\sys, \ell^{2n}} ]  =  \tr [ A P^\star \hat\rho_{\sys,0} (0)]= \tr[A \mu^\star] 
\eeq
In the case where $\mu^{\star}=\mu^{\be}$, the last line of course  equals $ \langle A \rangle_{\be}$.
The decoherence result \eqref{result: decoherence} is a simple consequence of the fact that $\norm P^{\star}\normba <C$, hence $\sum_{s}\mu^{\star}(s) \e^{\ga_0 \str v \str} <C $ with $v$ as in Section \ref{sec: kernels and rescaling of space}.

\subsubsection{Proof of results in Section \ref{sec: results} for general times}\label{sec: proof of diffusion general times}

Let us now consider general times in \eqref{diffalarg0}. %Recall we proved Propositions \ref{prop: overview b behavior} and \ref{prop: overview t behavior} for all $\frt_0\in[\tau,\ell^2\tau]$.
We will generalize a bit our RG iteration. We can run the RG iteration
i.e.\ Propositions \ref{prop: overview b behavior} and \ref{prop: overview t behavior} as well with scaling factor
$2\ell$ (by possibly reducing $\la$) and we can also at each iteration step choose arbitrarily
between  the two factors. 
Let $\sigma\in \{0,1\}^\bbN$ label the possible choices: if $\sigma_n=0$
at the $n$:th step apply scaling $\ell$,  if $\sigma_n=1$
 apply scaling $2\ell$. Denote the resulting sequence of operators $T_{n,\sigma}$. They satisfy
 Proposition  \ref{prop: overview t behavior} uniformly in $\sigma$ and hence we obtain projectors $P^\star_\si$ and diffusion constants $D^\star_\si$.  
\begin{lemma}
$P^\star_\sigma, D^\star_\si$ are independent of $\sigma$.
\end{lemma}
\begin{proof} 
 Denote the
 corresponding times by $t_{n,\sigma}$:
 \beq\label{tnsigma}
 t_{n,\sigma}= \la^{-2}\ell^{2n}4^{m}\frt_0
 \eeq
where $m=m(n,\sigma):=\#\{i  \leq   n :  \sigma_i= 1 \}$. Since $T_{n,\sigma}={\bsS }_{ (t_{n,\sigma}/t_0)^{\hf}} [Z_{ t_{n,\sigma}}]$ we have $ T_{N,\sigma}= T_{N,\sigma'}$ whenever $ t_{N,\sigma}= t_{N,\sigma'}$. In such case,
since both sequences satisfy the bounds of Proposition \ref{prop: overview t behavior} we conclude
\beq\label{diffof ts}
\norm P^\star_\si- P^\star_{\sigma'} \normba, \str D^\star_\si- D^\star_{\sigma'} \str \leq C\ell^{-c  N}.
\eeq
Now, given $\sigma$, $\sigma'$ and $n$ there exists $N\geq n$ and $\tilde\sigma$, $\tilde\sigma'$ s.t.
$\tilde\sigma_i=\sigma_i$ for all $i\leq n$, similarly for the primes, and $t_{N,\tilde\sigma}=t_{N,\tilde\sigma'}$. Indeed,
just take $N=n+|m(n,\sigma')-m(n,\sigma)|$ and  $\tilde\sigma_i'=0$, $\tilde\sigma_i=1$ for $i>n$ in case 
$m(n,\sigma')-m(n,\sigma)>0$ and vice versa in the opposite case. 

Hence, we 
conclude  that \eqref{diffof ts} holds for all  $\sigma$, $\sigma'$ and $N$ and the claim follows.
\end{proof}
 Since the claims of the  Propositions \ref{prop: overview b behavior} and \ref{prop: overview t behavior}
 are uniform in the initial time $\frt_0\in[\tau_0,2\tau_0]$ as well we have $\sigma$-independent limits   $P^\star_{\sigma,\frt_0}\equiv P^\star_{\frt_0}, D^\star_{\si,\frt_0} \equiv D^\star_{\frt_0} $. Let us also denote the  $\frt_0$-dependence of the times \eqref{tnsigma}
 explicitly as $ t_{n,\sigma,\frt_0}$. Let $U\subset[\tau_0,2\tau_0]$ be the set of $\frt_0$ s.t. there 
 exist $\sigma$, $\sigma'$ and $n,n'$ such that  $ t_{n,\sigma,\frt_0}= t_{n',\sigma',\tau_0}$. 
By similar reasoning as in the proof above, one sees that for $\frt_0 \in U$, one can in fact  find arbitrarily large $n,n'$ such that this equality holds. It follows that
 on the set $U$, the limits $P^\star_{\frt_0}, D^\star_{\frt_0}$ are constant. 
 
\begin{lemma} If $\log 2/\log\ell$ is irrational, then
 $U$ is dense in $[\tau_0,2\tau_0]$.
\end{lemma}
\begin{proof}
Consider $\log t_{n,\sigma,\frt_0}=-2\log \la+\log\frt_0+2n\log\ell+2m\log 2$. If
$\log 2/\log\ell$ is irrational, then the map $x\to (x+\log 2/\log\ell) \mod 1$ has  dense  orbits
and the claim follows easily from this.
\end{proof}
 By a similar  density of
orbits argument, it is easy to see that 
there exists $T<\infty$ such that
the set of times $ t_{n,\sigma,\frt_0}$ with $\frt_0\in U$  is dense
in $[T,\infty)$. Along any sequence of times from this dense set, the limit  \eqref{eq: convergence to gaussian}  exists and is independent of the chosen sequence. 
  By the strong continuity  of the $Z_t$ (which will be easily derived in Lemma \ref{lem: bound on h}),
 the function  \eqref{eq: convergence to gaussian} 
is continuous in $t$ and hence the limit is independent of $\ell, \frt_0$. This completes the proof of Theorems \ref{thm: equilibrium} and \ref{thm: nonequilibrium}.  Theorem 
\ref{thm: convergence to gibbs} follows analogously by considering sequences of times in \eqref{eq: convergence to stat state}.

\section{Flow of $T$}\label{sec: flow of t}

As announced in Section \ref{sec: plan of proof}, we prove the  induction step for $T$. 
The following convention applies to  Sections \ref{sec: flow of t}, \ref{sec: linear rg flow} and \ref{sec: nonlinear rg flow}:  We always assume that the induction hypotheses Propositions \ref{ass:  properties of T} and \ref{prop: overview b behavior} are satisfied on scale $n$ (we then prove them on scale $n+1$), and we do not repeat this assumption in the statements of all lemmata and propositions.  Throughout this Section, as well as in Sections \ref{sec: linear rg flow} and \ref{sec: nonlinear rg flow}, we use the conventions on $\ell, \ep_0, \ep_{\ini,0}$ that were explained at the beginning of Section \ref{sec: induction hypotheses}.

\subsection{Powers of $T_n$}

We show how to bound powers of $T_n$ and we abbreviate $T=T_n, R=R_n, \ldots$ whenever no confusion is possible.  This bound will be an important ingredient of both  induction steps for $T_n$ and   $G^c_{n,A}$.  
\begin{lemma}\label{lem: powers of t}
$T=T_n$ satisfies  
\baq
&  &  \int \d x  \,  \e^{10 \gamzero \str x \str }   \norm ({\bsS_{\ell}}T^m)(x)    \norm^{}_{\banone}      \leq C, \qquad    1 \leq m \leq \ell^2  \label{eq: lone bound for repeated t}  \\[1mm] 
&&   \sup_{x }   \e^{20 \gamzero \str x \str }    \norm( {\bsS_{\ell}}T^{\ell^2})(x)    \norm^{}_{\banone}   \leq C
 \label{eq: linfinity bound for repeated t} \\[1mm] 
&&\sup_{x}  \e^{\frac{20\ga_0}{\sqrt{m}}  \str x\str}    \norm T^m(x)\norm_\banone  \leq C m^{-d/2} \label{eq: linfinity bound for m repeated t}, \qquad    20 \leq  \sqrt{m} \leq \ell 
\eaq
\end{lemma}
\begin{proof}
 Let us first suppose $  \max(20, \frac{16d}{\tilde \al}) \leq  \sqrt{m} \leq \ell $. By analyticity
\beq  \label{eq: power of t to calculate}
 \e^{\frac{\twentygam}{\sqrt{m}}  |x|} T^{m}(x) =  (2 \pi)^{-d} \mathop{ \int}\limits_{\bbT_n} \d p   \, \hat T (p
 + \i \frac{ \twentygam }{\sqrt{m}}e_x ) ^{m}  \e^{\i p x} 
\eeq
where $e_x=x/|x|$.
To evaluate this integral let $\chi_n(p)$ be the indicator function for  $\str \Re p \str  \leq \bana_n $. On its support
(and with $\str\Im p\str \leq\ga_0$) we have  
$\hat T(p)^m= e^{mf(p)}R(p)+ ((1-R(p))\hat T (p))^m$ and using  \eqref{eq: smallgap} and \eqref{upper and lower bound for ef},
\beq\label{smallpbound}
 \mathop{\int}\limits_{ \bbT_n}\d p \chi_n(p)
  \norm  \hat T (p + \i \frac{ \twentygam }{\sqrt{m}}e_x) ^{m} \norm_{\banone}     \leq C \e^{ C\ga^2_0  }  \mathop{\int}\limits_{ \bbR^d} \d p \,  \e^{- \frac{1}{2} D_n m p^2}  +     \largegapn_n^m \bana_n^d
  \leq  C (m^{-d/2}  + \ell^{-c mn} ) 
\eeq
with $c>0$ since $\bana_n$ grows no faster than $ C(n \log \ell)^{\hf}$.

To perform the integral with $1-\chi_n$, we use (\ref{eq: largegap}) and $m >20$ (first inequality), and  then $m >  \frac{16d}{\tilde \al}$ (last inequality), to obtain
\beq \label{largepbound}
\mathop{ \int}\limits_{\bbT_n} \d p  \, (1-\chi_n(p))\norm \hat T^m (p + \i \frac{ \twentygam }{\sqrt{m}}e_x)  \normba   \leq   \largegapn_n^{m}  \mathrm{Vol}(\bbT_{n})  =    (2\pi)^d (\hf)^m \ell^{-\frac{\tilde \al}{8} m n} \ell^{d n}  \leq  \largegapn_n ^{m/2}
 \eeq
 where $\mathrm{Vol}(\bbT_{n})= (2\pi)^d \ell^{dn}$ is the volume of the $d$-dimensional torus $\bbT_n= \ell^n\bbT^d$.  
%  \baq
% &&  \int \d p  \norm  \hat T (p)  \norm_\banone^{2}   \\
%&= &  % (\hf   \ell^{-\frac{\tilde\al}{4}n})^{m-2} 
% \largegapn_n ^{m-2}     \int \d x  \norm  T (x)  \norm_\banone^{2}  \e^{ \frac{ 40 \ga_0 }{\sqrt{m}} \str x \str } \\
%& \leq  &      C% (\hf   \ell^{-\frac{\tilde\al}{4}n})^{m-2}   
%   \largegapn_n ^{m-2}   \sup_{x } \norm  T (x)  \normba  \e^{ \ga_0 \str x\str}  \leq  C  %(\hf   \ell^{-\frac{\tilde\al}{4}n})^{m-2}     
%   \largegapn_n ^{m-2}    , \qquad  
%\eaq
%\baq \label{largepbound}
% \mathop{ \int} \d p  \, (1-\chi_n(p))\norm \hat T^m (p)  \normba & \leq &      \largegapn_n ^{m-2}  
% %(\hf   \ell^{-\frac{\tilde\al}{4}n})^{m-2} 
% \int \d p  \norm  \hat T (p)  \norm_\banone^{2}   \\
%&= &  % (\hf   \ell^{-\frac{\tilde\al}{4}n})^{m-2} 
% \largegapn_n ^{m-2}     \int \d x  \norm  T (x)  \norm_\banone^{2}  \e^{ \frac{ 40 \ga_0 }{\sqrt{m}} \str x \str } \\
%& \leq  &      C% (\hf   \ell^{-\frac{\tilde\al}{4}n})^{m-2}   
%   \largegapn_n ^{m-2}   \sup_{x } \norm  T (x)  \normba  \e^{ \ga_0 \str x\str}  \leq  C  %(\hf   \ell^{-\frac{\tilde\al}{4}n})^{m-2}     
%   \largegapn_n ^{m-2}    , \qquad  
%\eaq
Combining the bounds \eqref{smallpbound} and \eqref{largepbound},  
we get a bound  $O(m^{-d/2})$ for  \eqref{eq: power of t to calculate} which implies (\ref{eq: linfinity bound for m repeated t}) for  $\max(20, \frac{16d}{\tilde \al}) \leq  \sqrt{m}$.  Taking $m=\ell^2$ and using that 
$
(\bsS_{\ell}T^{m})(x) = \ell^d T^{m}(\ell x)
$, we get  \eqref{eq: linfinity bound for repeated t}.

To get \eqref{eq: lone bound for repeated t}, we first note that, from (\ref{eq: linfinity bound for m repeated t})
 we get immediately
\beq
  \int\d x \,  \e^{\frac{10 \ga_0}{\sqrt{m}}  | x |} \norm T^{m}(x)\norm_\banone  \leq C
\eeq
By scaling and using   $\frac{10  \ell \ga_0}{\sqrt{m}}  \geq {10   \ga_0} $, 
 \eqref{eq: lone bound for repeated t} follows
 for  $\sqrt{m} \geq  \max(20, \frac{16d}{\tilde \al})$. For $20 \leq \sqrt{m} \leq  \frac{16d}{\tilde \al} $  we get  $\norm T^m (x)  \norm_{\banone}   \e^{ \gamzero \str x \str }  \leq C$ from eq.\ \eqref{tpos}. This settles \eqref{eq: linfinity bound for m repeated t} and, upon scaling, also
\eqref{eq: lone bound for repeated t}.
\end{proof}
For later purposes we register the following consequence of (\ref{smallpbound}) and (\ref{largepbound}):
\beq\label{precise bound for Tm}
 \norm( {\bsS_{\ell}}T_n^{\ell^2})(x) -\tilde T_{n+1}(x) \norm^{}_{\banone}  
 \leq C\ell^{-cn}e^{-{20\ga_0}|x|}
  \eeq
where  ($\tilde p:=p+20\i\ga_0e_x$)
\beq\label{precise bound for Tm1}
\tilde T_{n+1}(x)=   (1/2 \pi)^{d}  \e^{-{20\ga_0}|x|} \mathop{ \int}\limits_{\bbT_{n+1}} \d p   \, R_n (\tilde p/\ell) e^{\ell^2f_n(\frac{\tilde p}{\ell})}  \e^{\i p x} \chi_n(p/\ell) .
  \eeq

\subsection{Contribution to $T_{n+1}$ from $G^c_{A}$}  \label{sec: contribution to tnplusone}

We recall  the expression \eqref{eq: T from connected correlation functions}  for $T_{n+1}$ in terms of quantities at scale $n$;
\beq
T_{n+1}=  \bsS_{\ell}  T_n^{\ell^2}  + \sum_{ \caA \in \poly(I_{\tau'})}  \, \, \,    \caT_{\tau'}  \bsS_{\ell} \left[  \mathop{\otimes}\limits_{A \in \caA}  G^c_{n,A}  \mathop{\otimes}\limits_{\tau \notin \supp \caA}  T_n ({\tau})  \right]    \label{eq: t recursion repetition}
\eeq
and below we define $E_n$ as the second term  on the RHS.

We now derive the necessary bounds on $E$. This bound is the only place in the induction step $T \to T'$ in which we need information on the cumulants $G^c_A$. 
\begin{lemma}\label{lem: bound on e}
Let 
\beq
E_n:=  T_{n+1}- \bsS_{\ell}  [T_n^{\ell^2}].  
\eeq
Then
\beq\label{xbound for E}
\sup_x  \e^{10 \ga_0 \str x\str}\norm E_n(x) \norm_{\scrG}       
\leq  C \ep_n
\eeq
and hence in particular $\norm E_n \norm_{\ga_0}  \leq C \ep_n$.
\end{lemma}

\begin{proof}
We (again) abbreviate $T=T_n$, $E=E_n$ and $G_A^c= G^c_{n+1,A}$. 
We take $\ell > 10$ and  apply Lemma \ref{lem: main property of norms} to eq.\ \eqref{eq: t recursion repetition}:
\beq
\norm E \norm_{\tengam}= \left\norm \caT_{\tau'}\bsS_{\ell} \left[  \mathop{\otimes}\limits_{A \in \caA} G^{c}_{A}  \mathop{\otimes}\limits_{\tau \notin \supp \caA}  T({\tau})  \right] \right\norm_{\tengam}  \leq  \prod_{A \in \caA}  \norm G^{c}_{A} \norm_{\ga_0}    \prod_{ J  }    \norm T^{\str J \str} \norm_{\tengam/\ell}
\eeq
where the product  $\prod_{ J  }$ runs over all discrete intervals $J$ in the sets $I_{\tau'} \setminus \supp\caA$. That is; we decompose the set $I_{\tau'} \setminus \supp\caA$ into a union of discrete intervals  $J$ (i.e.\ sets of consecutive numbers) such that no two of those intervals are consecutive.
  By invoking Lemma \ref{lem: powers of t}, we bound 
\beq
 \norm T^{\str J \str} \norm_{\tengam/\ell}  \leq  C,     \qquad  \text{since} \,\,  \str J \str \leq \ell^2
\eeq
The number of discrete intervals is at most $1+ \str\supp \caA \str < 2\str\supp \caA \str  $ and hence we can bound  
\beq
\norm E \norm_{\tengam} \leq  \sum_{ \caA \in \poly(I_{\tau'})} C^{\str \supp\caA \str}  \prod_{A \in \caA}   \norm G^c_{A}  \norm_{\gamzero}  \leq  C  \frac{\ell^2\ep_n^2 }{1- C \ell^2\ep_n^2 }   \label{eq: simple nonlinear bound}
\eeq
where the second inequality follows from Proposition \ref{prop: overview b behavior} and the fact that $\sum_{A \in \caA}\str A \str$  is at least $2$.

To get a bound on $\sup_x  \norm E(x) \normba  \e^{\ga \str x \str}$, we must proceed differently; we split the contributions to $E=E_{I}+ E_{II}$ in \eqref{eq: t recursion repetition} into those where there is at least one $T$, namely when $\supp \caA \neq I_{\tau'}$, and those for which there are no $T$'s, i.e.,  $\supp \caA = I_{\tau'}$.  For the second term, $E_{II}$, we use that 
\baq
\norm E_{ II} \norm_{\tengam}   &\leq &  \sum_{ \caA \in \poly(I_{\tau'}),\supp \caA =I_{\tau'}}   C^{\str \supp\caA \str}   \prod_{A \in \caA}  \norm G^c_{A} \norm_{\gamzero}  \leq     C^{\ell^2}   \ell^{2\ell^2} \ep_n^{\ell^2}
\eaq
where  $\ell^{2\ell^2}$ originates as a crude bound for the number of partitions of $\ell^2$ elements (bounding the sum over $\caA \in \poly(I_{\tau'})$). 
Further, remark that  $E_{II}, E_{I}$ are translation invariant and use  the a priori bound on the supremum norm given by 
\beq
\sup_{x } \norm E_{II}(x) \normba  \e^{\ga \str x \str} \leq   \ell^{dn}  \int_{\bbX_n} \d x \norm E_{II}(x) \normba  \e^{ \ga \str x \str} =  \ell^{dn} \norm E_{II}\norm_{\ga}
\eeq
This bound originates from the definition of $\int_{\bbX_n} \d x $  in \eqref{def: sum over lattice}.    Using now the bound on   $\norm E_{II} \norm_{\tengam}$ derived above, we get 
\beq
\sup_{x}  \norm E_{II}(x)  \normba  \e^{\tengam \str x \str} \leq     \ell^{ dn} (C \ell^2\ep_n)^{\ell^2}   \leq  C \ep_0 \ell^{ -cn \ell^2\tilde \al }  \leq \ep_n
\eeq
where the second and third inequality follows from $\ep_n= \ep_0 \ell^{-\tilde \al n}$ (by Proposition \ref{prop: overview b behavior}). 

We now deal with $E_I$.  For any $\caA$, choose a $\tau$ such that $\tau \notin \supp \caA$ (by the definition of $E_I$, there is at least one) and let the product $\prod_{J}$ run over the discrete intervals $J$ in $I_{\tau'} \setminus \supp \caA \setminus \tau$. Then
\baq
\sup_x \norm E_I(x) \normba  \e^{\tengam \str x \str }& \leq &
  \sum_{ \caA \in \poly(I_{\tau'}),\supp \caA  \neq I_{\tau'}}  \prod_{J} \norm T^{\str J \str} \norm_{\tengam/\ell}    \prod_{A \in \caA} \norm G^{c}_{A} \norm_{\ga_0}    \left(  \sup_{x} \norm  T(x) \normba  \e^{\ga_0 \str x \str}  \right)        \\
  &  \leq & C \ep_n    \sup_{x}  \norm  T(x) \normba  \e^{\ga_0 \str x\str}   \leq    C \ep_n 
\eaq
The second inequality is obtained by $\prod_{J} \norm T^{\str J \str} \norm_{\tengam/\ell}  \leq C^{\str\supp \caA\str}$ and summing $\caA$ as in \eqref{eq: simple nonlinear bound}. The third inequality is 4) of Proposition \ref{ass:  properties of T}.
To get the first inequality, first apply point 2) of Lemma \ref{lem: property of norm} in all $s$-coordinates and then apply  point 3) of the same Lemma in the $x$-coordinates, choosing $\tau$ for the index `$i$'.
\end{proof}

\subsection{Analysis of $ T_{n+1}$} \label{sec: analysis of tprime}

We complete the  induction step for $T$ by establishing

\begin{lemma}\label{lem: flow of parameters for t}
Proposition \ref{ass:  properties of T} holds on scale $n+1$. 

\end{lemma}

\begin{proof}   Whenever confusion is excluded, we  abbreviate $T=T_n, R=R_n$.  From
Lemma  \ref{lem: bound on e} we get
\beq
 \hat T_{n+1}(p)= \hat T_n(p/{\ell})^{\ell^2}+\hat E_n(p), \qquad \text{with}\,\,   \sup_{\Im p \leq 2 \ga_0} \norm \hat E_n(p) \normba   \leq  C {\ep}_n.
    \label{eq: Tprimeand E}
\eeq
Let first $ |\Re p| \geq   \bana_{n+1}$ (recall the definition of $\bana_{n}$ in \eqref{banandef}).  To study the first term in (\ref{eq: Tprimeand E}) with  Proposition  \ref{ass:  properties of T}
we need to consider two cases:
\vskip 1mm
 \noindent 1. $\str\Re p\str \geq \ell \ \bana_n$. By (\ref{eq: largegap})
\baq
 \norm    \hat T(p/{\ell})^{\ell^2}   \normba  \leq   \largegapn_n^{\ell^2} < \frac{1}{2}\largegap_{n+1}.
\eaq
Now use  $C\epsilon_n=C\epsilon_0\ell^{-\tilde\al n}\leq \frac{1}{2}\largegap_{n+1}$ by taking $\epsilon_0\leq c(\ell)$ to get
(\ref{eq: largegap}) for $n+1$.
\vskip 1mm
\noindent 2. $  \bana_{n+1}   \leq | \Re p| \leq \ell \ \bana_n$.
By Proposition \ref{ass:  properties of T} 2) we have in this region
\beq\label{Tpell}
  \hat T(p/{\ell})^{\ell^2} =
R(p/{\ell})\e^{\ell^2 f_n( \frac{p}{\ell})}   + (1- R(p/{\ell}))  \hat T(p/{\ell})^{\ell^2}
\eeq
%\beq\label{Tl2split}
%\norm \hat T'(p)\normba & \leq&   \norm   W_1(p)\normba  +  \norm   W_2(p)\normba  +  \norm \hat E(p)\normba \\
 %\leq \norm R(p/{\ell}) \normba  \left\str\e^{\ell^2 f( \frac{p}{\ell})} \right\str  +   (\smallgapn )^{\ell^2} + C {\ep}
%\eeq
To control the first term on the RHS, we use  \eqref{eq: hyp bound on r}, \eqref{banandef} and \eqref{upper and lower bound for ef} to get
\beq
\norm R(p/{\ell}) \e^{\ell^2 f_n(p/\ell)  }  \normba  \leq    \e^{- \hf  D_n \bana_{n+1}^2} \e^{  C \ga_0^2 }\leq C\ell^{-
 \frac{1}{4}\tilde\al(n+1)}
\leq \frac{1}{3}\largegap_{n+1} \label{eq: deriving upper bound on f2}.
\eeq
The second term on the RHS of (\ref{Tpell}) is bounded by \eqref{eq: smallgap},
\beq
\norm  (1- R(p/{\ell}))  \hat T(p/{\ell})^{\ell^2} \normba  \leq
\largegapn_n^{\ell^2}\leq
 \frac{1}{3}\largegap_{n+1}
 \label{eq: deriving upper bound on 1-RT}
\eeq
Again, we may assume 
$C\epsilon_n\leq \frac{1}{3}\largegap_{n+1}$  and so (\ref{eq: largegap}) holds for $n+1$. 

\vskip 2mm

\noindent Let then  $ |\Re p| \leq   \bana_{n+1}$.
  We apply spectral perturbation theory, see Lemma \ref{lem: perturbation banach}.
Referring to the notation in that lemma, we  have
$A= \hat T_{n+1}(p)$ and $A_0 = \hat T_{n}(p/{\ell})^{\ell^2} $, $a_0=  \e^{\ell^2 f_n(p/\ell)  } $, $P_0=R_n(p/\ell)$
and $A_1= E_n(p)$. 

Let us first consider $p$ in a wider strip $|\Re p|\leq\bana_{n+1}$ and $\str\Im p\str \leq 2\gamma_0$. In this region
by    \eqref{upper and lower bound for ef},
\beq\label{a0bound}
|a_0| =
 |\e^{\ell^2 f_n(p/\ell)  }|\geq \e^{- \frac{3}{2}D_n\bana_{n+1}^2-C\ga^2_0} >c\ell^{-\frac{3\tilde\al}{4} (n+1)}
 \eeq
whereas from (\ref{eq: smallgap})
\beq
\norm A_0 - a_0 P_0\normba   \leq  \largegapn_n^{\ell^2}=(\hf)^{\ell^2} \ell^{-\ell^2\frac{\tilde\al}{8} n}.
%\leq  \smallgapn_0^{\ell^2} \ell^{-\ell^2\frac{\tilde\al}{4} n}.
 \eeq
This is no bigger than $|a_0|/2$ say and so \eqref{eq: isolated and bound}, the first condition of the  Lemma \ref{lem: perturbation banach}, holds. 
Because of $\|A_1\|\leq C\ep_n$  and  \eqref{eq: hyp bound on r},   the condition (\ref{eq: arbirary bound spectral})
holds if
$$
C\ep_n\leq |a_0| 
$$
which is true for small $\ep_0$. Lemma \ref{lem: perturbation banach} then implies that the isolated eigenvalue persists and 
\baq
\left\str \e^{f_{n+1}(p)} - \e^{\ell^2 f_n(\frac{p}{\ell})}  \right\str  &\leq&  C  {\ep}_n, \label{eq: error f next scale} \\[2mm]
\norm  R_{n+1}(p) -  R_n(p/\ell)  \normba &\leq &   C\ep_n(|a_0|-\norm A_0 - a_0 P_0\normba-C\epsilon_n)^{-1}
  \leq C \ep_0  \ell^{-n \tilde \al}\ell^{\frac{3}{4}\tilde\al (n+1)} \label{eq: error r next scale}
\eaq
so the first claim in (\ref{eq: resulting steptostep}) follows. Furthermore, 
\beq
\norm  (1- R_{n+1}(p))  \hat T_{n+1}(p)\normba  \leq 
C\norm  R_{n+1}(p) -  R_n(p/\ell)  \normba +%C\smallgapn_n^{\ell^2} +C\ep_n\leq \smallgapn_{n+1}.
C\largegapn_n^{\ell^2} +C\ep_n\leq \largegapn_{n+1}.
% \label{eq: deriving upper bound on 1-RT}
\eeq
It remains to iterate  \eqref{eq: fclosetoparabola}.  
The analyticity of $f_{n+1}$ follows immediately.  To derive the symmetry properties of $f_{n+1}$, consider a lattice symmetry $O$. We start from the symmetry property (cfr.\ Section \ref{sec: symmetries})
\beq T_{n+1} =     \Adjoint(V^{-1}_O))  T_{n+1} \Adjoint(V_O)
\eeq
 By the result in Section \ref{sec: symmetries of the lattice}, if follows then that
 \beq
  \hat T_{n+1}(O^{-1} p)  =        I^{-1}_{p,O}       \hat T_{n+1}( p) I^{}_{p,O}
 \eeq
 which implies that $f_{n+1}(p)=f_{n+1}(O^{-1} p) $, and hence the Taylor expansion of $f_{n+1}(p)$ around $p=0$ contains no odd powers.

  Furthermore, by  \eqref{a0bound} and \eqref{eq: error f next scale}, we have $f_{n+1}(p) =   \ell^2 f_n({p}/{\ell}) +   g(p)$
  with
  \beq
  \sup_{|\Im p|\leq  2 \ga_0} \str g(p) \str   \leq     C  \ep_n  \ell^{ \frac{3}{4} \tilde \al (n+1) } \leq \sqrt{\ep_0}   \ell^{- \frac{1}{4} \tilde \al n  }    \label{eq: bound on jp}
\eeq
  The diffusion constant $ D_{n+1}$ is defined as the quadratic term in the Taylor expansion: $ D_{n+1} \delta_{i,j}= \hf \nabla_{i}\nabla_{j} f_{n+1}(p=0) $.  By Cauchy,
\beq
\str D_{n+1} -D_n  \str  \leq   2 \ga_0^{-2} \sup_{\str p \str \leq 2\ga_0} \str g(p) \str
\eeq
and the second claim of  (\ref{eq: resulting steptostep}) follows by the bound \eqref{eq: bound on jp}.
Furthermore, for $|\Im p|\leq \ga_0$ we have
\baq
\left\str    f_{n+1}(p) + D_{n+1}p^2  \right\str   &\leq&      |    \ell^2 f_n({p}/{\ell}) + D_n p^2  | + | g(p)-\hf \sum_{i,j} p_ip_j ({\nabla_i\nabla_j} g)(0) | \\[2mm]
 &\leq&    (\ell^{-1} \ell^{-\frac{\tilde\al}{4}n}       + 6 \ga_0^{-3}\sup_{|\Im p' |\leq  2 \ga_0}  \str g(p') \str ) \left\str  p  \right\str^3 
\eaq
which proves \eqref{eq: fclosetoparabola} at $n+1$ because of \eqref{eq: bound on jp}.

Eq. \eqref{eq: projector r explicit} follows since 
 $T_n$ necessarily preserves the trace of density matrices and 
$\Tr \rho = \sum_{s \in \caS_{0}} \hat\rho(p=0,s)$. The bound
\eqref{eq: projector r explicit bound} iterates due to \eqref{eq: error r next scale}
and will be established for $n=0$ in Section \ref{sec: t at n equal 0}.

\vskip 2mm

\noindent We complete the proof of items 2) and 3) of Proposition \ref{ass:  properties of T} by noting that the constants $C$ in the bounds (\ref{eq: hyp bound on r}, \ref{eq: projector r explicit bound}) may be chosen uniform in $n$ by invoking \eqref{eq: resulting steptostep} for all $n' \leq n$.
 To get item 1) of Proposition \ref{ass:  properties of T}, note similarly that the proofs of 2) and 3) give for
all $p$ with $|\Im p|\leq\ga_0$
$$
\norm \hat T_{n+1}(p)-\hat T_n(p) \normba \leq C\ell^{-cn}
$$
with $c>0$. Since this holds with $n$ replaced by $n'$ for all $n' \leq n$, $\hat T_{n+1}(p)$ can be bounded by an $n$-independent constant.  The bound on $\hat T_{n+1}(0,p)$ propagates trivially due to $\hat T_{n+1}(0,p)=\hat T_{n}(0,p/\ell)$.
\vskip 2mm

\noindent  It then remains to prove item 4). For this  we use the $L^{\infty}$-bound on $E_n(x)$ (Lemma \ref{lem: bound on e}),  and the bound
\eqref{precise bound for Tm}. 
Together they show 
$$
 \e^{\tengam \str x \str }\norm T_{n+1} (x)-\tilde T_{n+1}(x)  \norm_{\banone}    \leq C\ell^{-cn}.
$$
Consider now  $\tilde T_{n+1}(x)$ and $\tilde p$ defined in \eqref{precise bound for Tm1}. From \eqref{eq: fclosetoparabola} and \eqref{eq: bound on jp}
we get $|\Re\ell^2f_n(\tilde p/\ell)+D_np^2|\leq C$ on the support of $\chi_n$. Hence 
$$
 \e^{\tengam \str x \str }\norm \tilde T_{n+1}(x)  \norm_{\banone}    \leq  CD_n^{-d/2}\sup_{|\Re p|\leq \bana_n, |\Im p|\leq \ga_0} \norm R_n(p) \norm_{\banone} .
$$
These estimates show 4) holds uniformly in $n$.

This finishes the proof of the induction step $n\to n+1$ for Proposition \ref{ass:  properties of T} on the condition
of the bounds of Proposition \ref{prop: overview b behavior} at level $n$.
\end{proof}

\section{Flow of $G^c_A$: Linear part} \label{sec: linear rg flow}

We will now analyze the recursion relation for the correlation functions $G^c_{A}$ given in 
\eqref{eq: from n to npluseen}. 
In the present section and in Section \ref{sec: nonlinear rg flow}, we drop the subscript $n$ on operators, writing e.g.\ $G^{c}_{A} =G^{c}_{n,A}, T=T_n $ and we denote operators on scale $n+1$ by a prime, e.g.\   $G^{c'}_{A}=G^{c}_{n+1,A}$

%\baq \label{eq: repetition general recursion b intro}
%G^{c'}_{A'}   = \sum_{ \caA \in \poly(I_{A'}), \caG_{A'}(\caA) \conn }  \, \, \,     \bsS_{\ell} \caT_{A'} \left[  \mathop{\otimes}\limits_{A \in \caA} G^{c}_{A}  \otimes  \mathop{\otimes}\limits_{\tau \notin \supp \caA}  T_{\tau}  \right] 
%\eaq
We  split the recursion relation into a  \emph{linear} and a \emph{nonlinear} part (although these adjectives do not fit literally).
Consider the terms in the sum in \eqref{eq: from n to npluseen} with $ \caA$ a collection that consists of just one set, namely $ \caA=\{A\}$. 
We use the notation $A \to A'$ to indicate that $A$ contributes to $G^{c'}_{A'}$ in this sense.  In other words,
 \beq
 A \to A' \qquad  \Leftrightarrow    \qquad  A \subset I_{A'} \,\, \text{and}\,\, 
 \forall \tau' \in A': \left(A \cap I_{\tau'} \right) \neq \emptyset  . \label{eq: def contributing set}
 \eeq
The \emph{linear} RG flow is the contribution of such $ \caA$, with the additional restriction that $\str A' \str= \str A \str$: 
\beq        \label{eq: def linear rg flow}
G^{c'}_{A',\lin} : =  \, \, \,   \sum_{A: A \to A', \str A' \str= \str A \str } \, \,    \bsS_{\ell}  \caT_{A'} \left[   G^{c}_{A} \otimes  \mathop{\otimes}\limits_{\tau \in I_{A'} \setminus A} T({\tau})  \right] 
\eeq
The aim of this section is to state good bounds on  $\norm G^{c'}_{A',\lin} \norm_{10\ga_0}$, such that the induction hypothesis Proposition \ref{prop: overview b behavior} can be carried from scale to scale if  $G^{c}_{A}$ is replaced by $G^{c}_{A,\lin}$. %The \emph{nonlinear} contribution to \eqref{eq: from n to npluseen}, which is simply defined by
%\beq
%G^{c'}_{A',\nlin}= G^{c'}_{A'}-G^{c'}_{A',\lin}
%\eeq
We call the  remaining terms in the sum \eqref{eq: from n to npluseen} nonlinear and we treat them in Section \ref{sec: nonlinear rg flow}.    By inspecting Proposition \ref{prop: overview b behavior}, it is easy to understand the qualitative difference between the linear and nonlinear contributions. Recall that in Proposition \ref{prop: overview b behavior} we claim a  bound proportional to $\ep_n^{|A|}$. The nonlinear contributions to  $G^{c'}_{A'}$ carry at least
 $\str A' \str+1$ powers of  $\ep_n$ and this gives us some leeway. However, the linear contribution has just $\str A' \str$ factors of $\ep_n$ and this forces us to do a careful analysis to be able to get a bound
 with $\ep_{n+1}^{\str A'\str}$.  There are two  mechanisms at work that help us here:
\begin{itemize} \item \it{Rescaling of time}: \normalfont   Since the microscopic times described by the macrotimes $\tau$ on scale $n+1$ are $\ell^2$ larger than those on scale $n$, the blow-up factor $\dist(A)^{\al}$ for $A \to A'$ is roughly speaking $\ell^{2 \al (\str A' \str-1)}$ times larger than $\dist(A')^{\al}$. 
\item    \it{Ward Identities}: \normalfont  The Ward identities from unitarity and (in the case $\be_1=\be_2$) reversibility allow to improve our estimate on each term in the sum on the RHS of \eqref{eq: def linear rg flow}.   The gain from each of the Ward identities is roughly $\ell^{-1}$. (this will be explained later) 
\end{itemize}
The main force working against us  is
\begin{itemize} \item \it{Entropy factors}: \normalfont 
 For each $A'$, there are  $\ell^{2 \str A'\str}$ sets $A$ such that $A \to A', \str A\str=  \str A'\str$. However, due to the summability over $A$ with $\min A$ fixed in Proposition \ref{prop: overview b behavior}, only one factor $\ell^2$ (corresponding to $\min A$), shows up in our estimates. 
\end{itemize}

The gain from the  Ward Identities is explained in Section \ref{sec: smoothing}.  Then, in Sections \ref{sec: preliminaries linear rg} and \ref{sec: proof of prop linear rg noneq}, we prove Propositions \ref{prop: linear rg noneq}, \ref{prop: linear rg initial}, \ref{prop: linear rg eq}, which we present now. 

Proposition \ref{prop: linear rg noneq}  provides a bound on $G^{c'}_{A',\lin} $ which is valid regardless of whether $\initialresfinite$ is an equilibrium state or not.  For this reason we call this bound 'non-equilibrium'.  
 In what follows, we denote by $\lfloor x \rfloor$ the largest integer smaller than $x$.

\begin{proposition} [Non-equilibrium linear RG]\label{prop: linear rg noneq} 
\baq  
&& \sum_{A' \subset \bbN: \str A' \str =k,   \min A'=1 } \dist(A')^{\al}\norm G^{c'}_{A',\lin}\norm_{\tengam}  \leq   C^{k} \ep_n^{k}   \ell^{1-4\al}     \ell^{-2\al \lfloor \frac{k -3}{2}\rfloor}(\log^C\ell), \qquad  k >2     \label{eq: bound on lin more then two}  \\[2mm] 
 &&  \sum_{A' \subset \bbN : \str A' \str = 2 ,\min A'=1} \dist(A')^{\al}\norm G^{c'}_{A',\lin}\norm_{\tengam}  \leq    C   \ell^{1-2\alpha}   \ep_n^{2}  (\log^C\ell)   \label{eq: bound on lin two}
\eaq
\end{proposition}
Note that the bound in the case $\str A' \str=2$ is only useful when $\al >1/2$, since for $\al \leq 1/2$ the power of $\ell$ on the RHS becomes nonnegative. 
The next two propositions are only meaningful in the equilibrium case $\be_1=\be_2=\be$.  Proposition \ref{prop: linear rg initial} deals with correlation functions $G_{n,A}^c$ with $A \ni 0$. 

\begin{proposition} [Initial time linear RG]\label{prop: linear rg initial}
Assume $\be_1=\be_2=\be$, then 
\begin{align} \label{eq: bound on lin initial time}
&  \sum_{A' \subset \bbN_0: \str A' \str=k,   \min A' =0} \dist(A')^{\al}\norm G^{c'}_{A',\lin} \norm_{\tengam}   \leq   C^{k} \ep_n^{k} \ep^{}_{\ini,n}      \ell^{-2\al \lfloor \frac{k-1}{2}\rfloor} (\log^C\ell), \qquad k>2     \\[2mm]
&  \sum_{A' \subset \bbN_0 : \str A' \str=2,   \min A' =0}  \dist(A')^{\al}  \norm G^{c'}_{A',\lin}\norm_{\tengam}    \leq   C \ep_n^{2} \ep^{}_{\ini,n} \ell^{-\min(1,2\al)} (\log^C\ell)
\end{align}
\end{proposition}
The improved bound in Proposition \ref{prop: linear rg initial} (compared to Proposition \ref{prop: linear rg noneq}) is due to the fact that the macroscopic time $0$ does not get rescaled, and hence there is essentially no entropy factor from sets $A$ containing $0$.  

Finally, we state a bound that applies to bulk sets $A$, but that is better than Proposition \ref{prop: linear rg noneq} since it allows to treat the $\str A' \str=2 $ correlation function for $\al < 1/2$.

\begin{proposition} [Equilibrium linear RG]\label{prop: linear rg eq} 
Assume  $\be_1=\be_2=\be$ and that $ 0 < \al  \leq 1/2$. Then 
\baq  \label{eq: bound on lin equilibrium}
&& \sum_{A' \subset \bbN: \str A' \str =2,   \min A' =1} \dist(A')^{\al}\norm G^{c'}_{A',\lin}\norm_{\tengam}  \leq   C \ep_n^{2}   \ell^{-2\al}    (\log^C\ell)
\eaq
\end{proposition}
The improved bound in Proposition \ref{prop: linear rg eq} is due to the fact that we can use two Ward identities instead of one.

If one neglects the contribution of the nonlinear part in the recursion relation \eqref{eq: from n to npluseen}, then the above bounds establish the induction hypothesis Proposition \ref{prop: overview b behavior}.  Indeed, for $ 1/2 <\al <1 $ (the case $\be_1 \neq \be_2$),   Proposition \ref{prop: linear rg noneq} implies the bound \eqref{eq: induction assumption bulk} of Proposition \ref{prop: overview b behavior} provided that 
\baq
 1-2 \al   & < &   -2 \tilde \al  \\
 (1-4 \al) -2 \al \lfloor    (k-3)/2   \rfloor  &< &  -k \tilde\al, \qquad  k >2 
\eaq
and this is indeed satisfied for our choice $\tilde \al = \al/2-1/4$.    For $ 1/4 <\al  \leq 1/2$ (the case $\be_1 = \be_2$), we need (combining Propositions \ref{prop: linear rg noneq} and \ref{prop: linear rg eq})  
 \baq
  -2 \al   & < &  -2 \tilde \al  \\
 (1-4 \al) -2 \al \lfloor    (k-3)/2   \rfloor  &< &  -k \tilde\al, \qquad  k >2 
\eaq
to satisfy the bound \eqref{eq: induction assumption bulk}
and (by Proposition \ref{prop: linear rg initial})
  \baq
  -2 \al   & < &  -2 \tilde \al -\tilde\al_{\ini}  \\
 -2 \al \lfloor    (k-1)/2   \rfloor  &< &  -k \tilde\al  -\tilde\al_{\ini}  , \qquad  k >2 
\eaq
to satisfy the bound \eqref{eq: induction assumption initial}. One inspects that  $\tilde \al= \al/2-1/8$ and $\tilde\al_{\ini} =1/4$ do the job.  

\subsection{Contraction from the Ward Identities} \label{sec: smoothing}
Consider  the connected correlation function  $G^{c}_{A} \in \scrR_A$ (we have dropped the subscript $n$ again). As explained in Section \ref{sec: ward identity from unitarity}, it satisfies the Ward identity
\beq \label{eq: ward identity intro}
 \mathop{\int}\limits \d x'_\tau   ({P} \otimes 1  \ldots  \otimes 1)   G^{c}_{A}(x'_A, x_A )=0, \qquad  \tau =\max A
\eeq
where $P$ is any projector in $\banone$ of the form 
$
P = \str \mu \rangle \langle   1_{ \caS_{0}}  \str 
$ with $\sum_{s \in \caS_{0}}\mu(s)=1$, as introduced in Section \ref{sec: ward identity from unitarity}.  
We will use this identity to derive a bound on the operator 
\beq   \label{eq: powers of t following k}
  \left( T^m\otimes 1 \otimes \ldots \otimes 1 \right)  G^{c}_{A}
\eeq
where $T^m$ acts on the leg of the tensor product indexed by $\tau =\max A$. The relevance of \eqref{eq: powers of t following k} is that it comes up after writing out the RHS of \eqref{eq: def linear rg flow} explicitly, i.e.\ performing the contraction $\caT_{A'}$.  However, for the purpose of the present section, one can just accept \eqref{eq: powers of t following k} as the basic object of study. 
The reason why one can get a bound on \eqref{eq: powers of t following k} is that, formally speaking,  $T^m$ has a right
eigenvector with eigenvalue $1$ given by $ \indicator_{\caS_0}= \indicator_{[s \in \caS_{0}]}$ (in particular, it is independent of $x$) and it contracts  vectors orthogonal to  $\indicator_{\caS_0}$ to size $\sim m^{-1/2}$. 
The condition \eqref{eq: ward identity intro} ensures that $G^{c}_{n,A}$ as a function of $x'_\tau$ is orthogonal to $\indicator_{\caS_0}$  and hence \eqref{eq: powers of t following k} should decay $ \sim m^{-1/2}$.   This intuition is captured by
\begin{lemma} \label{lem: contraction from unitarity}
Let $K \in \scrR$ be a kernel satisfying the 'Ward Identity'
\beq   \label{eq: ward in lemma}
  \int_{}  \d x'  \,   P K(x',x) =0, \qquad   
\eeq
 Then, there is an operator  $V_{m,\ga} \in \scrR $ satisfying
  \beq  \label{eq: ward ready to use}
  \norm  \bsS_\ell V_{m,\ga_0/2} \norm_{10\ga_0}  \leq     \frac{ C \log(1+m) }{\sqrt{m}}, \qquad   \text{for $1 \leq m \leq \ell^2$.}
\eeq
and such that
 \beq \label{eq: inequality in ward lemma}
  \str T^m K \str   \leq   \str V_{m,\ga} \str\,  \str K \e^{\ga \str x'-x\str}   \str,
 \eeq
 where we recall that $\str L\str$ is the operator whose kernel  is the absolute value of the kernel of $L$, and the inequality in the last formula is between kernels. 
 \end{lemma}
The proof of this lemma is postponed to Section \ref{sec: proof of lem: contraction from unitarity}. 

Now it is straightforward to bound \eqref{eq: powers of t following k} (after rescaling), indeed, let $\ga= 10 \ga_0$, then
\baq  
\norm \bsS_\ell \left(T^m\otimes 1 \otimes \ldots \otimes 1 \right)  G^{c}_{A}  \norm_{10 \ga_0}    &= &    \norm \, \str  \left( T^m\otimes 1 \otimes \ldots \otimes 1 \right)  G^{c}_{A}  \str \,  \norm_{10 \ga_0/\ell}  \\
 & \leq &    \left(\norm  V_{m,\ga_0/2} \otimes 1 \otimes \ldots \otimes 1 \right)  (G^{c}_{A} \e^{\frac{\ga_0}{2} \str x'_{\max A} - x_{\max A} \str}  )  \str \,  \norm_{10 \ga_0/\ell}  \\
  & \leq &   \frac{ C \log^{C} (1+m) }{\sqrt{m}}  \norm G^{c}_{A}   \norm_{\frac{10 \ga_0}{\ell}+ \frac{\ga_0}{2}} 
\eaq
The first inequality uses \eqref{eq: inequality in ward lemma} with the coordinates $z_{A \setminus \max A}, z'_{A \setminus \max A}$ kept fixed.  The second inequality uses \eqref{eq: ward ready to use} and the definition of the norm $\norm \cdot \norm_{\ga}$. 

Let us now try an analogous trick     based on the Ward identity \eqref{eq: ward identity with L} which we rewrite here with the help of the projectors $P^{\be}:= \str \mu^{\be} \rangle \langle 1_{\caS_0}\str$ and $ P^{\reff}= \str \mu^{\reff} \rangle \langle 1_{\caS_0}\str$;
\beq  \label{eq: ward identity with L repeat}
\int \d x_a   \, G^{c}_{\{1, \tau\}}  ( { x_a}', x_b'; {x}_a,    x_b) (1 \otimes P^{\be})  =      \int \d x_a L_{ \tau} ( { x_a}', x_b'; {x}_a,    x_b)  (1 \otimes P^{\reff}   )
\eeq
It will be used to get smallness for the right action by $T^m$ i.e. for $ G^c_{ \{ 1,\tau \}}  ( 1 \otimes T^{m} ) $ by a
similar reasoning to the one given above for the Ward identity from unitarity. The operator $T$ formally has a unique left eigenvector with eigenvalue $1$, given by $\mu_{T_n}$ as defined in \ref{eq: projector r explicit}, 
and $T^m$   contracts  vectors orthogonal to $\mu_{T_n}$ to size $\sim {m}^{-1/2}$.  Hence, if \eqref{eq: ward identity with L repeat} would hold with the RHS replaced by  $0$ and  $P^{\be}$ replaced by $R(0)$, then we could repeat the reasoning of the unitarity Ward identity.  However, the projection $P^{\be}=|\mu^\be\rangle
\langle 1_{ \caS_{0}}|$ is $\sqrt{\ep_0}$-close to $R(0)= |\mu_{T_n}\rangle\langle 1_{ \caS_{0}}|$, see \eqref{eq: projector r explicit bound} in the induction hypothesis. Actually $\mu_{T_n} $ converges to $\mu^\be$ as $n \to \infty$, but this is not even needed here. The presence of the RHS in \eqref{eq: ward identity with L repeat} introduces an additional error term which is small because $L_{ \tau}$ consists of correlation functions that  include $\tau=0$ and therefore contract faster.

\begin{lemma}   \label{lem: contraction from reversibility}
Let $K, \tilde K \in \scrR$ be kernels satisfying the 'Ward Identity'
\beq
\int \d{x}  \, K(x', x) P^{\be}  =  \int \d{x} \,   \widetilde K(x',x )  P^{\reff}
\eeq
 Then, there are operators    $W_{m, \ga}, \widetilde W_{m, \ga} \in \scrR$ satisfying 
 \baq
\norm \bsS_\ell W_{m, \ga_0/2}  \norm_{\tengam}  &\leq&   C m^{-1/2} \log m + C \norm  R(0)-  P^{\be} \normba, \\[2mm]
  \norm \bsS_\ell  \widetilde W_{m, \ga_0/2} \norm_{\tengam}   &\leq&  C
\eaq
 for $1 \leq m \leq \ell^2$, 
 and 
\begin{align}
 &  \str  K T^m\str  \leq    \str K \e^{\ga \str x-x'\str}  \str    W_{m, \ga}   +    \str \widetilde K  \e^{\ga \str x-x'\str} \str  \widetilde W_{m, \ga}   \label{eq: bound gibbs smoothing on correlations}
\end{align}
\end{lemma}
\noindent We will use the Lemma to derive a bound on the correlation function
$
( T^{m_+} \otimes 1 )  G^c_{ \{ 1,\tau \}}  ( 1 \otimes T^{m_-} )   % \label{eq: correlation function needing two wards}
$.
Its kernel is bounded
\baq  
  && \str  \left(T^{m_+}\otimes 1 \right)  G^{c}_{\{1, \tau \}}   \left(1 \otimes T^{m_-}  \right)    \str 
   \leq    V_{m_+, \ga}   \str  \e^{\ga \str x'_\tau-x_\tau \str }  G^{c}_{\{1, \tau \}}   \left(1 \otimes T^{m_-}  \right)  \str \,     \\[2mm]
 &&\leq     V_{m_+, \ga}  \left(   \str  \e^{\ga (\str x'_\tau-x_\tau \str+\str x'_1-x_1 \str)  }  G^{c}_{\{1, \tau \}} \str W_{m_-, \ga}   +  \str  \e^{\ga \str x'_1-x_1 \str  }  L_\tau \str \widetilde W_{m_-, \ga}      \right)  
 \eaq
 where we used first \eqref{eq: inequality in ward lemma} and then \eqref{eq: bound gibbs smoothing on correlations}. 
Collecting the bounds on $V, W, \tilde W$, we get
\baq  
&& \norm\bsS_\ell  \left(T^{m_+}\otimes 1 \right)  G^{c}_{\{1, \tau \}}  \left(1 \otimes T^{m_-}  \right)   \norm_{10 \ga_0} \nonumber \\
& \leq &  C \frac{\log m_+}{\sqrt{m_+}}   \norm   \str  \e^{\frac{\ga_0}{2 }\str x'_\tau-x_\tau \str }  G^{c}_{\{1, \tau \}}   \left(1 \otimes T^{m_-}  \right)  \norm_{10 \ga_0/\ell} \nonumber \\
& \leq &  C \frac{\log m_+}{\sqrt{m_+}}  \left( \norm  G^{c}_{\{1, \tau \}} \norm_{\frac{10 \ga_0}{\ell}+ \frac{\ga_0}{2}} ( \frac{\log m_-}{\sqrt{m_-}} +  \norm  R(0)-  P^{\be} \normba  ) +   \norm L  \norm_{\frac{10 \ga_0}{\ell}}    \right)  \label{eq: double smoothing}
\eaq
The operator norm $ \norm L_\tau \norm_{\ga} $ can be bounded by norms of correlation functions: 
\beq
 \norm L_\tau \norm_{\ga}  \leq   \norm G^{c}_{\{0,1, \tau\}} \norm_{\ga}  +  \norm  G^{c}_{\{0, \tau\}} \norm_{\ga}  +    \norm G^{c}_{\{0, \tau-1 \}}  \norm_{\ga}   \label{eq: bound on L in correlation functions}
\eeq
by \eqref{eq: ward identity with L} and Lemma \ref{lem: main property of norms} (1).  
The upshot of  the calculation in  \eqref{eq: double smoothing} is that all terms between the large brackets on the last line are smaller than $ \norm  G^{c}_{\{1, \tau \}} \norm_{\frac{10 \ga_0}{\ell}}$ (which would be the resulting bound if we had used only the unitarity Ward identity). As anticipated before Lemma \ref{lem: contraction from reversibility}, the smallness comes either from the factor $ \frac{\log m_-}{\sqrt{m_-}}$, the difference $\norm  R(0)-  P^{\be} \normba$  and the correlation functions contributing to $L_\tau$ that are small because they involve the initial time $0$ and therefore contract faster, see Proposition \ref{prop: linear rg initial}.

\subsubsection{Proof of Lemma \ref{lem: contraction from unitarity} } \label{sec: proof of lem: contraction from unitarity}

We start from \eqref{eq: ward in lemma}, i.e.\ 
$\int  \d x \,  P K(x,x_0)=0$ for any $ x_0$. This implies trivially that
\beq
T^m K(x',x_0)=
\int \d x  \,  \e^{-\ga \str x-x_0 \str} \left( T^m(x',x) -  T^m(x', x_0)P \right)   K(x,x_0)   \e^{\ga \str x-x_0 \str}  \label{eq: ward as equality}
\eeq
 We now choose the projector $P$ to equal $R(0)$ (the spectral projector of $\hat T(p=0)$) and we define
  $V_{m, \ga}\in \scrR$ 
\beq
V_{m, \ga} (x',x)  := \sup_{x_0}  \left\str  \e^{-\ga \str x-x_0 \str} \left( T^m(x- x') -    T^m(x_0-x')  R(0)  \right)  \right\str  \label{def: operator v}
\eeq
where  the absolute value $\str \cdot \str$ is applied to a kernel in $s,s'$, that is, for an operator $D$ in $\scrG$ we set $\str D \str$ to be the operator with kernel $\str D \str(s',s) :=\str D (s',s) \str  $. 
Note that $V_{m, \ga}$ is translation invariant by the translation invariance of $T^m$. 
From \eqref{eq: ward as equality} and \eqref{def: operator v}, we get the inequality \eqref{eq: inequality in ward lemma} of Lemma \ref{lem: contraction from unitarity} 

It remains to establish the bound on $V_{m, \ga} $, i.e.\ \eqref{eq: ward ready to use}. It suffices to consider
$\sqrt{m}$ large enough only, we take $\sqrt{m}\geq 80$.
Since the operator $V_{m, \ga}$ is translation invariant, we can use Lemma \ref{lem: banone equal to weird} and hence 
it suffices to show that 
\beq \label{eq: sup bound on v to obtain}
s_m:=\sup_{x}  \e^{\frac{\twentygam}{\sqrt{m}} \str x \str } \norm  V_{m, \ga_0/2}(0,x)  \norm_{\banone}  \leq  C m^{-1/2(d+1)} \log m.
\eeq
Indeed, we then get 
 \beq
\norm V_{m, \ga_0/2} \norm_{10\ga_0/\ell}  \leq s_m  \int   \d x \,     \e^{ (\frac{10\ga_0}{\ell}   -\frac{\twentygam}{\sqrt{m}})  \str x \str } 
\leq C m^{-1/2} \log m
\eeq
which implies \eqref{eq: ward ready to use} by scaling.

To prove \eqref{eq: sup bound on v to obtain}
 we split $V_{m, \ga} (x',x)=V'_{m, \ga} (x',x)+V''_{m, \ga} (x',x)$ where the terms are defined 
by restricting the supremum in \eqref{def: operator v}  to  $\frac{\ga_0}{2} \str x-x_0 \str \leq  \log m $ and $\frac{\ga_0}{2} \str x-x_0 \str >  \log m $ respectively.  
We bound 
\baq
\e^{\frac{20\ga_0}{\sqrt{m}}  \str x\str}  \str V''_{m, \ga_0/2} (0,x) \str &  \leq & \e^{\frac{20\ga_0}{\sqrt{m}}  \str x\str}   \sup_{x_0: \ga_0 \str x-x_0 \str >  \log m}  \e^{-\frac{\ga_0}{2} \str x-x_0 \str} \left\str T^m(x) -    T^m(x_0)  R(0)   \right\str    \nonumber \\[1mm]
& \leq &    \frac{1}{m} \e^{\frac{20\ga_0}{\sqrt{m}}  \str x \str}  \str T^m(x)  \str  + \sup_{x_0}\frac{1}{\sqrt{m}}   \e^{-\frac{\ga_0}{4} \str x-x_0 \str}  \e^{\frac{20\ga_0}{\sqrt{m}}  \str x_0 \str} \left\str  T^m(x_0)  R(0)   \right\str  \nonumber
\eaq
where we used the triangle inequality $\str x\str \leq \str x_0\str  + \str x-x_0\str  $ and $\sqrt{m}\geq 80$ in the last step.  The desired bound \eqref{eq: sup bound on v to obtain} (without the log)
now follows from \eqref{eq: linfinity bound for m repeated t}.

Next, we consider $V'_{m, \ga_0/2}$.  The argument for showing \eqref{eq: sup bound on v to obtain} for  $V'_{m, \ga_0/2}$ is   analogous to the proof of Lemma \ref{lem: powers of t} and we will use some notation from that proof. 
We start from the momentum space representation
\beq
T^m(x',x) -  T^m(x',x_0)R(0) =   (1/2\pi)^d \int_{\bbT_n} \d p   \,    \e^{\i p (x-x')}  \left[  \hat T^m(p) -  \hat T^m(p)  R(0)  \e^{- \i p (x-x_0)}    \right]
\eeq
and proceed as in \eqref{smallpbound} and  \eqref{largepbound} to conclude that up to terms exponentially
small in $m$ (and $n$) it suffices to bound
\beq%\label{precise bound for Tm1}
%\tilde T_n(x)=e^{-{20\ga_0}|x|}
\left\norm\int_{\bbT_{n}} \d p   \,  e^{mf_n(\bar p)}  \e^{\i p x} (R (\bar p)-R (0)e^{-\i\bar p(x-x_0)})\chi_n(p)\right\normba
%\leq  \int_{\bbT_{n+1}} \d p   \,  |e^{mf_n(\bar p)}]  (\norm R (\bar p)-R (0)\normba +|1-e^{-\i\tilde p(x-x_0)}|)\chi(p).
  \eeq
where $\bar p=p+ \frac{20\i\ga_0}{\sqrt{m}}e_x$. This is bounded by
\beq \label{eq: bounds ward leading}
C \int_{\bbT_{n}} \d p   \,  e^{-\frac{D_n}{2}m  p^2} (\norm R (\bar p)-R (0)\normba +|1-e^{-\i\tilde p(x-x_0)}|)\chi_n(p).
  \eeq
using the bound \eqref{upper and lower bound for ef}.
Since $$|1-e^{-\i\tilde p(x-x_0)}|\leq C(|p|+\ga_0/\sqrt{m})\log m$$ 
the second term on the RHS of  \eqref{eq: bounds ward leading} contributes $\caO(m^{-(d+1)/2}\log m)$. For the first term use analyticity of $R (p)$
in a ball of radius $\ga_0$ at the origin to get 
$$\norm R (\bar p)-R (0)\normba\leq C( |p|+\ga_0/\sqrt{m}+1_{|\Re p|>c\ga_0})$$
and hence the first term on the RHS of \eqref{eq: bounds ward leading}  contributes $\caO(m^{-(d+1)/2})$, as well.

\subsubsection{Proof of Lemma \ref{lem: contraction from reversibility}} \label{sec: proof of lemma contraction from reversibility}

We start from the Ward Identity 
\beq
\int \d x' K(x'', x') P^{\be}  =  \int \d x'  \widetilde K(x'',x' )  P^{\reff} 
\eeq
Then 
\baq
 && \int \d x' K(x'',x' ) T^m(x', x)  \\
 &=&  \int \d x' K(x'',x' ) \e^{\ga \str x''-x'\str} \left( T^m(x', x) - P^{\be}  T^m(x'',x)     \right) \e^{-\ga \str x''-x'\str} \\
 &+&    \int \d x' \widetilde K(x'',x' )  \e^{\ga \str x''-x'\str}   P^{\reff}  T^m(x'',x) \e^{-\ga \str x''-x'\str}
\eaq
We now define the operators $W_{m, \ga}, \widetilde W_{m, \ga}$ by specifying their reduced kernels
\baq
 W_{m, \ga}(x',x) &:= & \sup_{x''}   \left\str T^m(x',x) - P^{\be}  T^m(x'',x)     \right\str \e^{-\ga \str x''-x'\str}    \label{def: v for reversibility} \\
  \widetilde W_{m, \ga}(x',x) &:= & \sup_{x''}   \left\str  P^{\reff}  T^m(x'',x)  \right\str \e^{-\ga \str x''-x'\str}  
\eaq
(the absolute values on the RHS are again meant as absolute values of  kernels on $\bbA_n \times \bbA_n$, see the remark below \eqref{def: operator v})
Note that both operators are translation invariant. 
By  completely analogous reasoning as the one leading to \eqref{eq: inequality in ward lemma}, we get the inequality \eqref{eq: bound gibbs smoothing on correlations}.

Now to the bounds on $W_{m, \ga}, \widetilde W_{m, \ga}$. 
If $P^{\be}$ is replaced by $R(0)$ in the definition of $W_{m, \ga}(x',x) $, we get the first term of the bound by an analogous proof as that of Lemma \ref{lem: contraction from unitarity}. The error term due to $R(0)-  P^{\be}$  is estimated by  $ C\norm  R(0)-  P^{\be} \normba$ where the constant $C$ originates from bounding $\bsS_{\ell}T^m$ with the help of Lemma \ref{lem: powers of t}.  Similarly, the bound on $\widetilde W_{m, \ga_0/2}$ is immediate from Lemma \ref{lem: powers of t}.

\subsection{Preliminaries for the induction step $G^c_{A} \to G^{c'}_{A'} $} 
\label{sec: preliminaries linear rg}
In this section, we gather some tools for the proofs of Propositions  \ref{prop: linear rg noneq}, \ref{prop: linear rg initial}, \ref{prop: linear rg eq}.   
We state the main bound that we will use to control the sum over correlation functions at  the lower scale. We abbreviate $k:= \str A \str = \str A' \str$  and $m_+ := \max I_{A'}- \max A $;
\baq
\norm  G^{c'}_{A', \lin}    \norm_{10 \ga_0}   & \leq &  \sum_{A \to A', \str A \str =k} \left\norm \bsS_{\ell} \caT_{A'} \left[    G^{c}_{A}  \mathop{\otimes}\limits_{\tau \in I_{A'} \setminus A} T({\tau})  \right]  \right\norm_{10 \ga_0}    \nonumber  \\
&\leq&   \sum_{A \to A', \str A \str =k}  \norm (T^{m_+} \otimes 1 \otimes \ldots \otimes 1) G^{c}_{A} \norm_{10 \ga_0 /\ell}  \prod_{j =1}^{k}   \norm  T^{\tau_{j}-\tau_{j-1} -1} \norm_{10 \ga_0/\ell}, \qquad    \nonumber  \\
&\leq& 
 \sum_{A \to A', \str A \str =k} C^{k}  \frac{\log (1+m_+)}{\sqrt{1+m_+}}  
 \,  \norm   G^{c}_{A}   \norm_{\ga_0}      \label{eq: basic bound linear rg}
\eaq
On the second line, we set the dummy $\tau_0=0$ or, if $\tau_1=0$, then $\tau_0=-1$.
The bound is obtained by proceeding as in the proof of Lemma \ref{lem: bound on e},
 bounding powers $T^{m}$ by Lemma \ref{lem: powers of t}, except the power $T^{m_+}$, which is bounded by  using Lemma \ref{lem: contraction from unitarity}, as explained in the beginning of Section \ref{sec: smoothing}.

Next,  we introduce some useful notation. We let $h(A)$ stand for a function of the ordered times $\tau_1, \ldots,  \tau_k$ of $A$ with the property that, for $\min A >0$, 
 \baq
 &&  h(A) =  h(A+\tau)   \qquad   \textrm{(time-translation invariance)}   \label{eq: translation invariance of h}
 \eaq
and satisfying the normalization conditions, uniformly in $\tau_k,\tau_1$, respectively,  
\beq
\mathop{\sum}\limits_{0<\tau_1 < \tau_2 < \ldots < \tau_k: \tau_k \, \textrm{fixed}} h(A)  \leq 1, \qquad  \textrm{and} \qquad \mathop{\sum}\limits_{\tau_1 < \tau_2 < \ldots < \tau_k: \tau_1 \, \textrm{fixed}} h(A)  \leq 1  \label{eq: normalization of h}
\eeq
Note that the first condition follows from the second by time-translation invariance.
In particular, we will use that the function of $k-1$ variables obtained by performing the sum $\sum_{\tau_1: \tau_1< \tau_2} h(A)$ or  $\sum_{\tau_k: \tau_k> \tau_{k-1}} h(A)$ satisfies the same conditions, and hence we will also call it $h(\cdot)$. 
 This is how $h(\cdot)$ will enter: we write
\beq      \label{eq: how h emerges}
  \norm  G^{c'}_{A', \lin}    \norm_{10 \ga_0}   \leq  \sum_{A \to A'}  \ep^{\str A \str}   \frac{\log (1+m_+)}{\sqrt{1+m_+}}     \dist(A)^{-\al}  h(A).      
\eeq   
by using \eqref{eq: basic bound linear rg}. Now our task is to estimate the RHS, which is done by elementary analysis in the next sections (note indeed that all operators have disappeared from the RHS). Below, we consistently use the notation $\tau_1, \tau_2, \ldots$ for the ordered elements of $A$ and $\tau'_1, \tau'_2, \ldots$ for those of $A'$. The variables $\tau_1, \tau_2, \ldots$ range over the sets $I_{\tau'_1}, I_{\tau'_2}, \ldots$

\subsubsection{Change of variables}\label{sec: change of variables}

We need to sum \eqref{eq: how h emerges}.
Hence we will have to convert the factor  $\dist(A')^{\al}$ into $\dist(A)^{\al}$. In doing this, we will gain small factors of $\ell^{-2\alpha}$.  Indeed, we find the inequality
\beq
\dist(\tau'_1, \tau'_2 )^{\al}  \leq \ell^{-2\al}   \dist(A)^{\al}, \qquad   A \to \{\tau'_1, \tau'_2 \}, \,  \tau'_2- \tau'_1 >1
\eeq
That is, if the macroscopic times are not neighbors, then we gain a small factor.  However, even if the macroscopic times are neighbors,  then we still get a small factor from every second difference of macroscopic times, i.e.\ 
\beq
\dist(A')^{\al} \leq C^{\str A' \str} \left(\ell^{-2 \al}\right)^{\lfloor \frac{\str A' \str-1}{2} \rfloor}  \dist(A)^{\al}, \qquad \textrm{for any} \, \, A \to A' \, \, \textrm{with} \, \,  \str A' \str =\str A \str \geq 3
\label{eq: inequality change of variables}
\eeq 
Indeed, for all $1<j<|A'|$, either $|\tau_j-\tau_{j+1}|\geq \frac{\ell}{2}|\tau'_j-\tau'_{j+1}|$
or $|\tau_j-\tau_{j-1}|\geq \frac{\ell}{2}|\tau'_j-\tau'_{j-1}|$.   
We will refer to these bounds as 'change of variables'.

\subsection{Proof of Proposition \ref{prop: linear rg noneq}}\label{sec: proof of prop linear rg noneq}

To prove Proposition \ref{prop: linear rg noneq}, we need to control the sum
\beq
 \sum_{A': \str A' \str=k,   \max A'   \mathrm{fixed}  }   \dist(A')^{\al} \sum_{A \to A'}   \frac{1}{\sqrt{1+\max I_{\tau'_{k}} -\tau_k    }}  
     \dist(A)^{-\al} h(A)    \label{eq: sum linear rg to perform}
\eeq
which emerges\footnote{In fact, Proposition \ref{prop: linear rg noneq} demands that we keep $\min A'$ fixed. However, by time-translation invariance, the claim with $\max A'$ fixed is equivalent} from the LHS of (\ref{eq: bound on lin more then two},\ref{eq: bound on lin two}) by  \eqref{eq: how h emerges}. 
We will do this by using the change of variables and the integrability of $h(\cdot)$, both introduced above.   Note that we dropped the factor $\log^C (1+ m_+)$ since it can always be estimated by $C' \log^C \ell$ which suffices for our purposes.  We also did not write $\ep_n^{\str A \str}$ since this factor just goes through all estimates.

\subsubsection{The case $\str A' \str= 2$}
Let us assume first that $ \tau'_1 < \tau'_2 -1$, then \eqref{eq: sum linear rg to perform} reduces to

\baq
&&  \sum_{\tau'_1: \tau'_1 < \tau'_2 -1} (1+\tau'_2- \tau'_1)^{\al} \sum_{\tau_1 \in I_{\tau'_1}, \tau_2 \in I_{\tau'_2} }               h(\tau_1, \tau_2)    (1+ \tau_2-\tau_1)^{-\al}     (\max I_{\tau'_2} -\tau_2+1)^{-1/2}  \\[2mm]
&\leq& C\ell^{-2\al} \sum_{\tau_1 < \tau_2-\ell^2}  \sum_{\tau_2 \in I_{\tau'_2} }               h(\tau_1, \tau_2)      (\max I_{\tau'_2} -\tau_2+1)^{-1/2}  \\[2mm]
&\leq& C \ell^{-2\al} \sum_{\tau_2 \in I_{\tau'_2} }     (\max I_{\tau'_2} -\tau_2+1)^{-1/2}   = C\ell^{-2\al}  \sum_{x=1}^{\ell^2} x^{-1/2} \leq  C  \ell^{1-2\al}    
\eaq
To get the first inequality, we used the change of variables formula \eqref{eq: inequality change of variables}, the second inequality follows from the properties of $h(\cdot)$, i.e., from \eqref{eq: normalization of h}.

Next, let us treat the case $\tau'_2=\tau'_1+1$, then \eqref{eq: sum linear rg to perform} gives (we can bound $\dist(A')^{\al}$ by a constant here)
\baq
&& C  \sum_{\tau_1 \in I_{\tau'_1}, \tau_2 \in I_{\tau'_2} }               h(\tau_1, \tau_2)    (1+ \tau_2-\tau_1)^{-\al}     (\max I_{\tau'_2} -\tau_2+1)^{-1/2}  \\[2mm]
&\leq& C    \sum_{\tau_2} (1+\tau_2-\min I_{\tau'_2})^{-\al}  (\max I_{\tau'_2} -\tau_2)^{-1/2}= 
  \sum_{x=1}^{\ell^2} x^{-\al}  (1+\ell^2-x)^{-1/2}
     \leq  C       \ell^{1-2\al}    
\eaq
In the  first inequality, we used  $(1+ \tau_2-\tau_1)^{-\al} \leq (1+ \tau_2-\min I_{\tau'_2})^{-\al}  $, and we bounded the sum over $\tau_1$ by \eqref{eq: normalization of h}. The last inequality follows easily after splitting the sum in $ 1< x \leq \ell^2/2 $ and $ \ell^2/2 < x\leq \ell^2$. 
Putting the two contributions together,  we get 
$C \ell^{1-2\al}    $ as a bound for \eqref{eq: sum linear rg to perform} in the case $\str A \str=\str A' \str= 2$. \\

\subsubsection{The case $\str A' \str= 3$} \label{sec: the case of three times}
We start with the subcase  $\tau'_1=\tau'_2-1=\tau'_3-2$ which turns out to be the most tricky one.  
First, we dominate
\baq
&& \sum_{\tau_{1,2,3} \in I_{\tau'_{1,2,3} } }  \dist(\tau_1, \tau_2, \tau_3)^{-\al}  (1+\max I_{\tau'_3}- \tau_3)^{-1/2}  h(\tau_1,\tau_2, \tau_3)   \\
&\leq & \sum_{\tau_{2,3} \in I_{\tau'_{2,3} } }   \dist(\tau_2, \tau_3)^{-\al}(1+  \tau_2 -\min I_{\tau'_2} )^{-\al}   (1+\max I_{\tau'_3}- \tau_3)^{-1/2}  h(\tau_2,\tau_3) \eaq
by $(1+\tau_2-\tau_1)^{-\al} \leq  (1+\tau_2-\min I_{\tau'_2})^{-\al}$ and the properties of $h$, i.e.\ \eqref{eq: normalization of h}.
Then, we change variables 
\baq
x=    \tau_2 -\min I_{\tau'_2}, \qquad    y =  \max I_{\tau'_3}- \tau_3 , \qquad   0 \leq x,y < \ell^2  
\eaq
such that we have to bound
\beq 
  \sum_{x,y=0}^{\ell^2-1}  F(x,y), \qquad    F(x,y) =   h(x, 2 \ell^2-y-1)      ( x +1)^{-\al}  (y+1)^{-1/2}   (2 \ell^2- x-y-1)^{-\al}
\eeq
where we used the invariance of $h(\cdot, \cdot)$ under joint translations of its arguments.

To perform these sums, we define the sets
\baq
\caC^{a}_n   &:= & \{ (x,y) \big\str \, (1- \delta_{n,0})  2^{n}  \leq  x \leq 2^{n+1} ,   \quad     (1- \delta_{n,0})  2^n  \leq  y  \leq 2/3\ell^2               \}    \\[2mm]
\caC^{b}_n   &:= & \{ (x,y) \big\str \,(1- \delta_{n,0})   2^{n}   \leq  y \leq 2^{n+1} ,   \quad  (1- \delta_{n,0})  2^n  \leq  x  \leq 2/3 \ell^2             \} \\[2mm]
\caC^{c}   &:= & \{  (x,y) \big\str \, 2   \ell^2-x-y  \leq 2/3 \ell^2             \} 
\eaq
And we will split 
\beq
\mathop{\sum}\limits_{x,y=0}^{\ell^2-1}    F(x,y)  \leq \sum_{n=0}^{n^*} \mathop{\sum}\limits_{\caC^{a}_n} F(x,y) +  \sum_{n=0}^{n^*} \mathop{\sum}\limits_{\caC^{b}_n} F(x,y)  +  \mathop{\sum}\limits_{\caC^{c} } F(x,y)    \label{eq: split it a strips}
\eeq
for $n=0, \ldots  n^*$ with $n^*$ the smallest natural number such that $2^{n^*}  \geq \ell^2$. 
To perform the sum in $\caC^{a}_n $,  we bound    
\beq
  \sup_{\caC^{a}_n}\left( (2 \ell^2- x-y-1)^{-\al} ( x +1)^{-\al}  (y+1)^{-1/2} \right) = C \ell^{-2\al} (2^n)^{-1/2-\al},
\eeq 
 we use properties \eqref{eq: normalization of h} to perform the $y$-sum and we bound the $x$-sum by $2^n$, the 'width' of its  domain.  This yields
\beq
  \sum_{\caC^{a}_n } F(x,y)  \leq C \ell^{-2\al} (2^{n} )^{1-1/2-\al}
\eeq
The sum in $\caC^{b}_n $ is done analogously, except that now the $x$-sum is controlled by \eqref{eq: normalization of h} and the $y$-sum by its width $2^n$.
On  $ \caC^{c} $, we can bound $( x +1)^{-\al}  (y+1)^{-1/2}  \leq \ell^{-1-2\al}$, and then the sum is done straightforwardly as
\baq
 \sum_{\caC^{c} }F(x,y)   &\leq&    \ell^{-1-2\al}   \sum_{x,y=0}^{\ell^2-1}   h(x, 2 \ell^2-y)     (2 \ell^2- x-y-1)^{-\al}  \\[2mm]
&\leq&   \ell^{-1-2\al} \sum_{x=0}^{\ell^2-1}    ( \ell^2- x)^{-\al}  \leq  \ell^{1-4\al}
\eaq
The sum over $n = 1, \ldots, n^* $ yields 
\baq
  \sum_{x,y=0}^{\ell^2-1}  F(x,y)  &\leq &  \ell^{1-4\al} +  \ell^{-2\al}  \sum_{n \geq 1}^{n^*}     (2^{n} )^{1-1/2-\al}  \\
 &  \leq &  C   \ell^{-2\al}      (2^{1/2-\al})^{n^*}   \leq  C \ell^{1-4 \al}
\eaq

Let us now look at the case where one pair of  $\tau'_1, \tau_2, \tau'_3$, is consecutive. If $ 1+ \tau'_1 < \tau'_2= \tau'_3-1$, then we get a factor $\ell^{-2\al}$ from $(1+ \tau_2-\tau_1)^{\al}$ by change of variables, we integrate $\tau_1$ by using the normalization of $h(\cdot)$ and we can treat the remaining times $\tau'_2, \tau'_3$ as outlined in the case $\str A' \str =2 $, gaining an extra factor $\ell^{1-2 \al}$. 
The other possible case with one consecutive pair, i.e.\ $ 1+ \tau'_1 = \tau'_2 <  \tau'_3-1$, can be related to the previous case by symmetry considerations, using the translation invariance of $h(\cdot)$.
%  Indeed,  since $h$ depends on the differences only,  one can replace the factor $ (\max I_{\tau'_3}-\tau_3+1)^{-1/2}$ by $(1+ \tau_1- \min I_{\tau'_1})^{-1/2}$ when summing over $\tau_{1,2,3}$ with $\tau'_{1,2,3}$ fixed. Then the problem reduces to the case $ 1+ \tau'_1 < \tau'_2= \tau'_3-1$. 
Finally, the case where no pair is consecutive is of course analogous to the corresponding case with $ \str A' \str=2$: One gets $\ell^{-4 \al}$ from  change of variables, and the remaining sum over $\tau_{1,2,3}$ yields $\ell$ by using the normalization of $h(\cdot)$ and the factor  $ (1+ \max I_{\tau'_3}-\tau_3)^{-1/2}$.\\

Hence in all subcases with $\str A'\str=3$, we obtain a factor $\ell^{1-4 \al}$.

\subsubsection{The case $\str A' \str=k >3$}
We perform the sum over $\tau_1, \tau_2, \ldots, \tau_{k-3}$.  As argued in \eqref{eq: inequality change of variables},  we get at least one small factor $\ell^{-2\al}$ from every second sum.  This yields the bound  $ \ell^{-2\al \lfloor \frac{\str A' \str -3}{2}\rfloor}$. The last three sums are then bounded by $C \ell^{1-4\al} $ by repeating the analysis of the case $\str A' \str=3$.
 
Combining all cases, we get the claim of Proposition \ref{prop: linear rg noneq}

\subsection{Initial time linear RG: Proof of Proposition \ref{prop: linear rg initial}}
We start as in Section \ref{sec: proof of prop linear rg noneq} from the bound \eqref{eq: sum linear rg to perform}.
Now  we perform the sum over times $\tau_1, \ldots, \tau_k$ starting from $\tau_k$, keeping in mind  that $\tau_1=0$.

In the case $k= |A'| >2$, we get a factor $\ell^{-2\alpha \lfloor \frac{k-1}{2}\rfloor}$ by the change of variables, and the sum is performed by using the summability of the function $h(\tau_1, \ldots, \tau_{k})$ with $\tau_1=0$ fixed. In contrast to the proofs above, this suffices since $\tau_1$ is indeed fixed here: $\tau_1=0$.
For the case $k =|A'| =2$, we  perform the sum over $\tau_2$ by splitting it in the regions $\tau_2 \geq \ell^2/2$ and $\tau_2 < \ell^2/2$. For $\tau_2 \geq \ell^2/2$, we can use the change of variables, and hence we get 
\baq
 \textrm{\eqref{eq: sum linear rg to perform}}  & \leq &  C    \sum_{ \tau_2  \geq \ell^2/2}  \ell^{-2\al}  h(0, \tau_2)  +  C  \sum_{1 \leq \tau_2 <\ell^2/2}     (\ell^2-\tau_2)^{-1/2}   (1+\tau_2)^{-\al}   h(0, \tau_2) \nonumber  \\
 & \leq &  C   \ell^{-2\al}   +  C \ell^{-1} 
\eaq
This finishes the proof. 

\subsection{Equilibrium linear RG: Proof of Proposition \ref{prop: linear rg eq} }

We have to perform the sum
\beq
 \sum_{\tau'_1  : \tau'_1 < \tau_2'} \dist (\tau'_1, \tau_2')^{\al}  \norm G^{c'}_{\{ \tau'_1, \tau'_2\}, \lin } \norm_{\tengam}    \label{eq: intergoal equilibrium}         \eeq
 Let us  abbreviate 
\beq
m_+ =    \max I_{\tau'_2} - \tau_2, \qquad     m_- =   \tau_1 -  \min I_{\tau'_1} 
\eeq
Then we obtain 
 \baq
  \norm G^{c'}_{\{ \tau'_2, \tau'_2\}, \lin } \norm_{\tengam}   & \leq &   \sum_{ \{ \tau_1, \tau_2\} \to \{ \tau_1', \tau'_2\} }   C  \frac{\log (1+m_+) }{(1+m_+)^{1/2} }  \Big(  \frac{\log (1+m_-) }{(1+m_-)^{1/2} }    \norm G^c_{\{ \tau_1, \tau_2\} } \norm_{\ga_0}   \nonumber \\[3mm]
 &    &   \qquad  \qquad  +       \,    \norm G^c_{ \{0, 1, \tau_2-\tau_1+1 \}}  \norm_{\ga_0}     +   \norm G^c_{ \{ 0,  \tau_2-\tau_1+1 \} }  \norm_{\ga_0}  + \norm G^c_{ \{ 0,  \tau_2-\tau_1 \} }  \norm_{\ga_0}    \Big)   \label{eq: reversibility in three terms bound}
\eaq
 where we used the time-translation invariance property $G^c_{ \{\tau_1, \tau_2 \}}= G^c_{ \{1, \tau_2-\tau_1+1 \}}$, the bounds \eqref{eq: double smoothing} and   \eqref{eq: bound on L in correlation functions},  and 
the bound \eqref{eq: projector r explicit bound} for
$\norm P^{\be}-R(0) \normba$. 

 Next, we perform the sum over $\tau_1'$ of the RHS in \eqref{eq: reversibility in three terms bound}.  This RHS is  split in four terms. The $\tau'_1$-sum of the last three terms  is bounded brutally by $C\ell^2  \ep^{2}_n \ep_{\ini,n} $ by using the integrability of the function $h(\cdot)$ to perform the sum over $\tau_2-\tau_1$ and estimating  the sum over $\tau_1$ by $\ell^2$.  By the smallness of $\ep_{\ini,n}$,  this bound is sufficient for  Proposition \ref{prop: linear rg eq}. 
Next, we focus  on the ($\tau'_1$-sum of the) first term in \eqref{eq: reversibility in three terms bound}. It is of the form
\baq
 \sum_{\tau'_1  : \tau'_1 < \tau_2'}  \dist(\tau'_1, \tau'_2)^{\al} \sum_{\{ \tau_1, \tau_2\} \to \{ \tau_1', \tau'_2\}}   C (1+m_+)^{-1/2}      (1+m_-)^{-1/2}  \dist(\tau_1, \tau_2 )^{-\al} h(\tau_1,\tau_2)   \label{eq: correlation to be estimated for eq}
\eaq
To evaluate this sum, let us first consider the case $\tau'_1 < \tau_2'-1$. We set 
\beq
z :=  \tau_2 -\tau_1, \qquad   x=    \max I_{\tau_2'}- \tau_2
\eeq
and we estimate \eqref{eq: correlation to be estimated for eq} restricted to $\tau'_1 < \tau_2'-1$ by
\baq
 C \ell^{-2\al} \sum_{z =\ell^2}^{\infty}  h(\tau_1, \tau_1+ z)    \sum_{x =0}^{\ell^2-1}     (1+ x)^{-1/2}     (1+ (x+z)\mathrm{mod}\ \ell^2)^{-1/2}  
&\leq&  C \ell^{-2\al} \log \ell   
\eaq
We used the change of variables to get the factor $\ell^{-2\al} $ and the Cauchy-Schwarz inequality to estimate the $x$-sum. 
Restricting  \eqref{eq: correlation to be estimated for eq}  to $\tau'_1 = \tau_2'-1$ and setting $y := \tau_1 - \min I_{\tau_1'}$, we get 
\beq
 C   \sum_{x,y =0}^{\ell^2-1}       h(1, \ell^2-x-y)      (1+ x)^{-1/2}     (1+y)^{-1/2}  ( \ell^2-x-y-1)^{-2\al}
\eeq
This sum is  analogous to the one treated in Section \ref{sec: the case of three times}, the only difference being that one exponent is $1/2$ instead of $\al$. The multiscale treatment can be copied without changes. The result is that the sum is bounded by $C \ell^{-2\al}$ and this yields Proposition \ref{prop: linear rg eq}.

\section{Flow of $G^c_A$: nonlinear part} \label{sec: nonlinear rg flow}

In this section, no information on $T_n$ is needed, except for the bound on $T_n^m$ from Lemma \ref{lem: powers of t}.  Starting from the induction hypothesis on the cumulants at scale $n$, we deduce an estimate on the nonlinear part of the contribution to scale $n+1$.    We do not distinguish between bulk and boundary terms, and hence we have $A \subset \bbN_0$ throughout. We drop the scale subscript  $n$ and we mark operators on scale $n+1$ by a prime, as explained at the beginning of Section \ref{sec: linear rg flow}.

 As was explained in Section \ref{sec: linear rg flow}   as well, the nonlinear contribution to the RG flow is defined by excluding from the sum \eqref{eq: from n to npluseen}
%\baq \label{eq: repetition general recursion b intro nonlinear}
%G^{c'}_{A'}   = \bsS_{\ell} \sum_{ \caA \in \poly(I_{A'}), \caG_{A'}(\caA) \conn }  \, \, \,     \caT_{A'} \left[  \mathop{\otimes}\limits_{A \in \caA} G^{c}_{A}  \otimes  \mathop{\otimes}\limits_{\tau \notin \supp \caA}  T_{\tau}  \right] 
%\eaq
the contributions of $\caA=\{ A \}$ (a single set) with $\str A \str = \str A' \str$. Hence we still need to study the remaining
terms
\beq
G^{c'}_{A,\nlin} :=  \bsS_{\ell} 
\mathop{\sum}\limits_{\scriptsize{  \left. \begin{array}{c}\caA \in \poly(I_{A'}),  \caG_{A'}(\caA)\,  \textrm{connected} \\  \sum_{A \in \caA}\str A \str > \str A' \str   \end{array} \right. } }    \caT_{A'} \left[  \mathop{\otimes}\limits_{A \in \caA}  G^{c}_{A}  \otimes  \mathop{\otimes}\limits_{\tau \notin \supp \caA}  T({\tau})  \right] \label{eq: nonli}
\eeq
that is,  either $\caA$ has more than one element, or if it consists of one element $A$, then $\str A \str > \str A' \str$, and this is combined in the condition  $\sum_{A \in \caA}\str A \str > \str A' \str$. 
Our result is
\begin{proposition}\label{prop: nonlinear rg}
Fix  an exponent $\hat\al $  with $\hat \al < \al$ and a constant $c_v>0$. If $c_v$ is chosen small enough, then \beq
\sup_{ \tau'} \sum_{A' \subset \bbN_0: A' \ni \tau' } (\hat\ep')^{-\str A' \str}   \dist(A')^{\al} \norm  G^{c'}_{A',\nlin} \norm_{\tengam}    \leq  1, \qquad  \text{with}\,\, \hat\ep'=  c^{-1}_v \ell^{-\hat\al} \ep,
\eeq
\end{proposition}
We remind the  conventions introduced at the beginning of Section \ref{sec: induction hypotheses}; $\ell^{-1}$ and $\ep_0$ should be chosen sufficiently small  compared to  constants like $c_v$.
Proposition \ref{prop: nonlinear rg}  establishes (the nonlinear parts of) Proposition \ref{prop: overview b behavior}. Indeed, for $\al>1/2$, it suffices to fix $\tilde \al <\hat \al < \al$ whereas for the case $\al\leq 1/2$, we observe that we can choose $\hat \al$ such that additionally $  k \tilde\al + \tilde \al_{\ini} < k \tilde \hat \al $ holds for $k \geq 2$.

To derive Proposition \ref{prop: nonlinear rg}, one can ignore the Ward Identities completely since they can at best give a factor $\ell^2$ in each term of the sum, and such factors are irrelevant in view of the fact that we have extra factors of $\ep$ due to the condition $ \sum_{A \in \caA}\str A \str > \str A' \str $ (cfr.\ the discussion at the beginning of Section \ref{sec: linear rg flow}).  

Consequently,  we will use the following   basic  bound for the
cumulants. In fact, a special case of this bound appeared already in the proof of Lemma \ref{lem: bound on e}.

\begin{lemma}[A priori recursion relation] \label{lem: relation between scales cumulants}
\beq\label{A priori recursion}
\norm G^{c'}_{A'}   \norm_{\tengam}  \leq   \sum_{ \caA \in \poly(I_{A'}),\caG_{A'}(\caA) \,  \text{connected}}      \e^{C \str \supp \caA \str }   \prod_{A \in \caA}   \norm G^c_{A}  \norm_{\gamzero}   
\eeq

\end{lemma}

\begin{proof} 
We take $\ell > 10$ and  apply Lemma \ref{lem: main property of norms} to eq.\ (\ref{eq: from n to npluseen}):
\beq
 \left\norm \caT_{A'} \bsS_{\ell} \left[  \mathop{\otimes}\limits_{A \in \caA} G^{c}_{A}  \mathop{\otimes}\limits_{\tau \notin \supp \caA}  T({\tau})  \right] \right\norm_{10\ga_0}  \leq  \prod_{A \in \caA}  \norm G^{c}_{A} \norm_{\ga_0}    \prod_{ J  }    \norm T^{\str J \str} \norm_{\tengam/\ell}
\eeq
where the product  $\prod_{ J  }$ runs over all discrete intervals $J$ in the sets $I_{A'} \setminus \supp\caA$, cfr.\ the proof of Lemma \ref{lem: bound on e}.
  By invoking Lemma \ref{lem: powers of t}, we bound 
\beq
\norm T^{\str J \str} \norm_{\tengam/\ell}  \leq  C,     \qquad  \text{since}\,\,  \str J \str \leq \ell^2
\eeq
The number of discrete intervals $J$ is at most $2\str\supp \caA \str  $ and this yields the claim.
\end{proof}

%On the other hand, the entropy factors are more dangerous than for the linear RG, since the number of connected graph with vertex set   
%

\subsection{Sum over connected coverings}

Our strategy for bounding the sum over  polymers  
$\caA$ in eq.\ \eqref{eq: nonli} consists of the following splitting. For each $A \in \caA$, we let $S(A) \subset A'$ be the macroscopic domain of $A$, that is 
\beq
S(A) = \{  \tau' \in A',  A \cap I_{\tau'}  \neq \emptyset \}
\eeq
Note that any collection $\caA$ of sets $A$ induces a collection $\caS= \caS(\caA)$ with elements $S(A)$. We call a collection of sets \emph{connected} whenever it can not be split into two collections whose members are mutually disjoint. For any $\caA \in \poly(I_{A'})$, the connectedness of the graph $\caG_{A'}(\caA)$ implies that $\caS=\caS(\caA)$ is connected and that $\supp \caS=A'$. We call the set of connected collections $\frC$.   

With this terminology, \eqref{A priori recursion} can be written as
\baq
%&&  \sum_{ \caA \in \poly(I_{A'}), \caG_{A'}(\caA)\conn}    \prod_{A \in \caA}  \norm  G^{c}_{A}  \norm_{\gamzero}  \e^{C \str A \str}  \\[3mm]
\norm G^{c'}_{A'}   \norm_{\tengam} &\leq&    \sum_{ \caS \in \frC, \supp \caS= A'}    \sum_{ \caA \in \poly(I_{A'}):  S(\caA)=\caS}    \prod_{A \in \caA}    \norm  G^{c}_{A}  \norm_{\gamzero}  \e^{C \str A \str}    \\[3mm]
&\leq &   \sum_{ \caS \in \frC, \supp \caS= A'}   \prod_{S \in \caS}  F_1\left( \sum_{A \to S}   \norm G^{c}_{A}  \norm_{\gamzero} \e^{C \str A \str}  \right)%, \qquad \textrm{where} \, \,  F_1(x) =  x/(1-x)
\eaq
where the function $F_1(x): = \sum_{p=1}^{\infty} x^p=x/(1-x)$ appears because there can be more than one set $A \in \caA$ such that $A \to S$. 

Let us now determine how this bound can be modified if we restrict the sum in \eqref{A priori recursion} to those entering in the nonlinear RG, i.e. in (\ref{eq: nonli}) .   If $\str \caS\str >2$, then the condition $\sum_{A \in \caA}\str A \str > \str A' \str $ is always verified. If $\str\caS\str=1$, i.e. there is one $S=A'$, then there are either at least two $A \in \caA$ such that $A \to A'$, or there is only one but it satisfies $\str A \str > \str A'\str$.  Hence we get
\baq
 \norm G^{c'}_{A',\nlin}  \norm_{\tengam} &\leq&    \sum_{ \caS \in \frC, \supp \caS= A'   \str \caS \str >1}   \prod_{S \in \caS}  F_1\left( \sum_{A \to S}   \norm G^{c}_{A}  \norm_{\gamzero} \e^{C \str A \str}  \right)    \\[3mm]
&+&    \sum_{ A \to A',  \str A\str > \str A' \str }     \norm  G^{c}_{A}  \norm_{\gamzero}  \e^{C \str A \str} +    F_2\left( \sum_{A \to A'}   \norm G^{c}_{A}  \norm_{\gamzero} \e^{C \str A \str}  \right)%, \qquad \textrm{where} \, \,  F_2(x) = \sum_{p=2}^{\infty} x^p    
 \label{eq: three terms in nonlinear rg}
\eaq
 where $ F_2(x) = \sum_{p=2}^{\infty} x^p  $.  
 
 The following lemma shows how to control sums over $A \to S$ that will appear in the evaluation of the above expression

\begin{lemma}  \label{lem: all a that fit into s}

\ben
\item
\beq   \label{eq: easy sum}
\sum_{A \to S}    \e^{C \str A \str } \dist(A)^{\al}   \norm G^c_{A}  \norm_{\gamzero}    \leq (C \ep)^{\str S \str}
\eeq
\item
Abbreviate
\beq
  v_n(S)  := ( \ell^{ -\hat \al}\ep)^{-\str S \str}    \dist(S)^{\al}    \sum_{A \to S}   \norm G^{c}_{A}   \norm_{\gamzero}  \e^{C \str A \str}  \label{def: vns}.
  \eeq
  Then 
\beq
\sup_{\tau'} \sum_{S \ni \tau'}    v_n(S)
   \leq   C  \ell^{2+2 \al}.   \label{eq: nonlinear lemma linear part}
   \eeq
\item
\beq
\sup_{\tau'}  \sum_{S \ni \tau'}   ( \ell^{ -\hat \al}\ep)^{-\str S \str}  \dist(S)^{\al}    \sum_{A \to S, \str A \str > \str S \str}   \norm G^{c}_{A}   \norm_{\gamzero}  \e^{C \str A \str}
   \leq  C    \ell^{2+3 \al}  \ep.   \label{eq: nonlinear lemma nonlinear part}
\eeq
\een
\end{lemma}

\begin{proof}
To get \eqref{eq: easy sum}, we recall the function $h(\cdot)$ introduced in \eqref{eq: how h emerges}, and we bound the LHS by
\beq   \sum_{k \geq \str S \str}   (C \ep)^{k} \sum_{\tau_1 \in I_{\tau'_1}}   \sum_{A: \str A \str=k, \min A=\tau_1}   h(A)  \leq   \ell^2(C \ep)^{\str S \str}, \qquad \tau_1' = \min S  
\eeq
using the summability of $h$, i.e.\ \eqref{eq: normalization of h} and bounding the sum over $\tau_1=\min A$ by $\ell^2$.  

Now to \eqref{eq: nonlinear lemma linear part};
we rewrite
\beq
 \sum_{S \ni \tau'}    v_n(S)   \leq     \sum_{A:  A \cap I_{\tau'} \neq \emptyset}    \left(    \ell^{ -\hat \al}\ep  \right)^{-\str S(A) \str}      \ep^{  \str A \str }   \mathrm{dist}(S(A))^{\al} ( \mathrm{dist}(A))^{-\al}  h(A)  
 \eeq
Assume there is a $\tau' \in S(A), \tau' \neq \min S(A), \max S(A)$, such that the set $ A \cap I_{\tau'}$ contains only one element.  Then we get a factor $\ell^{-2\al}$ from change of variables, and if there are $p$  such $\tau'$, then we get at least $\lfloor (p+1)/2 \rfloor$ of these factors.  From elementary considerations, 
\beq
  2  (\lfloor (p+1)/2 \rfloor )  \geq  p \geq  \str S(A) \str -  (\str A \str - \str S(A) \str)  -2
\eeq
Hence, we bound
\baq
 \sum_{S \ni \tau'}  v_n(S)  & \leq &    \sum_{A:  A \cap I_{\tau'} \neq \emptyset}         \left(   \ell^{ -\hat \al}\ep  \right)^{-\str S(A) \str}      \ep^{  \str A \str  }  \ell^{- \al (  2\str S(A) \str - \str A \str -2 )  }  h(A)   \nonumber \\[2mm]
 & \leq &           \ell^{2\al}       \sum_{k_A \geq k_S \geq 2}    ( \ep \ell^{\al})^{k_A -k_S}   \ell^{(\hat \al-\al) k_S}           \sum_{A:  A \cap I_{\tau'} \neq \emptyset, \str A \str =k_A}    h(A)  \nonumber   \\[2mm]
& \leq &                \ell^{2\al+2}       \sum_{k_A \geq k_S \geq 2}    k_A   ( \ep \ell^{\al})^{k_A -k_S}     \ell^{(\hat \al-\al) k_S}     \leq      C   \ell^{2+2\al}      \label{eq: estimate of vns sum}
\eaq
The second inequality follows by setting $k_S=\str S(A) \str, k_A = \str A \str$. To get the third inequality, we bound  the sum over    $ h(A) $ with $A \ni \tau$ by $\str A \str$ times the sum over $h(A)$ with $\min A $ fixed, and then we proceed as in the proof of, \eqref{eq: easy sum}. This yields \eqref{eq: nonlinear lemma linear part}. To obtain (\ref{eq: nonlinear lemma nonlinear part}), one repeats the calculation with the only difference that the sum over $k_A , k_S $ in \eqref{eq: estimate of vns sum} is restricted to $k_A -1 \geq  k_S \geq 2$. 

 \end{proof}

\subsection{Proof of Proposition \ref{prop: nonlinear rg}}

We perform the sum 
\beq
 \sum_{A' \ni \tau'}  ({\hat \ep}')^{-\str A' \str}  \dist(A')^{\al}   \norm G^{c'}_{A',\nlin}  \norm_{\tengam}  \leq   \mathbf{I} + \mathbf{II}  +  \mathbf{III} 
\eeq
where the three terms on the RHS refer to the three terms in \eqref{eq: three terms in nonlinear rg}.    The second term can be bounded immediately with the help of \eqref{eq: nonlinear lemma nonlinear part} (with $S = A'$), yielding
\beq
 \mathbf{II}   \leq     C   \ell^{2+3\al}  \ep
\eeq
We estimate the term  $  \mathbf{III}$, with $x :=  \sum_{A \to A'}   \norm G^{c}_{A}  \norm_{\gamzero} \e^{C \str A \str}   \leq (C \ep)^{\str A' \str} $ by \eqref{eq: easy sum} in Lemma \ref{lem: all a that fit into s}. In particular,  we have $x < 1-c$, and hence 
$F_2(x) \leq   \ep^2 C  x $ with $F_2$ as in \eqref{eq: three terms in nonlinear rg}, since $\str A ' \str \geq 2$. Therefore,  
\beq
  \mathbf{III}   \leq   C \ep^2   \sum_{A' \ni \tau'}  ({\hat \ep}')^{-\str A' \str}  \dist(A')^{\al}    \sum_{A \to A'}   \norm G^{c}_{A}  \norm_{\gamzero} \e^{C \str A \str}  \leq   \ep^2 \ell^{2 +2 \al}
\eeq
where we used \eqref{eq: nonlinear lemma linear part} in Lemma \ref{lem: all a that fit into s} (recognizing the definition of $v_n(S)$ with $S=A'$).
Now we turn to
\baq
 \mathbf{I}  \leq \sum_{A' \ni \tau'}    ({\hat \ep}')^{-\str A' \str}  \dist(A')^{\al}  \sum_{ \caS \in \frC, \supp \caS=A', \str \caS \str >1}   \prod_{S \in \caS}   \sum_{A \to S}   \norm G^{c}_{A}  \norm_{\gamzero} \e^{C \str A \str}     \label{eq: nonlinear rg}
\eaq
where we dropped the function $F_1(\cdot)$ at the cost of increasing the constant $C$, using, as in the treatment of term   $\mathbf{III}$  above, that its argument is smaller than $1-c$. 
First, for any  $\caS \in \frC$ with $ \supp \caS= A'$,   Lemma \ref{lem: distance factors} implies 
\beq
\dist(A')^{\al}  \leq  \prod_{S \in \caS}   \dist(S)^{\al} 
\eeq
such that \eqref{eq: nonlinear rg} is dominated by
\beq  \label{eq: nonlin rg basic quantity}
   \sum_{ \caS \in \frC, \str \caS \str >1,  \supp \caS \ni \tau'}  ({\hat \ep}')^{- \str\supp \caS \str }   \prod_{S \in \caS}    \dist(S)^{\al}   \left(  \sum_{A \to S}   \norm G^{c}_{A}   \norm_{\gamzero} \e^{C \str A \str}   \right)
\eeq

By substituting \eqref{def: vns} for each $S \in \caS$ in \eqref{eq: nonlin rg basic quantity}, and setting $N(\caS)= \sum_{S \in \caS} \str S \str$, 
\baq
 \textrm{\eqref{eq: nonlinear rg}}& \leq &   \sum_{ \caS \in \frC,  \str \caS \str >1,  \supp \caS \ni \tau'}  \underbrace{({\hat \ep}')^{N(\caS)- \str\supp \caS \str }     (\ell)^{(2+2\al) \str \caS \str }  }_{r(\caS)}    \prod_{S \in \caS}  c^{\str S \str}_v  \ell^{-(2+2\al)} v_n(S)
\eaq

By the connectedness of $\caS$, we have $ N(\caS)  - \str \supp \caS \str  \geq \str \caS \str -1$, and combined with $\str \caS \str \geq 2$, this yields $r(\caS) \leq 1$. 
 Hence we are left with the task of estimating (the constraint $\str \caS \str >1$ can now be dropped)
\beq  \label{eq: final sum over ic}
 \sum_{ \caS \in \frC,    \supp \caS \ni \tau'}   \prod_{S \in \caS}   c^{\str S \str}_v  \ell^{-(2+2\al)} v_n(S).
\eeq
The tool to perform this sum is the bound  $\sum_{S \ni \tau'}  \ell^{-(2+2\al)} v_n(S) \leq C$ from Lemma \ref{lem: all a that fit into s}.  Indeed,  sums  over collections of  connected sets, as in \eqref{eq: final sum over ic}, can be handled conveniently with cluster expansions, with the bound from Lemma \ref{lem: all a that fit into s} playing the role of the Kotecky-Preiss criterion.  We state the relevant cluster expansion result in Proposition \ref{app: prop: cluster expansion} in Appendix \ref{app: combinatorics}. A basic corollary, eq.\  \eqref{eq: useful cluster expansion}, implies  that \eqref{eq: final sum over ic} is bounded by a constant that can be made smaller than $1$ by choosing $c_v$ small enough. 

Collecting the bounds on the three terms $\mathbf{I}, \mathbf{II}, \mathbf{III}$, Proposition \ref{prop: nonlinear rg} follows

%====================================================================================================================================================================================================================================================================================================================================================================================

%====================================================================================================================================================================================================================================================================================================================================================================================

%====================================================================================================================================================================================================================================================================================================================================================================================

\section{Estimates on the first scale: Excitations} \label{sec: estimates on the first scale excitations}

In this section, we prove bounds on the correlation functions $G^c_A$, i.e.\ the claims of induction hypothesis \ref{prop: overview b behavior}  for $n=0$.  We treat the bulk correlation functions ($0 \notin A$) in Section \ref{sec: bounds on bulk correlation functions} and the boundary correlation functions ($0 \in A$) in Section \ref{sec: bounds on boundary correlation functions}.
We  derive an explicit representation for the operators $G^c_A$ that is based on the Dyson expansion. This representation will also be useful to analyze $T$ for $n=0$ in Section \ref{sec: estimates on the first scale reduced evolution}.   Since we start so to say from scratch, the first part of our discussion is in finite volume $\La$ and we perform the thermodynamic limit in Section  \ref{sec: bounds on the dyson expansion}. 

\subsection{Dyson Expansion} \label{sec: expansions}

We recall the reduced  dynamics of the system at microscopic time $t$
\baq
Z_t  \rho_\sys & = &   \Tr_{\res} \left[    \e^{-\i t H }    (\rho_\sys \otimes \initialresfinite)   \e^{\i  t  H }    \right] 
\eaq
and the dynamics on the $n=0$ scale, $T_{n=0}$ (introduced in Section \ref{sec: general}), it is related to $Z_t$ through
$T_0= Z_{\la^{-2}\frt_0}$. %where $ 1/2 \leq \frt_0 \leq 1$.  
In the next sections, we will develop an expansion for $Z_t$ and we will relate this expansion to the correlation functions $G_A^c$.

\subsubsection{Derivation of the expansion}  \label{sec: operator product expansion}

We introduce the  time-evolved interaction Hamiltonian
\baq\label{intHam}
H_{I}(s) & := &  \e^{\i s  H_\res}   \sum_{q \in \La^*, i=1,2} \left(  \phi(q)  (W \otimes \e^{\i q X} \otimes  a^*_{i, q})  + h.c. \right)   \e^{-\i s H_\res}  
\eaq
and define the Liouvillians $L_I(s)= \adjoint(H_I(s))$,  $L_\sys = \adjoint(H_\sys)$ and 
$L_E = \adjoint(H_E)$. Let us also use the shorthand 
$U_s = \e^{-\i s L_\sys}$. Then the Duhamel formula 
$$
\e^{\i tL_E}\e^{-\i tL}=U_t+\int_0^t \d s\, U_{t-s}(-i\la L_I(s)) \e^{\i sL_E}\e^{-\i  sL}
$$
yields upon iteration the  Dyson series
for the reduced dynamics $Z_{t}$:
\baq \label{eq: first duhamel series}
Z_{t} \rho_\sys=     \sum_{m \geq 0} (-\la^2)^m \mathop{\int}\limits_{0< t_1 < \ldots < t_{2m} <t} \d t_1 \ldots \d t_{2m} \,   \Tr_{\res}\Big[U_{t-t_{2m}}  L_I(t_{2m})  
  \dots U_{t_2-t_1} L_I({t_1}) U_{ t_1} (\rho_\sys\otimes
 \initialresfinite ) \Big] 
\eaq
where the invariance $ \Tr_{\res}(\e^{-\i tL_E}A)= \Tr_{\res}A$ was used. For $m=0$, the RHS is understood as $U_t \rho_\sys$.

Next, we  write the integrand  using the formalism developed in Section \ref{sec: correlation functions}:
\baq
 &&  \Tr_{\res}\Big[U_{t-t_{2m}}  L_I(t_{2m})    \dots U_{t_2-t_1} L_I({t_1}) U_{ t_1} (\rho_\sys\otimes
 \initialresfinite ) \Big]   \\[2mm]
  &=&  \caT \  \bbE\left[  U_{t-t_{2m}}  \odot L_I(t_{2m}) \odot \ldots  \odot  U_{t_2- t_1}  \odot L_I({t_1}) \odot U_{ t_1}\right] \rho_\sys.
\eaq
We recall that the expectation $\bbE$ acts on an element of $\scrR^{\otimes^n}\otimes \scrR_E$
and the contraction $\caT:\scrR^{\otimes^n}\to\scrR$. As in Section \ref{sec: correlation functions}, it will be convenient 
to label the spaces $\scrR$ in $\scrR^{\otimes^n}$ by the times that occur in the 
corresponding operators. Given a  $2m$-tuple  of times $\{t_1,\dots,t_{2m}\}\equiv\ut \subset [0,t]$ 
let $t_0=0, t_{2m+1}=t$ and
denote the family of intervals $J_i \equiv [t_i,t_{i+1}], i=0, 2m$ by $\caJ(\ut)$. 
We will then index the space $\scrR$ where $L_I({t_i})$ lies by $\scrR_{t_i}$ and 
the space $\scrR$ where $U_{t_{i+1}-t_{i}}$ lies by $\scrR_{J_i}$. Thus
\beq 
\bbE\left[  U_{t-t_{2m}}  \odot L_I(t_{2m}) \odot \ldots  \odot  U_{t_2- t_1}  \odot L_I({t_1}) \odot U_{ t_1}\right] \in \mathop{\otimes}\limits_{i=1}^{2m} \scrR_{t_i}   \mathop{\otimes}\limits_{J \in \caJ(\ut)}  \scrR_J
\label{bigspace}
\eeq 
and as before the operator $\caT$  contracts the operators in the obvious chronological order, i.e.\ such that those in $\scrR_{[0,t_1]}$ are on the right, then those in $\scrR_{t_1}$, then those in $\scrR_{[t_1,t_2]}$, etc.

Let $\{u,v \}$ be a pair of (distinct) times with the convention that $u  < v$. Then we define 
\beq\label{KUV}
K_{u,v}   =   -\la^2  \bbE\left[ L_I({v}))  \odot  L_I({u}))   \right], \qquad    K_{u,v}  \in  \scrR \otimes \scrR
 \eeq
 and we view this operator as an element in $\scrR_v \otimes \scrR_u$.  We also  abbreviate $U_{s'-s}$ by $U_{J}$ with $J=[s,s']$ and we view  $U_{J}$ as an element in $\scrR_J$.
 Since $L_I$ is linear in the creation and annihilation operators, the Wick theorem  yields
\baq
\bbE\left[  U_{t-t_{2m}}  \odot L_I(t_{2m}) \odot \ldots  \odot  U_{t_2- t_1}  \odot L_I({t_1}) \odot U_{ t_1}\right] =   \sum_{\pi \in \textrm{Pairings}(\ut)}    \mathop{\otimes}\limits_{\{u,v \} \in \pi}  K_{u,v}  \mathop{\otimes}\limits_{J \in \caJ(\ut)} U_{J}     \label{eq: first condensed pair sum}
\eaq
where the sum on the RHS runs over pairings $\pi$, i.e.\ partitions of the times $t_1, \ldots, t_{2m}$ in $m$ pairs $(u,v)$ with the convention that $u < v$.

By plugging \eqref{eq: first condensed pair sum} into \eqref{eq: first duhamel series}, we obtain
\baq \label{eq: second duhamel series}
Z_{t} &=&     \sum_{m \geq 0}  \mathop{\int}\limits_{0< t_1 < \ldots < t_{2m} <t} d t_1 \ldots \d t_{2m} \, \,   \sum_{\pi \in \textrm{Pairings}(\ut)}   \caT\left[  \otimes_{\{u,v \} \in \pi}  K_{u,v} \otimes_{J \in \caJ(\ut)} U_{J}    \right] 
\eaq
As before, we equip products of $\scrR$, e.g.\ as in \eqref{bigspace}, with the norm $\norm \cdot \normw$.

\subsubsection{A formalism for the combinatorics} \label{sec: combinatorics}

The integral over ordered $\ut$, together with the sum over pairings, $\pi$,  on the set of times, is represented as a combined integral and sum over ordered pairs $ (u_i,v_i)$ with $u_i,v_i \in \bbR^+$ and $i=1, \ldots, m$,  such that 
\beq
u_i < v_i, \qquad    u_1< \ldots < u_m
\eeq
This is done as follows. For any pair $(r,s) \in \pi$, we let $u_i=t_r, v_i =t_s$ where the index $i=1,\ldots,m$ is chosen such that the $u_i$ are ordered $u_1 < u_2 \ldots < u_m$. 
We represent one pair $ (u_i,v_i)$ by the symbol $w_i$ and the $m$-tuple of them by $\uw$. 
We call $\Om_J$ the set of  $\uw$ such that  $u_i,v_i \in J$  (for arbitrary $m$), and we use the shorthand
\beq
\mathop{\int}\limits_{\Om_J} \d \underline{w}  : = \mathop{ \sum}\limits_{m \geq 0} \,   \,\mathop{ \int}\limits_{J^m} \d \underline{u} \mathop{ \int}\limits_{J^m} \d \underline{v}    \, \,  \indicator_{[u_i <v_i]} \indicator_{[u_1 < \ldots < u_m]}
\eeq
where the RHS is set to $1$ for $m=0$, corresponding to  $\uw = \emptyset$ in the LHS. 
In what follows, we will often consider the ordered times $\ut, \underline{u}, \uv$ to be implicitly defined by $\uw$. For example, we will   write $\caJ(\uw)$ instead of 
 $\caJ(\ut)$.  The Dyson expansion in terms of the sets of pairs $\uw$ reads
\baq  \label{eq: z as integral over pairs}
Z_t  
&=&   \mathop{\int}\limits_{\Omega_{[0, t]}}     \d \uw  \,    \caT\left[  \otimes_{w \in \uw} K_{w} \mathop{\otimes}\limits_{J \in \caJ(\uw)} U_{J}    \right]
\eaq

\subsubsection{Connected correlations and the Dyson series} \label{sec: treegraph}

To relate the previous sections to the setup in Section \ref{sec: general}, we need to discretize time and express the operators $G_A, G^{c}_{A}$ in terms of the Dyson series.  

 Recall that $\bbN$  is the set of macroscopic times. To each macroscopic time, we now associate a \emph{domain} of microscopic times,
\beq
\Dom  ( \tau) =  [  \la^{-2}\frt_0 (\tau-1),   \la^{-2}\frt_0   \tau ]
\eeq
  To a set $A \subset \bbN$ of macroscopic times, we then associate the  domain 
 \beq
 \Dom (A)= \bigcup_{\tau \in A} \Dom  ( \tau)
 \eeq
We take $t=\la^{-2}\frt_0 N$. Then, a set of pairs $\uw \in \Om_{[0,t]}$ determines a graph $\caG(\uw)$ on $\{1, \ldots,N \}$ by the following prescription: the vertices $\tau< \tau'$ are connected by an edge if and only if there is a pair $w=(u,v)$ in $\uw$ such that 
\beq
u \in   \Dom  ( \tau)  \qquad  \textrm{and} \qquad     v \in   \Dom  ( \tau')
\eeq 
(Note that there may be several such pairs).
We write $\supp(\caG(\uw))$ for the set of non-isolated vertices of $\caG(\uw)$, i.e.\ the vertices that have at least one connection to another vertex.   If $\uw \in  \Om_{\Dom A}$, than $\supp(\caG(\uw))$ is obviously a subset of $A$.  In that case we write $\caG_A(\uw)$ for the induced graph with vertex set $A$.
The graphs $\caG(\uw)$ in the Dyson expansion  with support $A$ give rise to the correlation function $G_A$ of Section \ref{sec: general}, and connected graphs  $\caG_A(\uw)$ give rise to the connected correlation functions. This goes as follows. 
  Recall the collection $\caJ(\uw)$ of intervals determined by the times in  $\uw$. Given
  a macroscopic time $\tau\in\bbN$ we
  define $\caJ_{\tau}(\uw)$ as the family of intervals $\{\Dom ( \tau)\cap J\ |\ J\in\caJ_{}(\uw)\}$
  and set
   \beq
\caJ_{A}(\uw)  :=\bigcup_{\tau \in A} \caJ_{\tau}(\uw)
 \eeq
 % We now define $\caJ_{\Dom A}(\uw)$ by noting that, for any $\uw \in \Om_{\Dom(A)}$,  the associated set of times  
 %$t_i, i=1, \ldots, 2 \str \uw \str-1$  defines a partition of  $\Dom (\tau)$ in intervals, for every $\tau \in A$.  We call this set of intervals  $\caJ_{\Dom ( \tau)}(\uw)$ and we set
 Using the group property of $U_s$ and the definition of contraction we have
 $$
 U_J=\iota(U_{J_2}\otimes U_{J_1})
 $$
 for $J={J_1}\cup {J_2}$ with $J_1,J_2$ consecutive intervals. Applying this to the intervals $J$
 in (\ref{eq: z as integral over pairs}) intersecting more than one $\Dom ( \tau)$ we get
 \baq  \label{blocking}
    \caT\left[   \mathop{\otimes}\limits_{w \in \uw} K_{w} \mathop{\otimes}\limits_{J \in \caJ_{}(\uw)} U_{J}    \right]
    = \caT\left[    \mathop{\otimes}\limits_{w \in \uw} K_{w} 
   \mathop{\otimes}\limits_{J \in \caJ_{[1,N]}(\uw)} U_{J}   \right]=: \caT\caV_{[0, N]}(\uw).
\eaq
where we abbreviated $[1,N]= \{ 1,2,\ldots, N\}$.
 The tensor product defining $\caV_{[1, N]}(\uw)$ factors across the macroscopic times, i.e:
  \baq  \label{blocking1}
\caV_{[1, N]}(\uw)=\mathop{\otimes}\limits_{\tau\notin\supp(\caG(\uw))}  \caV_{\Dom ( \tau)}(\uw)
 \mathop\otimes  \caV_{\supp(\caG(\uw))}(\uw)
\eaq
which can also be written as
  \baq  \label{blocking2}
\caV_{[1, N]}(\uw)=\mathop\otimes_{\tau\notin\supp(\caG(\uw))}  \caV_{\Dom ( \tau)}(\uw)
 \mathop\otimes_i  \caV_{\supp(\caG_i(\uw))}(\uw) 
\eaq
where $\caG_i$ are the connected components of $\caG$. 

As in Section \eqref{sec: recursion relations for correlation functions}, we may perform the time ordered contraction $\caT$ in two steps, first within
the time intervals  $\Dom ( \tau)$  and then contracting the rest:
\beq\label{eq: magic formula} \caT=\caT\mathop\otimes_{\tau=1}^N\caT_\tau
\eeq
where 
$$\caT_\tau: \mathop{\otimes}\limits_{i:t_i\in\Dom ( \tau)} \scrR_{t_i}   \mathop{\otimes}\limits_{J \in \caJ_{\tau}(\uw)}  \scrR_J \quad \to \quad \scrR_\tau.$$
The beautiful formula \eqref{eq: magic formula} is not a typo but a consequence of the fact that we defined $\caT$ both as contracting microscopic times and intervals, and macroscopic intervals. 
Writing as before $\caT_A=\mathop\otimes_{\tau\in A}\caT_\tau
$,
eq.\ (\ref{blocking1}) leads then to the expansion (\ref{eq: first expansion density matrix}) with
\baq
 && G_A    = \mathop{\int}\limits_{\scriptsize{\left.\begin{array}{c}   \Om_{ \Dom (A)}  \end{array}\right.  }}  \,   \d \uw   \,   \,  \indicator_{[ \supp(\caG(\uw)) =A  ]}    \caT_{A}  \left[ \mathop{\otimes}\limits_{w \in \uw} K_w \mathop{\otimes}\limits_{J \in \caJ_{A}(\uw) } U_{J}    \right]  \label{eq: correlation function as support of graphs}
\eaq
and
\beq
T= \mathop{\int}\limits_{\scriptsize{\left.\begin{array}{c}  \Om_{ \Dom(\tau)}   \end{array}\right.  }}  \,   \d \uw   \,   \,       \caT  \left[ \mathop{\otimes}\limits_{w \in \uw} K_w \mathop{\otimes}\limits_{J \in \caJ_{\tau}(\uw) } U_{J}    \right]   \label{eq: T in terms of caI}
\eeq
  (note that $\tau$ on the RHS is arbitrary) whereas 
 (\ref{blocking2}) 
gives
\baq
  G^{c}_{A}  
 &=&   \mathop{\int}\limits_{\scriptsize{\left.\begin{array}{c}  \Om_{  \Dom (A)}     \end{array}\right.  }}  \d \uw \, \,   \indicator_{[\caG_A(\uw) \,  \text{connected}]}  \,  \caT_{A}  \left[ \mathop{\otimes}\limits_{w \in \uw} K_w \mathop{\otimes}\limits_{J \in \caJ_{A}(\uw) } U_{J}    \right]       \label{eq: correlation function as connected graphs}
 \eaq

 \subsection{Bounds on bulk correlation functions} \label{sec: bounds on bulk correlation functions}
 
 In this Section, we state and prove Lemma \ref{prop: integrability}, which is actually 
 the induction hypothesis \ref{prop: overview b behavior} for bulk sets $A$ (not containing  $0$) and on scale $n=0$.  Boundary correlation functions ($A \ni 0$) will be treated in Section \ref{sec: bounds on boundary correlation functions}.

\subsubsection{Bounds on the Dyson expansion}   \label{sec: bounds on the dyson expansion}

As a first step, let us derive a term by term bound on the Dyson expansion and perform the thermodynamic limit.  
Let us write the operator $K_{u,v}$ in (\ref{KUV}) explicitly. Using the field operators
( \ref{Segal}) we can rewrite (\ref{intHam}) as
\beq
H_I(s)=\sum_{x}W\otimes \lone_x\otimes \Phi(x,s).
\eeq
Let us use the notation ${\adjoint}(A)=A^0-A^1$ where $A^0$ is left multiplication and  $A^1$
is right multiplication (by $A$). Then, (\ref{KUV}) 
becomes
\beq\label{kuv expression}
K_{u,v}=-\la^2\sum_{x,y}\sum_{a,b\in\{0,1\}}(-1)^{a+b}%\Tr_E\left(\Phi(v,x)^a\Phi(u,y)^b\initialresfinite\right)
\zeta^{ab}(x-y,v-u)(W\otimes \lone_x)^a\otimes(W\otimes \lone_y)^b
\eeq
where 
\beq\label{kuv expression1}
\zeta^{ab}(x-y,v-u)=\Tr_E\left(\Phi(x,v)^a\Phi(y,u)^b\initialresfinite\right)
\eeq
By (\ref{eq: microscopic expression correlation function}) 
\beq
\zeta^{00}(x,t)=\zeta^{10}(x,t)=\zeta^{11}(x,-t)=\zeta^{01}(x,-t)=\zeta(x,t)
%\Tr_E(\Phi(v,x)^a\Phi(u,y)^a\initialresfinite)=\zeta(x-y,v-u).
\eeq
(we used translation invariance in time and  $O(d)$ invariance  in $x$).

Eq.\ \eqref{eq: assumption free correlation functions thermo} (thermodynamic limit for $\zeta$) implies the kernel of $K_{u,v}$ has a pointwise limit as $\La\nearrow \bbZ^d$.  For the remainder of the present section, we use the notation $K_{u,v}$ and $h(s)$ (introduced below) both for  $\La$ finite and $\La=\bbZ^d$, indicating differences whenever necessary.  Since the kernel of $\lone_y^a$ is diagonal in the coordinates  $x,v$ (recall that $z=(x,v,\eta, e_{\links}, e_{\rechts}) \in \bbA_0 $ and note that this $v$ has nothing to do with the time-coordinates $u,v$ used below),
we get
\beq\label{KUVbound}
\|K_{u,v}\|_\ga\leq C\la^2
\eeq
for all $\ga$. In fact, using the time decay of $\zeta$ in assumption \ref{ass: decay micro alpha}
and denoting
\beq  \label{eq: def h pairings}
 \la^2 h(v-u ):=   \left\norm  K_{u,v} \right\norm_{\twentygam }, 
\eeq
we get  $ h(s) \leq C $  for $\La$ finite and 
\beq    \label{eq: working bound on h}
 \int_{0}^{\infty}  \d s    (1+\str s \str)^{\al}  h(s)   \leq   C, \qquad \text{for} \, \La = \bbZ^d  \\    
\eeq 

\begin{lemma}\label{lem: bound on h}
The sums and integrals on the RHS of  \eqref{eq: z as integral over pairs}, \eqref{eq: correlation function as support of graphs}  \eqref{eq: T in terms of caI} and \eqref{eq: correlation function as connected graphs}  converge absolutely.   For example, the  series defining $Z_t$  is bounded  by
\baq
&&  \int_{\Omega_{[0, t]}}   \,  \d \uw   \,    \left\norm   \caT\left[  \mathop{\otimes}\limits_{w \in \uw} K_{w} \mathop{\otimes}\limits_{J \in \caJ_{}(\uw)} U_J  \right] \right\norm_{20 \ga_0}  \leq     \left\{ \begin{array}{ll} \e^{C\la^2 t^2} &  \La \,\, \text{finite} \\[2mm] \e^{C\la^2 t} &  \La=\bbZ^d   \end{array} \right. 
  \label{eq: archetypical bound pairings}
\eaq
In particular, the limits of \eqref{eq: z as integral over pairs}, \eqref{eq: correlation function as support of graphs}  \eqref{eq: T in terms of caI} and \eqref{eq: correlation function as connected graphs} as 
$\La \nearrow \bbZ^d$ exist. 
Moreover $Z_t$ is strongly continuous in $t$.
\end{lemma}

\begin{proof}
The following reasoning applies for finite and infinite $\La$ alike.% (for finite $\La$, the function $h= h^{\La}$ is defined as in \eqref{eq: def h pairings} but with $K $ replaced by $K^{\La}$). 
Using Lemma 4.3. we get 
\beq
     \left\norm   \caT\left[  \mathop{\otimes}\limits_{w \in \uw} K_{w} \mathop{\otimes}\limits_{J \in \caJ_{}(\uw)} U_J  \right] \right\norm_{\ga} 
 \leq    \prod_{w \in \uw} \left\norm  K_w \right\normw    \times \prod_{J \in \caJ_{}(\uw)}  \norm U_J   \normw    \label{eq: bound on pairs inside gluing}
\eeq
and using standard propagation bounds for the lattice Laplacian $\Delta$, we have 
\beq\label{propagation bounds}
 \norm U_J  \norm_{\twentygam}      \leq   C \e^{  \la^2 C \str J \str}
\eeq
(recall that we treat $\ga_0$ as a constant).
Hence the LHS of \eqref{eq: archetypical bound pairings} can be bounded by
\baq
 &&  \mathop{\int}\limits_{\Omega_{}}   \,  \d \uw    \prod_{J \in \caJ_{[0,t]}(\uw)}  C   \e^{ \la^2 C \str J \str }      \times  \prod_{w \in \uw} \la^2  h(v-u)   \\
 & \leq &     \e^{\la^2 C t}  \sum_{m \in \bbN}  \mathop{\int}\limits_{0< u_1 < \ldots < u_{m} <t} d \underline{u} \,   \left( \prod_{i=1}^m  C\int_{u_i}^t \d v_i h(v_i-u_i) \right)  \\
 &  \leq  &   \e^{C \la^2 t}  \sum_{m\geq 0}   \frac{(\la^{2}C t)^m}{m!}\left( \int_0^t \d s h(s)\right)^m 
 \leq         \left\{ \begin{array}{ll} \e^{C\la^2 t^2} &  \La \,\, \text{finite} \\[2mm]   \e^{C \la^2 t} &  \La=\bbZ^d   \end{array} \right.
  \label{eq: bound pairings by exponential}  \eaq
To get the second last inequality, we first performed the $v_i$-integrals, and then we estimated the $u_i$-integrals by the volume of an $m$-dimensional simplex. In the last one we used
eq.\  \eqref{eq: working bound on h} and $h(s) >C$ for finite $\La$. 
For any finite $t$, the  term-by-term convergence of the series follows by the Vitali convergence theorem, since each term converges pointwise and the series is summable uniformly in $\La$.   For the series in  (\ref{eq: correlation function as support of graphs}, \ref{eq: T in terms of caI}, \ref{eq: correlation function as connected graphs}), similar reasoning applies.   Strong continuity of $t \mapsto Z_t$ follows from  
strong continuity of $s \mapsto U_s$.
\end{proof}

The following estimate is an immediate consequence of (\ref{eq: correlation function as connected graphs}) and  bounds as in the proof of Lemma \ref{lem: bound on h}.   We write  $\uw  \mathop{\longrightarrow}\limits_{\text{min}} A$ to denote that the graph $\caG_A(\uw)$ is connected and that no pair can be dropped from $\uw$ without losing this property.
 In particular, this implies that  
  $\caG_A(\uw)$ is a spanning tree on $A$.   We say that such $\uw$ 'minimally span' $A$.

 \begin{lemma} \label{lem: a priori}
 \beq  \label{eq: minimal spanning bound}
  \norm G_A^c  \norm_{\tengam,\weird}   \leq     \e^{   C \str A  \str  } 
   \mathop{\int}\limits_{\scriptsize{\left.\begin{array}{c}  \Om_{ \Dom  (A)}    \end{array}\right.  }}
    \d \uw\,        1_{[\uw  \mathop{\longrightarrow}\limits_{\text{min}} A ] }   \prod_{w \in \uw}  \la^2 C h(v-u) 
 \eeq
   \end{lemma}

\begin{proof}
By  \eqref{eq: correlation function as connected graphs} and Lemma \ref{lem: bound on h}, we have
 \beq
   \norm G_A^c  \norm_{\tengam}   \leq   \mathop{\int}\limits_{\scriptsize{\left.\begin{array}{c}  \Si_{ \Dom  (A)}    \end{array}\right.  }}  \d \uw  \,     \indicator_{[\caG_A(\uw) \,  \text{connected}]}   \mathop{\prod}\limits_{J \in \caJ_{A}(\uw) }C   \e^{ \la^2 C \str J \str }      \times  \prod_{w \in \uw} \la^2  h(v-u)   \label{eq: proof of maximally 1}
 \eeq
and, since $\str \caJ_{A}(\uw)  \str \leq \str \uw \str +\str A \str$, we may dominate the integrand by  $ \e^{ C  \str A \str } \prod_{w \in \uw} \la^2 C h(v-u)  $, since  $\str \Dom(A)\str = \la^{-2} \str A \str $.
Next, we state  an appealing estimate was the main motivation for encoding the pairings $\pi$ in the pair-sets $\uw$.
For any (integrable) function $f $ on $\Omega_{\Dom (A)}$, we have 
\beq
 \mathop{\int}\limits_{\scriptsize{\left.\begin{array}{c}   \Si_{ \Dom  (A)}    \end{array}\right.  }}   \,  \d \uw    \,     \indicator_{[\caG_A(\uw) \,  \text{connected}]}  \str f(\uw) \str  \quad   \leq \quad   \mathop{\int}\limits_{\scriptsize{\left.\begin{array}{c}   \Om_{ \Dom  (A)}      \end{array}\right.  }}   \,  \d \uw'    \,    1_{[\uw  \mathop{\longrightarrow}\limits_{\text{min}} A ] }   \quad     \mathop{\int}\limits_{  \Si_{ \Dom  (A)} }   \,  \d \uw''    \,      \str f(\uw' \cup \uw'') \str
\eeq
  To realize why this holds true,  choose a spanning tree $\scrT$ for the connected graph $\caG_A(\uw)$ and then  pick a minimal subset $\uw'$ of the pairs in $  \uw$ such that $\caG_A(\uw')=\scrT$.    
  Since, in general, this can be done in a nonunique way, the integrals on the RHS contain the same $\uw$ more than once, and the inequality is strict unless $f$ is concentrated on minimally spanning $\uw$.  
 
 We apply this inequality   with $f(\uw):=  \prod_{w \in \uw} \la^2 C h(v-u)$ to \eqref{eq: proof of maximally 1} (with the integrand dominated as indicated there). The resulting integral over $\d \uw''$ can then be performed similarly to \eqref{eq: bound pairings by exponential}, yielding $e^{C \str A \str}$.    This proves the claim. 
\end{proof}

 \vskip 2mm
 
 We now derive the main result of the present section
  \begin{proposition}\label{prop: integrability}
For $\la$ small enough and with $\ep_0=  C\str\la\str^{{\alpha}} $, 
    \beq  
 \mathop{\sum}\limits_{\scriptsize{\left.\begin{array}{c}  A \subset \bbN: \min A =1   \end{array}\right.  }} \, \ep_0^{-\str A \str} 
 \dist(A)^{\al}   \norm G^c_A  \norm_{\tengam}   \leq   1    \label{eq: bound to be proven}
 \eeq
  \end{proposition}
 \begin{proof} 
Any $\uw$ that spans $A$ minimally  determines a spanning  tree on $A$. Hence we can reorganize the bound \eqref{eq: minimal spanning bound} by first integrating all $\underline{w}$ that determine the same tree. This amounts to integrate, for each edge of the tree, all pairs $(u,v)$ that determine this edge. Furthermore, we use that
\beq
\dist(A)^{\al} \leq \prod_{ \{ \tau,\tau' \} \in \caE(\scrT) } (1+\str \tau'-\tau \str)^{\al},  \qquad  \textrm{for   any spanning tree   $\scrT$ on $A$}   \label{eq: inequality spanning trees}
\eeq
 where $\caE(\scrT)$ is the set of edges of the spanning tree $\scrT$ (see Appendix \ref{app: combinatorics} for a simple proof). 
We arrive at the bound
 \beq  \label{eq: spanning trees bound}
   \dist(A)^{\al}  \norm G_A^c  \norm_{\tengam}   \leq     \e^{C   \str  A \str}     \,   \sum_{\textrm{span. trees} \, \scrT  \, {on} \, A} 
 \prod_{  (\tau,\tau') \in \caE(  \scrT) }      \hat e_\al(\tau,\tau')   
 \eeq
where, for $ \tau < \tau'$, 
 \beq  \label{eq: definition edge factor}
 \hat e_\al (\tau,\tau') :=   \la^2 (1+\str \tau-\tau' \str)^{\al}  \int_{  \Dom(\tau)} \d u   \int_{  \Dom(\tau')}  \d v \,    h(v-u),
  \eeq
  and  $  \hat{e}_\al(\tau',\tau) :=   \hat{e}_\al(\tau,\tau')$. 
Next, we establish
 \beq    \label{eq: kotecky preiss at scale zero}
 \sum_{\tau' \in \bbN \setminus \{ \tau \}}   \hat e_\al (\tau,\tau')  \leq  C\str\la\str^{2{\alpha}}, \eeq
by using the bound \eqref{eq: working bound on h} on $h$ from Lemma \ref{lem: bound on h};
 \baq
    \hat{e}_\al(\tau,\tau+1)    & \leq  &  \la^2  C \int_0^{2 \la^{-2} \frt_0} \d s   \, s h(s)    \leq  \la^2 C   (2 \la^{-2} \frt_0)^{1-{\al}} \leq C\str\la\str^{2{\al}}  \nonumber \\[2mm]
  \sum_{\tau' > \tau +1}   \hat{e}_\al(\tau,\tau')   & \leq  &   \la^2  C ( \frt_0 \str\la\str^{-2})^{-\al}   \mathop{\int}\limits_{  \la^{-2}(\tau-1)\frt_0 }^{  \la^{-2}\tau\frt_0 }  \d u  \mathop{\int}\limits_{ u+ \la^{-2}  \frt_0}^{\infty  } \d v (v-u)^{\al} h(v-u)    \leq     C \str\la\str^{2\al}  \nonumber
 \eaq
Starting from the inequality \eqref{eq: spanning trees bound}, we bound the sum in \eqref{eq: bound to be proven} as 
\beq
\mathop{\sum}\limits_{\scriptsize{\left.\begin{array}{c}  A \subset \bbN  \\    \min A =1   \end{array}\right.  }} \, \ep_0^{-\str A \str} 
 \dist(A)^{\al}   \norm G^c_A  \norm_{\tengam}   \leq   
 \mathop{\sum}\limits_{\scriptsize{\left.\begin{array}{c}  \textrm{trees} \, \scrT \,  \textrm{on} \, \bbN:   \\ 
  \supp(\scrT) \ni 1, \str \scrT \str \geq 2  \end{array}\right.  }}
   (C\str\la\str^{2{\al}})^{-\str \caE(\scrT) \str}
 \prod_{  (\tau,\tau') \in \caE(  \scrT) }       \hat e_\al(\tau,\tau')   \label{eq: pairings bound by trees}
\eeq
where we denote by $\supp(\scrT)$ the vertices that have at least one edge and $\caE(\scrT)$ is the set of edges. 
 This sum can be controlled relying on the bound \eqref{eq: kotecky preiss at scale zero}. This is a special case of Lemma \ref{app: prop: cluster expansion} in Appendix  \ref{app: combinatorics} where \eqref{eq: kotecky preiss at scale zero} plays the role of the Kotecky-Preiss criterion \eqref{eq: kotecky preiss}, $\ka=1$, and the edges $\{\tau,\tau' \}$ of the tree $\scrT$ play the role of the sets $S$.
 
  \end{proof}

\subsection{Bounds on boundary correlation functions} \label{sec: bounds on boundary correlation functions}
In this section, we work in the equilibrium case $\be_1=\be_2=\be$, since only in that case the boundary correlation functions are relevant. 
We recall from Section \ref{sec: correlations involving zero}:
\baq
\breve Z_t  \rho_\sys & = &   \Tr_{\res} \left[    \e^{-\i   t H }  D_{\mathrm{rd}}  (\rho_\sys \otimes \initialresfinite) D_{\mathrm{rd}}^\ast  \e^{\i  t H }  \right] 
\eaq
which differs from the reduced evolution $Z_t$ by the fact that we included the Radon-Nikodym derivative $D_{\mathrm{rd}}$. 

For a dense set of $\rho \in \scrB_1(\scrH)$, we can write a formal Duhamel expansion   
\beq
 D_{\mathrm{rd}} \rho D_{\mathrm{rd}}^\ast=  \e^{\Delta F(\beta)} \sum_{m=0}^{\infty}  \mathop{\int}\limits_{0< \beta_1 < \ldots < \beta_m < {\beta}} \d \beta_1 \ldots \d \beta_m   \,  \la^{m} \,  \tilde L_I ( \beta_m)  \ldots    \tilde L_I  (\beta_1) \rho \label{duha1}
\eeq
where the term with $m=0$ is defined to be $\lone$ and
\beq
    \tilde L_I ( \beta_i) \rho=- \frac{1}{2}[e^{\frac{\beta_i}{2}(H_\sys+H_\res) }H_I e^{-\frac{\beta_i}{2}(H_\sys+H_\res) },\rho]_+
\eeq
with $[B,A]_+=BA+AB$.
We will now combine this expansion with the Dyson expansion for the unitary evolution
to obtain an expansion for   $\breve Z_t$. For this it is convenient to set in eq. (\ref{duha1}) $\beta_i=t_i-\beta$
with $t_i\in[-\beta,0]$ and 
$$L_I(t) := \tilde L_I  (t+\beta)\ \ \textrm{for}\ \  t\in[-\beta,0].$$  We also
generalize  
$$
U_J :=U_{J\cap[0,t]}\ \ \textrm{for}\ \  J\subset[-\beta,t].$$ 
and 
\beq
K_{u,v} :=    ( \i \indicator_{[u \geq 0]}+\indicator_{[u <0]})   ( \i \indicator_{[v \geq 0]}+\indicator_{[v <0]})   \bbE[L_I(v) \odot  L_I(u)]  
\eeq
This bizarre looking formula simply takes care of the fact that the Dyson expansion for the Radon Nikodym derivative does not have factors of $\i$.   In particular, the above definitions of $U_J, K_{u,v}$ reduces to the ones given previously in Section \ref{sec: operator product expansion} in the case $J \subset \bbR_+, 0<u<v$. 
For a set of pairs $\uw \in {\Omega_{[-\be, t]}}$, we let now  $\caJ(\uw)$ be the  induced collection of intervals partitioning the interval $[-\be, t]$ instead of the interval $[0,t]$.

It is now straightforward to check the formal expansion
\baq
\breve Z_t  
&=&  \e^{\Delta F(\beta)}  \int_{\Omega_{[-\be, t]}}      \d \uw   \,  \caT\left[  \mathop{\otimes}\limits_{w \in \uw} K_{w} \mathop{\otimes}\limits_{J \in \caJ_{[-\be, t]}(\uw)} U_{J}    \right]   \label{eq: expansion for breve Z}
\eaq
The factor $ \e^{\Delta F(\beta)} $ can be determined from the relation
\beq
1 =  \Tr[\breve Z_0 \rho^{\reff}_\sys]  =    \Tr[\breve Z_0 (\indicator_{[x=0]} \nu^{\reff}) ]   \label{eq: rewriting free energy}
\eeq
where   $ \rho^{\reff}_\sys$ and $\nu^{\reff}$ were defined in Section \ref{sec: equilibrium states} (note that $\indicator_{[x=0]} \nu^{\reff}$ is a function on $\bbA_0$) and the last equality exploits the fact that $ \rho^{\reff}_\sys$ and $\nu^{\reff}$ do not depend on $x \in \bbX_0$ and also the operator $\breve Z_0$ is translation-invariant.  The advantage of the rightmost term in \eqref{eq: rewriting free energy} is that  $\breve Z_0$ acts on an operator in $\scrB_1(\scrH_\sys)$ that is strictly localized on the lattice and in particular it has a limit as $\La \nearrow \bbZ^d$.
  Let us first define
\beq
\breve h(s)  =  \sup_{-\be \leq u\leq v:  v-u =s}  \norm K_{u,v} \norm_{\twentygam}   
\eeq 
Then, Assumption \ref{ass: decay micro alpha} still implies that $\int_{\bbR_+} \d s(1+s^{\al}) \breve h(s)  < C$. To see this, let us inspect how the operator $K_{u.v}$  written out explicitly in \eqref{kuv expression}, gets modified for negative $u,v$. 
\ben
\item Depending on $a,b \in \{ 0,1\}$, there are minus signs that do not affect our bounds.
\item The correlation function $\zeta(x,t)$ is evaluated  with $t$ in the complex strip $ \bbH_\be$ instead of only on the real axis, but, by Assumption \ref{ass: decay micro alpha}, this does not spoil the decay property. This is in fact the only reason why we needed the complex strip in Assumption \ref{ass: decay micro alpha}.
\item The operators  $\lone_x^{a}$ in \eqref{kuv expression} should be replaced by $\e^{r H_\sys } \lone_x^{a}\e^{-r H_\sys },  0< r < \be/2$. By a simple propagation estimate in imaginary time, cfr.\ \eqref{propagation bounds}, the  $\norm \cdot \norm_{\twentygam}$-norm of such expressions is uniformly bounded by a constant.
\een

By the above remarks, it now follows that Lemma \ref{lem: bound on h} still holds if we replace $h \to \breve h$. In particular, we can 
 bound the expansion in \eqref{eq: expansion for breve Z} as in Lemma \ref{lem: bound on h}, we control the expansion of the rightmost term in \eqref{eq: rewriting free energy} and hence we obtain the thermodynamic limit of the number $ \e^{\Delta F(\beta)}$ and the equilibrium state $\nu^{\be}$,  and we prove that the operator $T_1(0)$ defined in Section \ref{sec: correlations involving zero} is bounded and has a thermodynamic limit, too.  Finally, from the expansion we get as well that
 \beq
\sup_{s \in \caS}  \left(\e^{\twentygam \str v \str }\str \mu^{\be}(s)-  \mu^{\reff}(s)\str\right) \leq C \la^2 \label{eq: closeness free and interacting equilibrium state}
 \eeq
 with $\mu^{\be}(s), \mu^{\reff}(s)$ defined in Section \ref{sec: equilibrium states} and we recall that $s=(v,\eta, e_{\links}, e_{\rechts})$.

Let us next generalize the expression for the correlation functions $G^c_A$. For the macroscopic time $\tau=0$, we define 
$\Dom(\tau):= [-\be, 0]$.  Then the equalities \eqref{eq: correlation function as support of graphs} and \eqref{eq: correlation function as connected graphs}  remain true without any changes and Lemma \ref{lem: a priori} remains true upon replacing $h \to \breve h$. 

Finally, we prove the bound \eqref{eq: induction assumption initial} (induction hypothesis \ref{prop: overview b behavior}) at scale $n=0$. 
 \begin{proposition}\label{prop: integrability with boundary}
Recall that $\ep_0=  C\str\la\str^{{\alpha}} $ and $\ep_{\ini,0}= C\str\la\str^{2-2{\alpha}} $. 
For $\la$  small enough; 
    \beq  
 \mathop{\sum}\limits_{\scriptsize{\left.\begin{array}{c}  A \subset \bbN  \\    \min A =0   \end{array}\right.  }} \,  \ep_0^{-\str A \str} 
 \dist(A)^{\al}   \norm G^c_A  \norm_{\twentygam}   \leq   \ep_{\ini,0}    \label{eq: bound to be proven boundary}
 \eeq
  \end{proposition}
 \begin{proof}
 The proof mimics that of Proposition \ref{prop: integrability}.  We first derive
 \beq  \label{eq: spanning trees bound boundary}
   \dist(A)^{\al}  \norm G_A^c  \norm_{\weird}   \leq     \e^{  C   \str A \str}     \,   \sum_{\textrm{span. trees} \, \scrT  \, {on} \, A}  
 \prod_{  (\tau,\tau') \in \caE(  \scrT) }      \hat e_\al(\tau,\tau')   
 \eeq
where $ \hat{e}_{\al}(\tau,\tau')$ was defined in \eqref{eq: definition edge factor} and this definition carries over to the case where one of $\tau,\tau'$ equals $0$, provided that we replace again $h$ by $\breve h$.  However,  the bound  \eqref{eq: kotecky preiss at scale zero}  can now be improved as 
 \beq    \label{eq: kotecky preiss at scale zero boundary}
 \sum_{\tau' \in \bbN \setminus \{ \tau \}}    \hat e_\al(\tau,\tau')    \leq \left\{ \begin{array}{ll}   C\str\la\str^{2 {\al}} & \tau \neq 0 \\ [2mm] 
  C \la^2 & \tau = 0
\end{array} \right.  \qquad  
 \eeq
 the improvement for $\tau=0$ is due to the fact that the length of $\Dom(\tau)$ is $\be <C$ for $\tau=0$ instead of $\frt_0 \la^{-2}$ for $\tau \neq 0$.   In the sum over trees \eqref{eq: pairings bound by trees}, there is now always one edge of the form $\{0,\tau \}$ for which we have the bound $C\la^2 = C\ep_0^2 \ep_{\ini,0}$  instead of $C\ep_0^2$, and this accounts trivially for the improved bound on the RHS of \eqref{eq: bound to be proven boundary}. 
 \end{proof}

\section{Estimates on the first scale: reduced evolution $T$} \label{sec: estimates on the first scale reduced evolution}

The analysis of the operator $T=T_{n=0}$ proceeds in three steps. Note that $T$ depends on the coupling constant $\la$ both through  the strength of the interaction and through the definition of the time scale $\la^{-2}\frt_0$, since  $T = Z_{\la^{-2}\frt_0}$. In a first step, accomplished in Section \ref{sec: weak coupling limit}, we prove that, as $\la \to 0$, the operator $Z_{\la^{-2}\frt} $ can be replaced by the Markov semigroup $\e^{\frt  (-\i \la^{-2}L_\sys + M) }$ with $M$ a dissipative operator.  This is the 'Markov approximation' that we referred to already in Section \ref{sec: markov approximation}. 
This Markov approximation is well-known in the literature since the work of Davies \cite{davies1} and we describe it mainly for reasons of completeness. 
 
In the second step, in Section \ref{sec: analysis of the semigroup}, we argue that the Markov semigroup $t \mapsto \e^{t (-\i L_\sys + \la^2 M) }$ has good properties, corresponding more or less to the requirements that Prop   \ref{prop: overview t behavior}  imposes on $T$. The reasoning in that section is standard in the analysis of (classical) linear Boltzmann equations, involving tools as Weyl's theorem and the Perron-Frobenius theorem. 
Finally, in Section  \ref{sec: t at n equal 0}, we argue that $T$ inherits these properties from 
the Markov semigroup.   This follows immediately by perturbation theory.

\subsection{The weak-coupling limit}  \label{sec: weak coupling limit}

\subsubsection{The generator $M$}  \label{sec: generator m}

We first define the generator $M$ and we investigate its properties.  
Let  $U^{0}_{t}$  stand for the propagator $U_{t}$ with $\la=0$, that is;
\beq
U^{0}_{t}  =   \e^{-\i t  L_\spin}, \qquad    L_\spin=  \adjoint(H_\spin)
\eeq
and we will again write $U^{0}_{J} $ to denote $U^{0}_{s'-s}$ for an interval $J= [s,s']$. 
 We define 
\baq
\la^2 G^0_{v-u}  &:=  &   \,   \caT\left[ K_{u,v} \otimes  U^{0}_{[u,v]} \right] \label{Godefi}\\
\la^2  G_{v-u}   &:=  &      \caT\left[ K_{u,v} \otimes  U_{[u,v]} \right] \label{Gdefi}
\eaq
Since $U^{0}_{t} $ is a finite dimensional unitary tensored with the identity, acting only on the spin degrees of freedom, 
$\|U^{0}_{t} \|_\ga\leq C$ uniformly in $t,\ga$, and
(\ref{eq: working bound on h}) implies
\beq    \label{eq: working bound on G}
 \int_{0}^{\infty}  \d s    (1+\str s \str)^{\al}  \|G^0_s\|_{20\gamma_0}   \leq   C.%, \qquad  \textrm{for any }\, \ga, \ga_0.
\eeq
Next, let $P_{\varepsilon}$ be the projector to the the eigenspace corresponding to
the eigenvalue $\ve\in\caE=\sigma(L_{\spin})$, cfr.\ eq. (\ref{bohr frequencies}). 
Then
the connection between the operators $G^0_s$ and the generator $M$ is via
\baq   \label{eq: first definition of m}
M  =  % \sum_{\varepsilon \in \caE}   \int_0^{\infty}  \d s \,    P_{\varepsilon}  \e^{-\i s  L_\spin}G^0_sP_{\varepsilon} = 
 \sum_{\varepsilon \in \caE}  
 \int_0^{\infty}  \d s \,    \e^{-\i s \ve} P_{\varepsilon} 
  G^0_sP_{\varepsilon} 
  % =  \sum_{\varepsilon \in \caE} 
   %  \int_0^{\infty}  \d s   \,     G^0_s 
\eaq
Thus we also have
\baq   \label{eq: bound of m}
\|M \|_{20\ga_0}\leq C. 
\eaq

\subsubsection{Emergence of $M$ from the Dyson expansion}

We recall the expansion (\ref{eq: z as integral over pairs}) for the reduced evolution $Z_t$:
\baq  \label{eq: z as integral over pairs repeated}
Z_t  
&=&   \int_{\Omega_{[0, t]}}      \d \uw  \,    \caT\left[  \mathop{\otimes}\limits_{w \in \uw} K_{w} \mathop{\otimes}\limits_{J \in \caJ_{}(\uw)} U_{J}    \right]
\eaq
We will show that the leading contribution to $Z_t$ in the above integral comes from
of pairs
\beq
\Omega^0_{[0, t]}   := \{  \uw \in \Omega_{[0, t]} \,  :  v_i < u_{i+1} \ \ i=1,\dots , \str \uw \str-1 \}.
\eeq
Set $\Omega^1_{[0, t]}  =\Omega_{[0, t]} \setminus \Omega^0_{[0, t]}  $ and fix
   $\ga= 20 \ga_0$  in Lemmata \ref{lem: bound on nonlead} and \ref{lemma: inverse laplace of lead}. First we have
\begin{lemma}  \label{lem: bound on nonlead}
\beq
 \int_{\Omega^1_{[0, t]}}   \,  \d \uw   \,   \left \norm     \caT\left[  \mathop{\otimes}\limits_{w \in \uw} K_{w} \mathop{\otimes}\limits_{J \in \caJ_{}(\uw)} U_{J}    \right] \right \normw   \leq   C  \e^{C\la^2t }  \str\la\str^{2\al}     \label{eq: statement of bound on nonlead}
\eeq

\end{lemma}
\begin{proof}
We proceed as in the proof of Lemma  \ref{lem: a priori}, and we bound the LHS of \eqref{eq: statement of bound on nonlead} by
\beq
 \e^{\la^2 C t} \int_{\Omega^1_{[0, t]}}   \,  \d \uw   \prod_{w \in \uw}  \la^2 h(v-u)   \label{eq: first bound on nonlead}
\eeq
Every $\uw \in \Omega^1_{[0, t]}$ has to contain  at least two pairs $(u,v), (u',v')$ such that $u <u'$ but $u' < v$. Choose the first two such pairs and integrate over the coordinates of all other pairs, using that $\norm h \norm_1 = \int \d s\,  h(s) <\infty$ and proceeding as in the proof of Lemma \ref{lem: bound on h}.  This yields the bound
\beq
 (\ref{eq: first bound on nonlead})\leq \e^{ \la^2 t (C + \norm h \norm_1 )}  q(t)  
\eeq 
with
$$
q(t) =   \mathop{\int}\limits_{0 <u <u'<v<t ,  u' <v' <t  }          h(v-u) h(v'-u')  
$$
and $q(t)$ can be easily estimated by first performing the integral over $v'$, which gives a factor $\norm h \norm_1$, and then the one over $u'$, such that one gets (recall that $\alpha\leq 1$)
\baq
q(t)   &\leq&    \norm h \norm_1 \la^4   \mathop{\int}\limits_{0 <u <v<t}  \d u  \d v       h(v-u)  \str v-u \str  \\[2mm]
 & \leq &     \norm h \norm_1 \la^4   \int_0^t       \d u   \left( \max_{0< s< t-u} \str s \str^{1-{\alpha}} \right)     \int \d s \str s \str^{{\alpha}} h(s)   \\[2mm]
  & \leq &    \norm h \norm_1  \left(\int \d s \str s \str^{\al} h(s)\right)   (\la^2 t)^2  \str\la\str^{2{\alpha}} 
\eaq 
which yields the bound \eqref{eq: statement of bound on nonlead} upon invoking Assumption \ref{ass: decay micro alpha}. 
\end{proof}

Next, we organize the contributions from leading sets of pairs. We call $Z^0_{t}$ the integral over them, i.e. in eq. (\ref{eq: z as integral over pairs repeated}) we replace
$\Omega_{[0, t]}$ by $\Omega^0_{[0, t]}$. Note that these can be written explicitly as
\beq\label{ZOT}
Z^0_{t}  = \sum_{n=0}^{\infty}    \,   \la^{2n}   \int \d \uu  \d \uv      \ldots      U_{u_3-v_2}   G_{v_2-u_2} U_{u_2-v_1} G_{v_1-u_1} U_{u_1}
\eeq
where  $G_s$ and $U_s$ were defined above. 
We prove

\begin{lemma}\label{lemma: inverse laplace of lead} 
\beq\label{remainder to markov}
\norm Z_{t }^0 - \e^{ (-\i L_{\sys} +  \la^2  M)t}   \norm_{\ga}    \leq    C  \e^{\la^2 C t}    \str \la \str^{2\al} 
\eeq
\end{lemma}

\begin{proof}
Using integrability of $G_t$ and the bound (\ref{propagation bounds}) together with
 the product structure in eq. (\ref{ZOT}), we get that the  Laplace transform 
\beq
\hat Z_{z }^0:=    \int_0^\infty \d t \, \e^{-t z} Z_{t}^0 
\eeq
is analytic in the half plane $\Re z>C\la^2$ and given explicitly by
\beq
\hat Z_{z }^0 =  \left(  z+ \i L_{\spin}+\i\la^2  m_{\mathrm{p}}^{-1}{\adjoint}(\Delta)-\la^2\hat G_z  \right)^{-1}
\eeq
where $\hat G_z=  \int_0^\infty \d t \e^{-t z} G_t$ and we recalled that $L_\sys= L_{\spin} +\la^2  m_{\mathrm{p}}^{-1}{\adjoint}(\Delta)$ . Using (\ref{eq: working bound on h}) we get  $\hat G_z$ is analytic in  $\Re z>C|\la|^2$  too
and bounded there by
\beq
 \norm \hat G_z   \norm_\ga \leq   
 C.
 \eeq
 The same bound holds for ${\adjoint}(\Delta)$.  Recall that the spectrum of $L_{\spin}$ is real and given by the set $\caE$ (\ref{bohr frequencies}). Hence
 \beq
\|\hat Z_{z }^0\|_\ga\leq  \frac{1}{\mathrm{dist}(z, \i\caE)}  \label{eq: a priori1}
\eeq
if   $\Re z>C|\la|^2$. By the resolvent identity
\beq
\hat Z^0_z  -    (z+ \i L_{\spin})^{-1} =\la^2 
\hat Z^0_z(\i  m_{\mathrm{p}}^{-1}{\adjoint}(\Delta)-\hat G_z ) (z + \i L_{\spin})^{-1}  \label{eq: bound a}
\eeq
which implies, using (\ref{eq: a priori1}) and $\Re z>C|\la|^2$,
 \beq
\|\hat Z_{z }^0P_\ve\|_\ga\leq C \frac{1}{|z + \i\ve|}  \label{eq: a priori2}
\eeq
where we recall that $P_\ve$ is the spectral projector to eigenvalue $\ve$ eigenspace of
$L_{\spin}$.

Inverse Laplace transform gives
\beq\label{inverse laplace}
Z_{t }^0 - \e^{ (-\i L_{\sys} +  \la^2  M)t}  =\int_\caC  \d z \, \e^{tz} (\hat Z^0_z  -    (z+ \i L_{\sys}-\la^2 M)^{-1} )
\eeq
where $\caC=\{z:\Re z=A|\la|^2\}$ and $A$ is taken large enough, $A>C$. By the resolvent identity we have
 \beq
\hat Z^0_z  -    (z + \i L_{\sys}-\la^2 M)^{-1} =\la^2 %(z- \i L_{\sys}-\la^2\hat G_z )^{-1}
\hat Z^0_z(\hat G_z-M) (z + \i L_{\sys}- \la^2 M)^{-1} 
  %  \label{eq: difference of contour integrals}
%&+&    P_{\ve} (z- \i L_{\sys}- L)^{-1} P_{\ve} \left(  \la^2(\hat M_z- \hat M_0)\right)   (z- \i L_{\sys}- \la^2 \hat M_z)^{-1}
\eeq
and from (\ref{eq: first definition of m}) 
 $$
 M=\sum_\ve e^{\i \ve}P_\ve\hat G_{0 }^0P_\ve .
 $$ 
The integrand in (\ref{inverse laplace}) is large when $z$ is close to $-\i\ve$, $\ve\in\caE$. Thus we localize the
contour $\caC = \cup_{\ve \in \caE} \caC_\ve$ such that on $\caC_\ve$, $-\i\ve$ is the nearest element of $-\i\caE$. 
For $z\in \caC_\ve$ we write
\beq\label{splitting of G-M}
\hat G_z -M=(\hat G_z - \hat G^{0}_z)+ (\hat G^{0}_z - \hat G^{0}_{\i\ve})+(\hat G_{\i\ve}^0-M):=
D_1^\ve+D_2^\ve+D_3^\ve.% \underbrace{ \hat G^{\la}_z - \hat G^{0}_z }_{=:D_2} +\underbrace{\hat G^{0}_z - \hat G^{0}_\ve  }_{=:D_3}+ \underbrace{\hat G^{0}_\epsilon - \hat G^{0}_0}_{=:D_1} , \qquad   z  \in    \caC_\ve
\eeq
and set
 \beq
\la^2\hat Z^0_zD_i^\ve (z+ \i L_{\sys}- \la^2 M)^{-1} :=I_i^\ve(z)
  %  \label{eq: difference of contour integrals}
%&+&    P_{\ve} (z- \i L_{\sys}- L)^{-1} P_{\ve} \left(  \la^2(\hat M_z- \hat M_0)\right)   (z- \i L_{\sys}- \la^2 \hat M_z)^{-1}
\eeq
Then, by \eqref{eq: bound a},
 \beq
 Z_{t }^0 - \e^{( -\i L_{\sys} +  \la^2  M)t }  =\sum_\ve\int_{\caC_\ve}   \d z \, \e^{tz}  (I^\ve_1(z)+I^\ve_2(z)+I^\ve_3(z)).
\eeq
Recalling (\ref{Godefi}) and  (\ref{Gdefi}) and using
$
\|U_t-U^0_t\|_\ga\leq e^{C|\la|^2t}-1
$ together with 
 $
   \Re z  =A|\la|^2$, $A$ large enough, we bound
\beq
 \norm D_1^\ve \norm_\ga
 \leq  \int_0^\infty \d t \,  h(t)   | \e^{-C|\la|^2t}-1 |\leq   \sup_{t \geq 0}  \left(\frac{ | \e^{-C|\la|^2t}-1 |}{ (1+t)^\al} \right) \, \int_0^\infty \d t  \, h(t) (1+t)^\al
 \leq C|\la|^{2\al}.     \label{eq: bound with sup t}
 \eeq
where we used (\ref{eq: working bound on h}). 

To bound  $\norm D_2^\ve \norm_\ga$ we note that 
by (\ref{eq: working bound on G}) $\hat G^0_z$ is analytic in $\Re z>0$ and by the same bound as in \eqref{eq: bound with sup t}, we get for $z\in \caC_\ve$
%and $\Re z'=0$
 \beq  %  \norm \hat G^{0}_z - \hat G^{0}_{0} \norm_\ga 
  \norm D_2^\ve  \norm_\ga 
  \leq    C \int_0^\infty \d t h(t) \left\str \e^{-z t}-e^{\i\ve} \right\str 
 \leq C\min\{1,|z +\i\ve|^{2\al}\}
\eeq
Since (\ref{eq: a priori1})  also holds for the resolvent of  ${ -\i L_{\sys} +  \la^2  M} $
\beq
\|\int_{\caC_\ve}   \d z \, \e^{tz}  (I^\ve _1(z)+I^\ve _2(z)) \|_\ga
%\|\la^2 %(z- \i L_{\sys}-\la^2\hat G_z )^{-1}\hat Z^0_z(D_2+D_3) (z- \i L_{\sys}- \la^2 M)^{-1} \|_\ga
\leq C|\la|^2 \int_{\caC_\ve}  \d z  \, \e^{z t }   \frac{1}{[z+\i\ve|^2}(|\la|^{2\al}+\min\{1,|z+\i\ve|^{2\al}\} )) \leq C \e^{C\la^2  t}  | \la|^{2\al} .
    \label{eq:  contour integrals 1}
%&+&    P_{\ve} (z- \i L_{\sys}- L)^{-1} P_{\ve} \left(  \la^2(\hat M_z- \hat M_0)\right)   (z- \i L_{\sys}- \la^2 \hat M_z)^{-1}
\eeq
Finally, to estimate $I^\ve _3$, from (\ref{eq: first definition of m}) we have $M=\sum_{\ve'}
P_{\ve'} \hat G_{\i{\ve'}}^0P_{\ve'}$ and thus
$$
D^\ve _3=\hat G_{\i\ve}^0-M=\sum_{(\ve_1,\ve_2)\neq (\ve,\ve)}P_{\ve_1}\hat G_{\i\ve}^0P_{\ve_2}
-\sum_{\ve'\neq\ve}P_{\ve'} \hat G_{\i{\ve'}}^0P_{\ve'}
$$
 Using  (\ref{eq: a priori2}) and the analogous estimate for  $P_{\ve}(z+ \i L_{\sys}- \la^2 M)^{-1} $ we
 get 
 $$
 \|I_3^\ve(z)\|_\ga\leq C\la^2\sum_{(\ve_1,\ve_2)\neq (\ve,\ve)}(|z+\i\ve_1||z+\i\ve_2|)^{-1}
 $$
 and so
 \beq
\|\int_{\caC_\ve} e^{tz} I_3(z)dz\|_\ga
\leq C|\la|^2\sum_{(\ve_1,\ve_2)\neq (\ve,\ve)} \int_{\caC_\ve}  \d z  \, \e^{z t } (|z+\i\ve_1||z+\i\ve_2|)^{-1}
\leq C \e^{C\la^2  t}  | \la|^{2} |\log|\la||
    \label{eq: contour integrals 2}
%&+&    P_{\ve} (z- \i L_{\sys}- L)^{-1} P_{\ve} \left(  \la^2(\hat M_z- \hat M_0)\right)   (z- \i L_{\sys}- \la^2 \hat M_z)^{-1}
\eeq
 The bounds  (\ref{eq: contour integrals 1}) and (\ref{eq: contour integrals 2}) yield the claim.

\end{proof}

For completeness, we note that Lemmata \ref{lem: bound on nonlead} and \ref{lemma: inverse laplace of lead}  lead in a straightforward way to 
\begin{proposition}\label{prop: weak coupling}
For any $\frt <\infty$, and $\la$ small enough, 
\beq
\sup_{t < \la^{-2} \frt} \norm Z_{t}  -  \e^{  (-\i L_\sys+ \la^2  M)t }  \norm_{20 \ga_0}     \leq C \e^{C\frt} \str\la\str^{2\al}
\eeq
\end{proposition}
This is known as Davies' weak coupling limit and it was first proven in \cite{davies1}\footnote{To be precise, \cite{davies1} deals only with confined systems so that Propostion \ref{prop: weak coupling} is not covered. However, the arguments are essentially identical in both cases.}

\subsection{Analysis of the semigroup $\e^{t Q}$} \label{sec: analysis of the semigroup}

We will now analyze the Markov semigroup    $ \e^{ t (-\i L_\sys+ \la^2  M) } $ that was derived in the previous section.  
The splitting of the generator into $-\i L_\sys$ and $\la^2  M$ is logical from the point of view of the derivation from the microscopic dynamics, since $M$ represents the contribution due to the interaction with the environment, but it is not optimal for the analysis of the semigroup, instead, we write
 \beq\label{introduce Q}
-\i L_\sys+ \la^2  M =  -\i L_\spin + \la^2 Q, \qquad   \textrm{with} \,  Q     :=   -\i \adjoint(\frac{1}{m_{p}} \Delta)  + M 
 \eeq
 such that on the RHS, the first term is of $\caO(1)$ and the second is proportional to $\la^2$.  Moreover, both terms commute and consequently it suffices to investigate the semigroup $t \mapsto \e^{t \la^2 Q}$, or simply $t \mapsto \e^{t  Q}$.

\subsubsection{Fourier transform revisited} \label{eq: other fiber decomposition}

Note that up to this point, we have viewed $\rho$ as a function $\rho(x,s)$ on $\bbA_0$  and $M$ and $Q$ as kernels 
$K(x',s';x,s)$ on $\bbA_0 \times \bbA_0$. In particular, translation invariant kernels were studied in terms
of the Fourier transform $\hat K(p;s,s')$ in the variable $x'-x$. 
For the present discussion, it is easier to use a slightly different, in fact more natural, Fourier representation. Namely we replace the previously used position variable  $\lfloor \frac{x_\links + x_\rechts}{2} \rfloor $ by $\frac{x_\links + x_\rechts}{2}$. 
Recall from Section \ref{Diffusion} that  $\rho$ can be also be viewed as a kernel $\rho(x_\links, x_\rechts) \in \scrB(\scrS)$, with $x_\links,x_\rechts \in \bbZ^d$. Similarly, any
$K$ in $\scrR$ can be viewed as a kernel 
\beq
K(x'_\links, x'_\rechts, x_\links, x_\rechts )    \in   \scrB( \scrB(\scrS)).
\eeq
Define for $\rho\in \scrB_2(\scrH_\sys)$
\beq
\tilde\rho (p,k)  =    \sum_{x_\links, x_\rechts }   \e^{\i p \frac{x_\links + x_\rechts}{2}}  \e^{\i k  (x_\links - x_\rechts) }   \rho(x_\links, x_\rechts)
\eeq
which is in $L^2(\tor \times\tor , \scrB_2(\scrS))$.
Denote by $\tilde\rho (p):=\tilde\rho (p,\cdot) $. 

In Section \ref{sec: fourier transform}, we defined the Fourier transform $\hat \rho(p):=\hat \rho(p,\cdot)$ wrt.\ to the variable $x=\lfloor \frac{x_\links + x_\rechts}{2} \rfloor $. The two transforms are closely related, namely a short calculation
gives
$$ \tilde\rho (p) =  \caI_p  \hat\rho(p) $$
where  $ \caI_p:  l^2(\caS) \to L^2(\tor, \scrB_2(\scrS))$ is the unitary map
\beq\label{Ipdefi}
( \caI_pf)(k)   =      \sum_{v, \eta} f ( v, \eta)        \e^{\i p \frac{\eta}{2}}  \e^{\i k(2 v +\eta) }.
\eeq
where on the RHS, $f ( v, \eta)   \in \scrB_2(\scrS)$, and we recall that $v,\eta$ are the (position-like) coordinates that, together with $e_{\links}, e_{\rechts} \in \sigma(H_{\spin})$, constitute  the variable $s \in \caS$.

For a translation invariant kernel $K$ with $\|K\|_\ga <\infty$ for some $\ga>0$, we have as before
\beq
(\widetilde{K\rho} )(p)  = \tilde K(p)\tilde\rho(p).
\eeq
where $\tilde K(p)\in\scrB(L^2(\tor, \scrB_2(\scrS)))$ (cfr.\ the Remark in Section \ref{sec: definition of norm}).
Since $Q$ is  translation-invariant %operator, we can define its fiber decomposition $Q=\int^{\oplus}\d p\widetilde Q(p)$ as well; Let us for the time being consider $Q$ as an operator on $\scrB_2(\scrH_\sys)$, then $Q(p)$ acts on $L^2(\tor, \scrB_2(\scrS))$. 
%However, $Q$ has yet another symmetry that can be used to reduce the description, namely 
and  $[Q, L_\spin]=0$, as we see from (\ref{eq: first definition of m}),   we can decompose further
\beq\label{epsilondeco}
\tilde Q(p) = \oplus_{\ve \in \caE} \tilde Q_{\ve}(p)
\eeq
 where $\tilde Q_{\ve}(p) := P_{\ve}\tilde Q(p) P_{\ve}$ with $P_\ve$ the spectral projections of $L_{\spin}$.
 %(recall that $\caE$ is the spectrum of $L_\spin$). %We will combine this decomposition with the fiber decomposition to obtain
%\beq
%Q = \mathop{\oplus}\limits_{\ve \in \caE}  \mathop{\int}\limits^{\oplus}_{\tor}  \d p \, \widetilde Q_{\ve}(p), \qquad \textrm{with} \, \widetilde Q_{\ve}(p) = P_{\ve}\widetilde Q(p) P_{\ve}
%\eeq
By Assumption \ref{ass: fermi golden rule},  the spaces $P_\ve(\scrB(\scrS)), \ve \neq 0$ are one-dimensional and  and hence we can identify
$$
\widetilde Q_{\ve}(p) \in \scrB(L^2(\tor)), \ \ \ve\neq  0
$$ 
The space $P_{0}(\scrB(\scrS))$ has $\dim \scrS$-dimensions and there is a natural basis labelled by $e \in \sp(H_{\spin})$, such that  $P_{0}(L^2(\tor, \scrB_2(\scrS)))$ is identified with  $L^2(\caF)$, with $\caF = \sp(H_{\spin}) \times \tor$ and the measure is counting $\times$ Lesbegue,  and so we identify
$$
\widetilde Q_{0}(p) \in \scrB(L^2(\caF)).
$$

\subsubsection{Lindblad representation for $Q$} \label{eq: alternative constuction of m}

We can write a more  explicit expression for $M$ of 
(\ref{eq: first definition of m}) using (\ref{Godefi}) and (\ref{kuv expression}):
\beq
(M \rho )(x,y)  =      \sum_{\ve} \left(  \zeta_\ve(x-y)  W_{\ep}^*    \rho(x,y) W_{\ep} 
- \int_{0}^{\infty} \d s \zeta(0, s)   \e^{-\i s \ve} W_{\ep}W_{\ep}^*     \rho(x,y) 
-  \int_{-\infty}^{0}  \d s \zeta(0, s)   \e^{-\i s \ve}     \rho (x,y) W_{\ep}W_{\ep}^* \right).
\label{eq: three terms in m}
\eeq 
In (\ref{eq: three terms in m}) we have denoted $W_\ve:=P_\ve (W)$,  $ \zeta_\ve$ is the Fourier transform of $\zeta$ defined in  (\ref{ftxi}) and we use the notation $\rho(x,y)\in\scrB(\scrS)$ with $x,y \in \bbZ^d$ (see Section \ref{sec: initial state}).
To derive (\ref{eq: three terms in m}) from (\ref{kuv expression}) and  (\ref{Godefi}) , one needs to calculate
$$
P_\ve\caT((W\otimes \lone_x)^a\otimes U^0_t \otimes(W\otimes \lone_y)^b)P_\ve=
P_\ve W^aU^0_tW^bP_\ve
\otimes \lone_x^a\lone_y^b.
$$
To proceed write $P_\ve$ in terms of the projectors
 $\pi_e\in\scrB(\scrS)$  to the basis vector labelled by $e \in \sp(H_{\spin})$.
For $0\neq\ve=e_1-e_2$ we have $P_\ve(\rho)=\pi_{e_1}\rho\pi_{e_2}$ and for $\ve=0$ we have
$P_0(\rho)=\sum_e\pi_{e}\rho\pi_{e}$. Some algebra then gives the representation (\ref{eq: three terms in m}). %Starting from the definition of $M$ and $Q$ in Section \ref{sec: weak coupling limit}, we derive an explicit expression, namely
%\beq Q \rho =     \underbrace{\i \adjoint(H_{\lamb})\rho}\limits_{=: Q^{\lamb} \rho}  +  \underbrace{ \i \adjoint( \Delta) \rho}\limits_{=: Q^\kin \rho} +     \underbrace{ \Phi(\rho)}\limits_{=: Q^{gain}\rho }   \underbrace{-  (1/2)   \Phi^\star(1) \rho  -  (1/2)  \rho    \Phi^\star(1)}\limits_{=: Q^{loss}\rho }\eeq
%Here, $\Phi$ is a completely positive map on $\scrB_1(\scrH_\sys)$, $H_{\lamb}$ is an effective Hamiltonian acting on $\scrS$ (both constructed below) and  $\Phi^\star$ is the adjoint 
%of  $\Phi$ wrt.\ to the trace, i.e. $\Tr \rho \Phi^\star(O)  = \Tr O \Phi(\rho) $.
%To derive this representation, let us start from expression \eqref{eq: first definition of m} and evaluate the operator $G_s$ explicitly. This leads to 
%\baq M \rho   =      \sum_{x_\links, x_\rechts}  \sum_{\ve}  \Bigg(&& \int_{\bbR} \d s \zeta(x_\links-x_\rechts, s)   \e^{\i s \ve} (W_{\ep}   \otimes 1_{x_\links} )   \rho ( W^{*}_{\ep}   \otimes 1_{x_\rechts} )  \nonumber \\[1mm]&+& \int_{0}^{\infty} \d s \zeta(0, s)   \e^{\i s \ve} (W_{\ep}   \otimes 1_{x_\links} ) ( W^{*}_{\ep}   \otimes 1_{x_\rechts} )    \rho  \nonumber\\[1mm]&+&   \int_{-\infty}^{0} \d s \zeta(0, s)   \e^{\i s \ve}     \rho  (W_{\ep}   \otimes 1_{x_\links} ) ( W^{*}_{\ep}   \otimes 1_{x_\rechts} )  \Bigg)\label{eq: three terms in m}\eaq 

Denote the operator acting on $\rho$ in the first term of (\ref{eq: three terms in m}) by $\Phi$:
 \beq
(\Phi\rho)(x,y)   : =  \sum_{\ve}   \zeta_{\ep}(x-y)      W_{\ep}^*    \rho(x,y)    W_{\ep} . 
%, \qquad   W_{\ve} := P_{\ve}(W)
\eeq
Since $ \zeta_{\ep}(\cdot)$ is bounded,  one easily checks that $\Phi$  is bounded on $\scrB_p(\scrH_\sys), 1\leq p \leq \infty$ and completely positive, and that there is a bounded and completely positive map $\Phi^{\star}$ on $\scrB(\scrH_\sys)$ such that $\Phi$ acting on $\scrB_1$ is the adjoint of $\Phi^{\star}$. The pairing between $\scrB_1$ and $\scrB_\infty$ is given by the trace, i.e.\  $\Tr \rho \Phi^\star (O)  = \Tr O \Phi (\rho) $ (whereas `${}^{*}$' in $W_\ve^*$  denotes the Hermitian adjoint in $\scrB(\scrS)$). Concretely, 
\beq\label{phistar1}
 \Phi^\star (O)(x,y)=\sum_{\ve}   \zeta_{\ep}(y-x)      W_{\ep} O(x,y)     W_{\ep}^*
 \eeq

 In the second term of \eqref{eq: three terms in m}, using $\overline{ \zeta(0, s)}= \zeta(0, -s)$ we have $$ \int_{0}^{\infty}\d s \e^{-\i s \ve}  \zeta(0, s)= (1/2)
 \zeta_{\ep}(0) + \i t_{\ep}(0) $$ 
 with $t_{\ep}(0) $ real, and in the third term one has $ \int_{-\infty}^{0}\d s \,  \e^{-\i s \ve}  \zeta(0, s)= (1/2)  \zeta_{\ep}(0) - \i t_{\ep}(0) $.  Thus altogether we may write (\ref{eq: three terms in m}) as
 \beq
M \rho   =    \Phi(\rho)-  \frac{1}{2}(   \Phi^\star(\lone) \rho +  \rho    \Phi^\star(\lone))+\i [H_{\lamb},\rho]
\label{eq: M in terms of Phi}
\eeq 
and the "Lamb shift" to the Hamiltonian by
\beq
H_{\lamb} :=   \sum_{\ve}  t_{\ep}(0)     W_{\ep} W^*_{\ep}.
\eeq
   In terms of the Fourier transform variables of Section \ref{eq: other fiber decomposition} we have
\beq   \label{eq: explicit widetildephi}
(\widetilde{\Phi\rho})(p,k) =   \sum_{\ep}  \int \hat \zeta_\ep(\d q)       W_{\ep}      \tilde{\rho}(p,k-q)
  W^*_{\ep}   \eeq
  where the positive measure $\hat\zeta_\ep(\d q)$ is the Fourier transform of the function $  \zeta_{\ep}(x)$, introduced in \eqref{xihat}.
  Note that $\widetilde{\Phi}(p)$ is independent of $p$ and hence we will write 
  $$\widetilde{\Phi}=\widetilde{\Phi} (p).$$     
Let us next decompose as in (\ref{epsilondeco}). We get
$$
\widetilde{\Phi}_{\ve } =0\ \ {\mathrm{ if}} \ \ \ve\neq 0.$$ 
Indeed, writing $\ve=e_1-e_2$, we have $P_\ve(W_{\ve'}^*  P_\ve (\rho)   W_{\ve'} )=\pi_{e_1}W_{\ve'}^* \pi_{e_1}\rho\pi_{e_2}W_{\ve'}
\pi_{e_2}$ which obviously vanishes unless $\ve'= 0$. However, for  $\ve'=0$, 
$\zeta_{\ve'}=0$ by Assumption \ref{ass: fermi golden rule}. 

Consider then $\widetilde{\Phi}_{0 }$. As explained in the previous section it can be viewed as acting on functions
$\varphi(e,k)$ on $\sigma(H_\spin)\times\bbT^d$. Recalling the definition
(\ref{rates}) of the jump rates $j$, eq. (\ref{eq: explicit widetildephi}) becomes
%(In particular, $\widetilde{\Phi}_{\ve=0}$ preserves the positivity of functions $\varphi$ on $\caF$.)
\beq
(\widetilde{\Phi}_{0} \varphi) (e',k') =    \sum_{e\in\sigma(H_\spin)} \int  \,    j(e',  \d k; e,0 )   \varphi  (e,k'-k)
\eeq
and since $\hat\zeta_\ve$ is a finite positive measure,  $\widetilde{\Phi}_{0} $ defines a bounded
positivity preserving operator on  $L^1(\caF)$,  $L^{\infty}(\caF)$, and,  by Riesz-Thorin interpolation, also on $L^q(\caF),  1< q <\infty$.

Denote the escape rates of the Markov process by
\beq
w(e) :=    \sum_{e'}  \int j( e',\d k'; e,0)
 \eeq 
We can now give explicit expressions for the operators in the decomposition (\ref{epsilondeco}):
 \begin{proposition}
 \label{prop: properties of Q}
 (a) For $\ve=0$
\beq
 \widetilde  Q_{0} (p)=\widetilde\Phi_0+q_0(p)
\eeq
 where $\widetilde\Phi_0$  is a compact operator on  $L^1(\caF )$.   It is also positivity improving, i.e.\, if $\varphi \geq 0$, then  $  \e^{t \widetilde{\Phi}_{0}} \varphi >0$ (i.e.\ the function is strictly postive a.e.)  for any $t>0$. $q_0(p)$ is a multiplication
 operator by the function
   \beq\label{qopek}
 q_0(p)(e,k)=-w(e)+ \i E_{\kin}(p,k) \eeq
 where $ E_{\kin}(p,k)=\frac{2}{{m_{\mathrm{p}}}} \sum_{j=1}^d (\cos(\frac{p_j}{2}+ k_j)-\cos(\frac{p_j}{2}-k_j) )$ and 
  \beq\label{minwe}
 \min w:=\min_e w(e)>0.\eeq

 \noindent (b) For $\ve=e-e'\neq 0$, $Q_{\ve} (p)$ is a multiplication operator  by the function
   \beq
   \label{qvepek}
 q_\ve(p)(e,k)= -1/2(w(e) + w(e'))+ \i (E_{\kin}(p,k) +H_{\mathrm{Lamb}}(e)- H_{\mathrm{Lamb}}(e'))
 \eeq
 with $H_{\mathrm{Lamb}}(e) = \Tr_\scrS( \pi_e H_{\mathrm{Lamb}} \pi_e) $. In particular, the term $\i (\ldots)$ is purely imaginary.
  \end{proposition}
 \begin{proof}
 Compactness of  $\widetilde{\Phi}_{0}$  follows from  the fact the function  $ x \to   \zeta_{\ep}(x) $ decays at infinity for any $\ve \in \caE$  by Assumption \ref{ass: fermi golden rule}.   The positivity improving property of $  \e^{t \widetilde{\Phi}_{0}}$ follows directly from the irreducibility assumption in Assumption \ref{ass: fermi golden rule}.     For the same
 reason the escape rates are strictly positive, i.e.\ (\ref {minwe}).
 
 In (\ref{qopek}) the first term on the RHS comes from the $\Phi^\star(\lone)$ terms in (\ref{eq: M in terms of Phi})
 and the second one from the Laplacian term in (\ref{introduce Q}). In (\ref{qvepek}) the last term
 comes from the last term in (\ref{eq: M in terms of Phi}). This term is zero for $\ve=0$.
 
\end{proof}

\subsubsection{Perron-Frobenius theorem for maps on $\scrB(\scrH_\sys)$} \label{sec: noncommutative perronfrobenius}

\begin{definition}
A bounded positive map $\phi$ on $\scrB(\scrH_\sys)$ is ergodic if for any $A \geq 0$, $A\neq 0$, $\e^{t \phi}A >0$ for any $t>0$.  
\end{definition}
\begin{proposition}[Schrader \cite{schrader}]\label{prop: perron frobenius schrader}
Assume that a positive map $\phi$ on $\scrB(\scrH_\sys)$ is ergodic and that its spectral radius $r(\phi)$ is an eigenvalue with corresponding eigenvector $S$. Then $r(\phi)$ is a simple eigenvalue and $S$ can be chosen positive definite: $S >0$.
\end{proposition}

We will apply the above theorem to the map $\e^{t Q^{\star}}$, with $t>0$ arbitrary.   First we establish

\begin{lemma} The map $\Phi^{\star}$ is ergodic on $\scrB(\scrH_\sys)$.
\end{lemma}
\begin{proof}
By duality, it suffices to prove that $\e^{t \Phi}(\rho) >0$ for any nonnegative $\rho \in \scrB_1(\scrH_\sys)$.  Employing the fiber decomposition, we see that this is equivalent to $\e^{t \widetilde \Phi_{\ve=0}} \varphi >0$ for any nonnegative function $\varphi$ on $\caF$. This was proven in Proposition \ref{prop: properties of Q}. 
\end{proof}

\begin{lemma} The map $\e^{t Q^{\star}}$ is ergodic on $\scrB(\scrH_\sys)$.
\end{lemma}
\begin{proof}
 First, we note that the generator $Q^{\star} $ has the form 
 \beq
Q^{\star}(O)= \Phi^{\star} (O) + B O + O B^*, \qquad   B=  \i (\frac{1}{{m_{\mathrm{p}}}}\Delta + H_{lamb}) - \frac{1}{2}   \Phi^\star(\lone)
 \eeq
Since $V_t:  O \to \e^{t B} O\e^{t B^*} $ is a completely positive map, for any operator $B$, all integrands in the Duhamel expansion
\beq
\e^{ t Q^{\star} }   = 1 +  \sum_{m =1}^{\infty}        \mathop{\int}\limits_{0 < t_1 < \ldots < t_m <t}  \d t_m  \ldots \d t_1     \caQ_t(t_1, \ldots, t_m)
\eeq with $  \caQ_t(t_1, \ldots, t_m) :=  V_{t-t_m} \Phi^{\star} \ldots   \Phi^{\star}   V_{t_2-t_1}  \Phi^{\star} V_{t_1} $
are completely positive maps, as well. 

Therefore,  for any positive $O \in \scrB(\scrH_\sys)$, it suffices to find an $m_0$ and a subset of positive measure of the $m_0$-dimensional simplex such that   
$\caQ_t(t_1, \ldots, t_{m_0})(S)>0$.
Let $m_0$ be smallest natural number $m$ such that  $( \Phi^\star)^m O>0$. The ergodicity of $\Phi^\star$ ensures that $m_0< \infty$. 
Then by continuity, there is an neighborhood of $t_1=t_2 = \ldots = t_{m_0}=0$ in the simplex where the integrand is bounded away from $0$. 
\end{proof}
We conclude that
\begin{lemma} \label{lem: eigen of qstar}The operator  $Q^{\star}$ has only one eigenvalue, namely $0$, whose real part is nonnegative. This eigenvalue is simple and its eigenvector can be chosen to be $\lone$.
\end{lemma}
\begin{proof}
The fact that $\lone$ is an eigenvector with eigenvalue $0$ since $Q^{\star}$ generates a unity-preserving semigroup. For the same reason, there are no eigenvalues (in fact, no spectrum) with strictly positive real part. Assume that $Q^{\star} S=\la S$ for some $\la, \Re \la=0$ and $S \in  \scrB(\scrH_\sys)$. Since $Q^{\star}$ is bounded, we have $\e^{tQ^{\star} } S= \e^{t\la}S$ and we can choose $t$ such that $\e^{t\la}=1$.  Proposition \ref{prop: perron frobenius schrader} implies now that $S=\lone$.  
\end{proof}

Consider  $\widetilde{Q}(p) $  as bounded operators acting on $L^1(\scrB_1(\scrS))$.  It is straightforward to check that there are  bounded operators $\widetilde{Q}^\star(p)$ 
acting on $L^{\infty}(\scrB(\scrS))$ such that  $\widetilde{Q}(p) $  is the adjoint of $\widetilde{Q}^\star(p)$.   The above result has implications for the spectrum of $\widetilde{Q}^\star(p)$ which we describe below. The internal space $\scrS$ does not play any role in this argument and therefore we pretend it is 1-dimensional so that $\scrH_\sys \sim L^2(\tor)$ and $\widetilde{Q}^\star(p) $ acts on function on $\tor$ (until the end of the proof of the following lemma). The generalization to finite-dimensionsal $\scrS$ is then obvious. 
%For an operator $\rho \in \scrB_2(\scrH_\sys)$, we have fiber decomposition 
% \beq
% \rho= \int_{\oplus} \d p \widetilde\rho(p) 
% \eeq
%where $ \widetilde\rho(p) $ is defined for a.e. $p \in \tor$.   
Note first that for  $\rho \in \scrB_1(\scrH_\sys)$, the fiber $\widetilde \rho(p)$ can be uniquely defined for any $p$, such that 
\beq \label{eq: fiber bone}
\Tr [\e^{\i p X} \rho_p] = \langle 1, \widetilde\rho(p) \rangle
\eeq
where $\langle \cdot, \cdot \rangle$ is the pairing between $L^{\infty}$ and $L^{1}$ and $X$ is the operator acting by multiplication with $x \in \bbZ^d$ on $l^2(\bbZ^d)$.  This follows from the singular value decomposition. 
\begin{lemma}  
For $p \neq 0$,  $\widetilde{Q}^\star(p)$ does not have any eigenvalue with nonnegative real part.  For $p=0$, the only eigenvalue with nonnegative real part is $0$ and  it is simple.
\end{lemma}
\begin{proof} 
Assume $\widetilde{Q}^\star(p)\varphi=  \la_p \varphi$ for some $\varphi \in L^\infty$. We construct
 a related eigenvector of  $Q^\star$ with eigenvalue $\la_p$. Indeed, let
\beq
S_{\varphi,p}   = O_{\varphi} \e^{\i p X}  \in \scrB(L^2)
\eeq
where $O_\phi \in \scrB(L^2)$ is  the operator that acts by multiplication with the function $\phi(k), k \in \tor$.
We have $\widetilde{O\rho}(p)= \phi \widetilde{\rho}(p)$, i.e.\ the pointwise product of the functions, and hence, by \eqref{eq: fiber bone},
\beq
\Tr [ \e^{\i p X} O_\phi \rho] = \langle {\phi}, \widetilde \rho(p) \rangle
\eeq
 Hence, for all $\rho\in\scrB_1$
\beq
\Tr [ S_{\varphi,p}  Q \rho] =     \langle    \varphi,   \widetilde{Q}(p) \widetilde{\rho}(p) \rangle  = \langle  \widetilde Q^{\star}(p)  \varphi,  \widetilde{\rho}(p) \rangle   =   \la_p \langle   \varphi,  \widetilde{\rho}(p) \rangle    = \la_p  \Tr [ S_{\varphi,p}  \rho]    
\eeq
and it follows that  $Q^\star S_{\varphi,p}   =   \la_p  S_{\varphi,p} $. If $\Re \la_p \geq 0$, then Lemma \ref{lem: eigen of qstar} implies that $S_{\varphi,p} =C \lone$, which can only be true if $p=0$ and $\varphi= C 1$.  Hence the geometric multiplicity of the eigenvalue $0$ is $1$. In an analogous way, one argues that the algebraic multiplicity is $1$ so the eigenvalue at $0$ is simple. 

\end{proof}

Of course, the conclusions about the spectrum of $\widetilde{Q}^\star(0)$ could also  have been obtained via a `commutative' Perron-Frobenius type argument, by first restricting to  $\widetilde{Q}_0^\star(0)$ and using that this operator generates a Markov process with the irreducibility property of Assumption \ref{ass: fermi golden rule}.

\subsubsection{Spectral properties of $\widetilde Q_\ve(p)$} \label{sec: spectral properties of fibered operators}

We can now state the main result on spectral properties of  $\widetilde Q_\ve(p)$:
 \begin{proposition} \label{prop: main spectral properties of Qvep}
 There are strictly positive constants $\ga_Q,\bana_Q,a_Q,b_Q$ s.t. the following holds.
 \begin{itemize}
\item[a)]
  $\widetilde{Q}_{\ve}(p)$ extend to  an analytic family of bounded operators on  $L^1(\caF)$ for $\ve=0$
  and on $L^1(\bbT^d)$ for $\ve\neq 0$, in the region $|\Im p|<\ga_Q$.
\item[b)]
 Take $p$ with $|\Re p|\leq\bana_Q$ and $|\Im p|<\ga_Q$. Then
 the spectrum of $\widetilde{Q}_{0}(p)$ lies in the half plane $\Re z<-a_Q$ except for a simple isolated eigenvalue at $f_Q(p)$
 with $\Re f_Q(p)>-a_Q/2$ and $p \mapsto f_Q(p)$ analytic.  The spectrum of $\widetilde{Q}_{\ve}(p)$ lies in the half plane $\Re z<-a_Q$.
\item[c)] 
  Take $p$ with $|\Re p|\geq\bana_Q$ and $|\Im p|<\ga_Q$.
 Then
  $\sigma(\widetilde{Q}_{\ve}(p))\subset\{z:\Re z\leq -b_Q\}$ for all $\ve$.  
\item[d)] The  claims  a),b),c) remain true for the operators
   $
 \e^{\ka \partial_{k}} \widetilde Q_{\ve} (p)  \e^{-\ka  \partial_{k}}   $ 
 for $\ka\in\bbC$, $|\ka|\leq \ka_0$ for some $\ka_0>0.$ 
  \item[e)] The analytic function  $f_Q(p)$ satisfies
  \beq\label{fQasy}
|f_Q(p)+D_Qp^2|\leq C|p|^4 
\eeq
  for a strictly positive 'diffusion constant' $D_Q>0$. The corresponding left eigenvector and right eigenvectors
  can be chosen to be analytic in $p$ such that, at $p=0$, the former equals $1_{\caF}$ and the latter $\mu_Q$, where the function $\mu_Q(k,e)$ is independent of $k$, normalized to $\int \d k \sum_e \mu_Q(k,e)=1$ and such that  $C> \mu_Q(k,e) \geq c>0$. 
   \end{itemize}
  \end{proposition}

\begin{proof}
a) follows since by Proposition \ref{prop: properties of Q} the only $p$-dependence
of  $\widetilde{Q}_{\ve}(p)$ is in the multiplication operators which are analytic (in fact,  the restriction on $\Im p$ is not even needed here). 
For $\ve\neq 0$, the spectral claims b),c) follow from Proposition \ref{prop: properties of Q}  since the real part
of the multiplication operator is negative. 

We focus now on $\ve=0$. 
 Since $\widetilde \Phi_0$ is compact and the spectrum of $q_0$ lies in
a half space $\{z:\Re z\leq -b\}$  with $b>0$ the spectrum of $\widetilde{Q}_{0}(p)$ in  $\{z:\Re z\geq -b/2\}$ consists of
a finite number of
eigenvalues of finite multiplicity. Therefore, we can use duality between $\widetilde{Q}^\star(p)$ and $\widetilde{Q}(p)$ to conclude that the latter has a unique and simple eigenvalue with nonnegative real part for $p=0$, namely $0$, and none for $p \neq 0$. These properties carry over to the spectrum of $\widetilde{Q}_0(p)$.  As we already know that $\widetilde\Phi_0$ generates a positivity-improving semigroup, we immediately deduce that the eigenvector of $\widetilde{Q}_0(0)$ at $0$ is a strictly positive function  $\mu_Q(k,e)$.
  From the fact that the jump rate depends on $k,k'$ only through $k'-k$, it follows that $\mu_Q(k+q,e)$ is also invariant i.e.\ that $\mu_Q$ is constant in $k$. Therefore, the strict positivity and the fact that the number of $e$'s is finite, imply that $\mu_Q$ is bounded below and above. 
By analyticity  for $\ga_Q,\bana_Q$ small enough the simple eigenvalue
persists for small $\str p \str$and lies in  $\Re z>-a_Q/2$ whereas the rest of spectrum is in $\Re z<-a_Q$ for some $a_Q>0$.  We have now proven all claims of b).  To end the proof of c), it suffices to realize that $p$-space is compact and the eigenvalues depend analytically on $p$.

To get d), note that the operator $\partial_k$ commutes with $\widetilde\Phi_0$ so that
\beq\label{complex transl}
\widetilde Q_{\ve}^\kappa (p):=\e^{\ka \partial_{k}} \widetilde Q_{\ve} (p)  \e^{-\ka  \partial_{k}} =\widetilde Q_{\ve} (p) 
+ \i( E_{\kin}(p,k-\ka) - E_{\kin}(p,k))
\eeq
 Hence by perturbation theory a),b),c) remain valid if $|\ka|$ is small enough.

Finally, we turn to e).
By perturbation theory of isolated eigenvalues, we get (recall that $\langle \cdot, \cdot \rangle$ is the pairing between $L^{1}$ and $L^{\infty}$.)
\beq
f_Q(p) = \langle 1, (Q_0(p)-Q_0(0)) \mu_Q \rangle + O(\str p \str^2)%=  \langle 1, Q^\kin_0(p) \mu_Q \rangle + O(\str p \str^2).
\eeq
The first term on the RHS vanishes by the explicit expression (\ref{qopek}) and the fact that $\mu_Q$ is constant in $k$.  More generally,  the absence of odd powers in the Taylor expansion at $p=0$, follows from the invariance of the model under lattice symmetries, cfr.\ the reasoning in the proof of Lemma \ref{lem: flow of parameters for t}. 
  To determine the second order contribution to   $f_Q(p) $ we invoke second order spectral perturbation theory, yielding 
\beq
f_Q(p) =- \sum_{i,j=1}^d p_ip_j  \langle 1, v_j {{\widetilde Q}_0(0)}^{-1}   v_i \mu_Q \rangle + O(\str p\str^4),
\eeq
where $v_i = - \i \partial_{p_i} E_{\kin}(p,\cdot)$ is the "group velocity" (\ref{group velocity}). The  inverse of ${\widetilde Q}_0(0)$ in this formula is well-defined because $R_Q(0) v_i \mu_Q =0$ where
\beq\label{spectral proj}
R_Q(p) = |\mu_Q(p)\rangle \langle 1|
\eeq
 is the spectral projection to $\mu_Q$.
 
The non-negativity of $D_Q$ is deduced from the fact that the operators $\widetilde Q_p$ (for $\Im p=0$) generate contractive semigroups. 
To establish the strict positivity of $D_Q$, we employ a standard construction:
 First, consider the space $L^2(\caF, \mu_Q)$, defined by the scalar product  $\langle\psi, \psi'\rangle_{\mu_Q}:=  \langle  \mu_Q \psi, \psi'  \rangle$ where $\psi \mu_Q$ is the pointwise product of two functions,  and let the operator $K$ be defined as the closure in $L^2(\caF, \mu_Q)$ of the operator
 \beq
 K \psi =       \mu^{-1}_Q Q_0(0) (\mu_Q \psi),  \qquad  \psi \in L^1(\caF)
 \eeq
 where the multiplication operator $\mu^{-1}_Q$ is well-defined because $\mu_Q$ is bounded below.  By Proposition \ref{prop: properties of Q} and the remarks preceding it, we have that $K$ is bounded.  Furthermore, $K$ is sectorial, that is, there is a constant $C_K>0 $ such that
 \beq
 \str \langle \psi,  (\Im K)  \psi \rangle_{\mu_Q}  \str \leq  - C_K \langle \psi,  (\Re K)  \psi \rangle_{\mu_Q} 
 \eeq
 where  $ 2\Re K =  K + K^* , 2\Im K = K-K^* $ and the adjoint is taken in the Hilbert space $L^2(\caF, \mu_Q)$. 
  Both the left and right eigenvectors of $K$ are $1$ (constant function on $\caF$) and the diffusion constant can be represented as 
 \beq
  D_Q \delta_{i,j} = \langle v_i,    K^{-1} v_j \rangle_{\mu_Q}
 \eeq
 Moreover, $K$ inherits the unicity of the eigenvalue at $0$ and  the spectral gap from $\widetilde Q_0(0)$. It follows that $K^{-1} v_j \in L^2(\caF,\mu_Q)$. Using the sectoriality we can now deduce that 
  $  D_Q \neq 0$.  In the case where $\be_1=\be_2$, the operator $K$ is self-adjoint and the above reasoning can be slightly simplified.  
  Since the claimed properties of $\mu_Q$ were already established above, we have now proven all claims in e). 
\end{proof}

\subsection{Proof of Proposition \ref{ass:  properties of T} for $n=0$} \label{sec: t at n equal 0}

Recall that $T_0= Z_{\la^{-2}\frt_0}$.
From Proposition \ref{prop: weak coupling}, the relation $ -\i L_\sys+ \la^2M=-\i L_{\spin}+ \la^2 Q$ and the fact that  $L_{\spin}$ and $ Q$ commute,  we infer
\beq
\norm T_0  -  \e^{-\i \la^{-2} \frt_0 L_{\spin}} \e^{\frt_0 Q} \norm_{20 \ga_0} \leq C\str \la\str^{2\al} . \label{eq: error term weak coupling t}
\eeq
Then we set
\beq
 T_Q:= \e^{-\i \la^{-2} \frt_0 L_{\spin}} \e^{\frt_0 Q}  =   \sum_{\ve} \e^{-\i \frt_0  \la^{-2}\ve} P_{\ve}    \e^{\frt_0 Q}P_{\ve} .
\eeq
We need to translate the spectral information of  Proposition
\ref{prop: main spectral properties of Qvep} to the $\gamma$-norm.  Recall that $ \e^{\frt_0 \hat Q(p)}=
\caI^{-1}_p   \e^{\frt_0  \widetilde Q(p)}    \caI^{}_p$ and that from (\ref{Ipdefi}) we have
$$
e^{\ka\partial_k} \caI^{}_p= \caI^{}_p \e^{\i \ka(2v+\eta)}.
$$
Thus
\beq
\e^{\frt  \hat Q(p)}(v',\eta';v,\eta)=e^{\i \ka(2(v-v')+\eta-\eta')}(\caI^{-1}_p   \e^{\frt  \widetilde Q^\ka(p)}    \caI^{}_p)
(v',\eta';v,\eta)
\eeq
where we use the notation (\ref{complex transl}) and both sides are operators on $\scrB_2(\smallspace)$. Since
\beq
(\caI^{-1}_p   \e^{\frt_0  \widetilde Q^\ka(p)}    \caI^{}_p)
(v',\eta';v,\eta)=\langle\ \varphi', \e^{\frt_0  \widetilde Q^\ka(p)} \varphi \rangle
\eeq
where $\varphi(k)= \e^{\i p \frac{\eta}{2}}  \e^{\i k(2 v +\eta) }$ (and $\varphi'$ similarly) we conclude
\beq
\norm P_{\ve}    \e^{\frt_0 \hat Q(p)}P_{\ve} \norm_{\scrG}  \leq  C\sup_{|\ka| \leq \ga_0} \norm  \e^{\frt_0  \widetilde Q_\ve^\ka(p)}\norm 
\eeq
where the norm on the RHS is the operator norm on $\scrB(L^1(\tor)),  \scrB(L^1(\caF)) $, depending on $\ve$.   By similar reasoning, we get a bound on the $\norm\cdot \norm_\ga$ of the projection $R_Q(p)$. 

We now finally choose $\ga_0 := \min (\ka_0, \ga_Q)$ where the latter constants were introduced in Proposition \ref{prop: main spectral properties of Qvep}. That proposition can then be recast as follows. 

 \noindent (1) Let $p$ with $|\Re p|\leq\bana_Q$ and $|\Im p|<\ga_0$. Then
 \beq
 T_Q(p)=e^{\frt_0 f_Q(p)}R_Q(p)+(1-R_Q(p))T_Q(p) .
\eeq
with
\beq\label{isolated}
\norm(1-R_Q(p))T_Q(p)  \norm_{\scrG}  \leq  Ce^{-\frac{3a_Q}{4}\frt_0},
\eeq
and $f_Q(p)$ and $R_Q(p)$ are as in  Proposition
\ref{prop: main spectral properties of Qvep} with
\beq\label{isolated1}
|e^{\frt f_Q(p)}|   \geq  ce^{-\frac{a_Q}{2}\frt_0}.
\eeq
 \noindent (2) In the region  $|\Re p|\geq\bana_Q$ and $|\Im p|<\ga_0$ 
\beq\label{largeQbound}
\norm T_Q(p)  \norm_{\scrG}  \leq  Ce^{-\frac{b_Q}{2}\frt_0}.
\eeq
We can now apply spectral perturbation  theory, see e.g.\ Lemma \ref{lem: perturbation banach}. 
First, for $\frt_0$ sufficiently large (compared to the constants in the above statements), the eigenvalue $e^{\frt_0 f_Q(p)}$ is isolated, by \eqref{isolated} and  \eqref{isolated1}.
Thus,  taking then  $\la$ small enough,  $\str\la\str\leq \la(\frt_0)$ the following holds.

\noindent (a) The isolated eigenvalue persists and
the new eigenvalue $e^{ f_0(p)}$ has the same properties: ${ f_0(0)}=0$ by unitarity,
$\nabla{ f_0(0)}=0$  and $
|f_0(p)+D_0p^2|\leq C|p|^4 
$ by lattice symmetries, with $D_0\to \frt_0 D_Q$ as $\la\to 0$.

\noindent (b) The bounds \eqref{isolated} and \eqref{largeQbound}
hold for  $\e^{ f_0(p)}$ and $T_0$ (and different $C$ and $c$). Now take $\tau_0$  large enough and then  Proposition \ref{ass:  properties of T}  (except for \eqref{eq: projector r explicit bound} and \eqref{eq: analytic bounds t in prop} for $\hat T_n(0,p)$) holds for $n=0$ and $\frt_0\in [\tau_0,2\tau_0]$.

To get the claims for $R_0(0)$, i.e.\ \eqref{eq: projector r explicit bound}, note that \eqref{eq: closeness free and interacting equilibrium state} tells us that $\mu^{\be}$  is $\la^2$-close to the $\la=0$ system equilibrium state, i.e.\ the Gibbs state $\sim \e^{-\be H_{\spin}}$. The
latter is identical to the state $\mu_Q(0)$  in case $\be_1=\be_2=\be$, and hence one has $\norm P^{\be}- R_Q(0) \normba =\caO(\la^2)$.   Since $\norm R_Q(0) - R_{0}(0)\norm = \caO(\str \la\str^{2\al} )$,  \eqref{eq: projector r explicit bound}
holds for $n=0$.  Finally, the bound \eqref{eq: analytic bounds t in prop} for $\hat T_0(0,p)$ was proven in Section \ref{sec: bounds on boundary correlation functions}.

\appendix
\renewcommand{\theequation}{\Alph{section}\arabic{equation}}
\setcounter{equation}{0}

\section{Appendix: Spectral perturbation theory}\label{app: spectral perturbation theory}

In the lemma below, we gather some  spectral perturbation theory that is used in the proof of Lemma \ref{lem: flow of parameters for t}.
We consider operators on a  Banach space $\scrA$. Denote by $\norm \cdot \norm$ the usual operator norm on $\scrB(\scrA)$ and let $\norm \cdot \normdia$ by a norm on a dense subset of $\scrB(\scrA)$ such that 
\[ \norm  A  \norm  \leq   \norm  A  \normdia   \]
\[ \norm  A B  \normdia   \leq   \norm  A  \normdia   \norm  B  \normdia \]

% and we use a norm (denoted simply by $\norm \cdot \norm$) on (a subspace of) $\scrB(\scrA)$ that is different from the operator norm $\norm A \norm_{\scrB(\scrA)} = \sup_{\psi \in \scrA: \norm  \psi \norm=1} \norm A \psi \norm$. 
%Because of this twist, the theorem below is not literally a standard fact of perturbation theory, even though its proof is. 
%We assume that 

The following lemma  could just as well (and more naturally) be stated for $\norm \cdot \norm$ as for $\norm \cdot \normdia$. The use of $\norm \cdot \normdia$ is dictated by our application, where $\scrA= l^{\infty}(\caS)$ and $\norm \cdot \normdia= \norm \cdot \norm_\ga$.  In fact, in Lemma \ref{lem: flow of parameters for t}, we did not mention the space $\scrA= l^{\infty}(\caS)$ at all, but note that without such an underlying space, it is not a-priori clear what one means by expressions like 'a simple eigenvalue'. 

\begin{lemma} \label{lem: perturbation banach}
Let the operator $A_0 \in \scrB(\scrA)$, with $\norm A_0 \normdia <\infty$, have a simple eigenvalue $a_0$ with corresponding one-dimensional spectral projection $P_0$ such that 
\beq
 \norm A_0 - a_0 P_0 \normdia < \str a_0 \str, \qquad  \norm P_0 \normdia< \infty    \label{eq: isolated and bound}
\eeq
In particular, $a_0$ is an isolated point in the spectrum of $A_0$.
We study $A =A_0 + A_{1}$ for some perturbation $A_{1} \in \scrB(\scrA)$. Set
\beq
b(r):=  \frac{\norm P_0 \normdia}{ r} +    \frac{1}{ \str a_0 \str - r -  \norm A_0 - a_0 P_0  \normdia}, \qquad \textrm{for}\,\,  r <  \str a_0 \str -   \norm A_0 - a_0 P_0 \normdia
\eeq
If, for some $r>0$, we have 
\beq
\norm A_{1}\normdia b(r)  <1 
\eeq
then the eigenvalue persists as an isolated point in the spectrum when the perturbation is added (we call this eigenvalue $a$) and 
\beq
\str a-a_0 \str \leq r 
\eeq
The spectral projection $P$ corresponding to $a$ satisfies
\beq
      \norm P-P_0 \normdia \leq  2 \pi r       b(r)        \frac{ \norm    A_{1} \normdia b(r)  }{1 - \norm    A_{1}  \normdia b(r)  }  
\eeq
\end{lemma}

\begin{proof}
The condition \eqref{eq: isolated and bound} implies that the eigenvalue $a_0$ is an isolated point in the spectrum.  Let $D_r = \{ z \in \bbC, \str z-a_0 \str \leq r\}$  and choose $r$  such that  $A_0$ has no other spectrum than $a_0$ on $D_r$.   To prove that the eigenvalue persists in the interior of $D_r$, it suffices to show that  $\sup_{z \in  \partial D_r}\norm (z-(A_0 + k A_1))^{-1} \normdia$ remains finite for $k \in [0,1]$.  Without loss, we can assume $k=1$ in the proof below.
Consider the Neumann series
\beq
(z- A)^{-1} =  (z- A_0)^{-1}     \sum_{n \geq 0}      \left( A_{1}  (z- A_0)^{-1} \right)^n
\eeq
To establish that the resolvent remains finite, we show a stronger property (because of the different norms), namely
\beq   \label{eq: resolvent on circle}
 \sup_{z  \in \partial D_r}  \sum_{n \geq 0}   ( \norm  A_{1} \normdia \,       \norm   (z- A_0)^{-1}  \normdia)^n  < \infty  
\eeq
To get this, we bound
\beq
  \sup_{z  \in \partial D_r}     \norm   (z- A_0)^{-1}  \normdia  \leq   \sup_{z  \in \partial D_r}  \frac{\norm P_0 \normdia}{ \str z-a_0 \str} +   \sup_{z  \in \partial D_r}  \norm (z - (A_0 -a_0 P_0))^{-1} \normdia   \leq    b(r)
\eeq
where the first inequality  follows because $(z- A_0)^{-1}= P_0(z- a_0)^{-1} + (1-P_0)(z- A_0)^{-1}$ and the
last inequality follows from the Neumann series. 
Therefore, \eqref{eq: resolvent on circle} is finite if  $    \norm  A_{1} \normdia  b(r) < 1$ and we conclude that the eigenvalue $a$ lies in the interior of $D_r$.
The bound for the projection follows as (the integral over $\partial D_r$ is performed clockwise)
\baq
P- P_0  &=&  \frac{1}{2\pi \i } \int_{\partial D_r} \d z    \, \left( (z- A)^{-1} -(z- A_0)^{-1}  \right)  \\ 
 &=&   \frac{1}{2\pi \i } \int_{ \partial D_r} \d z    \,    (z- A_0)^{-1}    \sum_{n \geq 1}      (A_{1}  (z- A_0)^{-1})^n \\ 
\norm  P- P_0 \normdia & <&  \frac{1}{2\pi  }\int_{ \partial D_r} \d z   \,      \norm (z- A_0)^{-1}  \norm        \frac{ \norm    A_{1} \normdia \norm  (z- A_0)^{-1}  \normdia }{1 - \norm    A_{1}  \normdia \norm  (z- A_0)^{-1}  \normdia  }     \\ 
 &< &   r       b(r)        \frac{ \norm    A_{1} \normdia b(r)  }{1 - \norm    A_{1}  \normdia b(r)  }  
\eaq

\end{proof}

Let us denote by $c,C$ constants that depend only on $\norm P_0\normdia$.
If we assume that 
\beq \label{eq: arbirary bound spectral}
  \norm A_{1} \normdia \leq  c( \str a_0 \str -   \norm A_0 - a_0 P_0  \normdia),
\eeq
then we can choose $r = C \norm A_{1} \normdia  $ to estimate the spectral shift.    To bound the difference of projections, we choose $r = c ( \str a_0 \str -   \norm A_0 - a_0 P_0  \normdia)$, such that $r b(r) < C$, and we obtain
\beq
\norm P- P_0 \normdia  \leq  C  \frac{\norm A_{1} \normdia }{\str a_0 \str -   \norm A_0 - a_0 P_0  \normdia}
\eeq

\section{Appendix: Combinatorics}\label{app: combinatorics}

We collect some easy combinatorial lemma's.
\begin{lemma} \label{lem: distance factors}
Let $\scrT$ be a spanning tree on the finite set $A \subset \bbN_0$, and define 
\beq
\dist(\scrT)  : =    \prod_{\{ \tau,\tau'\} \in \scrE(\scrT)}    (1+\str \tau'-\tau \str)
\eeq
where $\scrE(\scrT)$ is the set of edges of  $\scrT$. 
Then
\beq
\dist(A)   \leq    \dist(\scrT)
\eeq
with $\dist(A)$ defined as in Section \ref{sec: convergence to fixed point}.
\end{lemma}
See e.g.\ \cite{deroeckkupiainen} for the 5-line proof.  The following result lies at the heart of cluster expansions. 
For finite subsets $S \subset \bbN_0$, write $S \sim S'$ whenever $S \cap S' \neq \emptyset$. Let $\caS$ be a collection of finite subsets $\caS=\{ S_1, \ldots, S_m\}$. 
Note that the connectedness of the graph with vertices $S_j$ and edges $\{S_i,S_j\}$ whenever $S_j \sim S_j$ is equivalent to the connectedness of $\caS$ defined in Section \ref{sec: nonlinear rg flow} and, as in that section, we write $\frC$ for the set of connected collections. 
\begin{proposition}\label{app: prop: cluster expansion}
Let $w(\cdot)$ be a  function on finite subsets $S \subset \bbN_0$.  Assume that for some $\ka>0$,
\beq\label{eq: kotecky preiss}
\sum_{S:  S \sim S' }  \e^{\ka \str S \str}    \str w(S)\str  \leq    \ka \str S' \str.
\eeq
Then, with $w(\caS) =  \prod_{S \in \caS} w(S)$, 
\beq
 \mathop{\sum}_{\caS}  \indicator_{[  \{S'\} \cup \caS \in \frC]}     \str w(\caS)\str  \leq    \e^{\ka \str S' \str}
\eeq
\end{proposition}
We refer to  \cite{ueltschi} for the proof.  The condition \eqref{eq: kotecky preiss} is often called the Kotecky-Preiss condition. 
An immediate consequence is 
\beq
\sum_{\caS \in \frC:  \supp\caS \ni \tau}   \str w(\caS)\str   \leq \sum_{S': S' \ni \tau}  \str w(S') \str     \mathop{\sum}_{\caS}  \indicator_{[  \{S'\} \cup \caS \in \frC]}    \str w(\caS)\str   \leq      \sum_{S': S' \ni \tau}  \str w(S') \str   \e^{\ka \str S' \str}  \leq \ka
\label{eq: useful cluster expansion}
\eeq

\bibliographystyle{plain}
\bibliography{mylibrary11}

 \end{document}